\documentclass[aps,longbibliography,floatfix,reprint,superscriptaddress,nofootinbib,prx]{revtex4-2}
\usepackage[utf8]{inputenc}
\usepackage{amsmath}
\usepackage{amsfonts,dsfont}
\usepackage{amssymb,amsthm}
\usepackage{graphicx}
\usepackage[]{hyperref}
\usepackage{algorithm}
\usepackage{algpseudocode}
\usepackage{algorithmicx}

\hypersetup{
  colorlinks = true,
  urlcolor = Blue,
  linkcolor = linkcol,
  citecolor = cyan
}
\usepackage[all]{hypcap}
\usepackage{url}
\newcommand{\ket}[1]{|#1\rangle}
\newcommand{\bra}[1]{\langle#1|}

\usepackage[caption=false]{subfig}
\usepackage[dvipsnames]{xcolor}
\usepackage{tikz}
\usetikzlibrary{quantikz2}
\usepackage{physics}

\definecolor{linkcol}{rgb}{0.0,0.5,0.65}

\usepackage{dsfont}  
\usepackage[cal=cm,scr=rsfs,bb=ams]{mathalpha} 
\usepackage{parskip}

\usepackage[LGR,T1]{fontenc} 
\usepackage{mathpazo}
\usepackage{bm}

\usepackage[capitalise,compress]{cleveref}
\crefname{section}{Section}{Section} 
\crefname{subsection}{Section}{Section}

\newtheoremstyle{mystyle}%
{4pt}
{0pt}
{\itshape}
{}
{\bfseries}
{.}
{.2em}
{}

\theoremstyle{mystyle}
\crefname{thm}{Theorem}{Theorems}
\Crefname{thm}{Theorem}{Theorems}
\newtheorem{thm}{Theorem}

\newtheorem{obs}{Observation}
\crefname{obs}{Observation}{Observations}
\Crefname{obs}{Observation}{Observations}
\newtheorem{lemma}[thm]{Lemma}
\crefname{lemma}{Lemma}{Lemmas}
\Crefname{lemma}{Lemma}{Lemmas}

\crefname{corll}{Corollary}{Corollaries}
\newtheorem{claim}{Claim}
\newtheorem{defn}{Definition}
\theoremstyle{remark}

\AddToHook{env/lemma/begin}{\crefalias{thm}{lemma}} 

\makeatletter
\newenvironment{subtheorem}[1]{%
  \def\subtheoremcounter{#1}%
  \refstepcounter{#1}%
  \protected@edef\theparentnumber{\csname the#1\endcsname}%
  \setcounter{parentnumber}{\value{#1}}%
  \setcounter{#1}{0}%
  \expandafter\def\csname the#1\endcsname{\theparentnumber.\Alph{#1}}%
  \ignorespaces
}{%
  \setcounter{\subtheoremcounter}{\value{parentnumber}}%
  \ignorespacesafterend
}
\makeatother

\newcounter{parentnumber}

\usepackage[normalem]{ulem}

\newcommand{\btheta}{\bm{\theta}}

\newcommand{\supp}{\mathrm{supp}}

\DeclareMathOperator*{\argmin}{arg\,min}
\DeclareMathOperator*{\Var}{Var}
\DeclareMathOperator*{\poly}{poly}

\makeatletter
\newcommand{\globalcolor}[1]{%
  \color{#1}\global\let\default@color\current@color
}
\makeatother

\begin{document}


\title{Dynamic parameterized quantum circuits: expressive and barren-plateau free}
\author{Abhinav Deshpande}
\affiliation{IBM Quantum, Almaden Research Center, San Jose, CA 95120, USA}
\author{Marcel Hinsche}
\affiliation{IBM Quantum, IBM Research Europe – Zurich}
\affiliation{Dahlem Center for Complex Quantum Systems, Freie Universit\"at Berlin, Germany}
\author{Khadijeh Najafi}
\affiliation{IBM Quantum, IBM T.J. Watson Research Center, Yorktown Heights, NY 10598, USA}
\affiliation{MIT-IBM Watson AI Lab, Cambridge, MA 02142, USA}
\author{Kunal Sharma}
\affiliation{IBM Quantum, IBM T.J. Watson Research Center, Yorktown Heights, NY 10598, USA}
\author{Ryan Sweke}
\affiliation{IBM Quantum, Almaden Research Center, San Jose, CA 95120, USA}
\author{Christa Zoufal}
\affiliation{IBM Quantum, IBM Research Europe – Zurich}


\begin{abstract}
Classical optimization of parameterized quantum circuits is a widely studied methodology for the preparation of complex quantum states, as well as the solution of machine learning and optimization problems. However, it is well known that many proposed parameterized quantum circuit architectures suffer from drawbacks which limit their utility, such as their classical simulability or the hardness of optimization due to a problem known as ``barren plateaus''.
We propose and study a class of \textit{dynamic} parameterized quantum circuit architectures. These are parameterized circuits containing intermediate measurements and feedforward operations. In particular, we show that these architectures:
\begin{enumerate}
\item Provably do not suffer from barren plateaus.
\item Are expressive enough to describe arbitrarily deep unitary quantum circuits.
\item Are competitive with state of the art methods for preparing ground states and facilitating the representation of nontrivial thermal states.
\end{enumerate}
These features make the proposed architectures promising candidates for a variety of applications.

\end{abstract}
\maketitle

\section{Introduction}\label{section:introduction}

Variational quantum algorithms (VQAs) are a broad class of highly studied quantum algorithms for solving a diverse range of problems~\cite{cerezo2021costfunctiondependentbarren}. Specifically, there are now VQA-based approaches for machine learning with classical data~\cite{benedetti2019parameterizedquantumcircuitsmachine}, classical optimization~\cite{zhou2020quantumapproximateoptimizationalgorithm}, learning quantum systems~\cite{caro2023outofdistributiongeneralizationlearningquantum} and the preparation of ground and thermal states of complex quantum systems~\cite{tilly2022variationalquantumeigensolverreview,cerezo2021costfunctiondependentbarren}. The basic idea of all variational quantum algorithms is as follows: one defines a class of parameterized quantum circuits, chooses an initial circuit structure from this class in some way, and then iteratively updates the circuit parameters or circuit structure using a classical optimization algorithm.
This iterative adjustment seeks to minimize a loss function, which is designed to evaluate the quality of a specific quantum circuit as a solution to the problem at hand.

Given the above, the first thing one needs to do when designing a variational quantum algorithm is choose an appropriate parameterized quantum circuit (PQC). At a high level, any good PQC for a specific problem should ideally satisfy all of the following criteria:

\begin{enumerate}
\item \textbf{Expressivity:} There should exist circuit parameter instances that correspond to good solutions to the problem of interest.
\item \textbf{Trainability:} One should be able to find the circuit instances corresponding to good solutions, via the chosen optimization method, using a polynomial amount of resources such as runtime.
\item \textbf{Classical hardness:} There should not exist a classical algorithm that can efficiently simulate the PQC, and hence, the variational algorithm.
\end{enumerate}
While a large variety of different PQC architectures have been proposed, there are unfortunately not many candidates which might satisfy all three criteria. In fact, it is not even clear how to rigorously formalize each criterion~\cite{gil-fuster2024relationtrainabilitydequantizationvariational}, or whether it is indeed possible to satisfy all three simultaneously for meaningful problems.
More specifically, for a large variety of architectures there is a known tradeoff between expressivity and trainability. In particular, one can show that expressivity often leads to ``barren plateaus'', which are an obstacle to trainability via gradient-based optimization algorithms~\cite{larocca2025barrenplateausvariationalquantuma}.
Often, this tradeoff also leads to another impediment.
In particular, one can also show that for a wide variety of architectures, decreasing expressivity sufficiently to mitigate barren plateaus can often lead to the existence of efficient algorithms for classical simulations~\cite{cerezo2023doesprovableabsencebarren}.

With this in mind, in this work we study \textit{dynamic} parameterized quantum circuits and argue that they help tackle the tradeoff between expressivity and trainability.
In doing so, we are free to restore expressivity, making these circuits classically hard to simulate in the worst case.
Specifically, we study parameterized quantum circuits that include nonunitary measurement and feedforward operations, which are known to provide a significant resource for quantum error correction~\cite{aharonov1996limitationsnoisyreversiblecomputation}, measurement-based quantum computing~\cite{briegel2009measurementbasedquantumcomputation}, and the preparation of interesting states~\cite{malz2024preparationmatrixproductstates,baumer2024efficientlongrangeentanglementusing}.
Indeed, the utility of dynamic circuit operations in quantum computing has already stimulated previous proposals of nonunitary parameterized quantum circuit architectures~\cite{cong2019quantumconvolutionalneuralnetworks,bondarenko2020quantumautoencodersdenoisequantum,beer2021dissipativequantumgenerativeadversarial,poland2020nofreelunchquantum,ilin2024dissipativevariationalquantumalgorithms,yan2024variationalloccassistedquantumcircuits,mele2024noiseinducedshallowcircuitsabsence}. Our contribution in this work is to provide a unifying framework for dynamic PQC architectures, and to analyze both analytically and numerically their potential for variational quantum algorithms, with respect to the criteria discussed above.

\subsection{Structure of this work}\label{ss:structure-of-this-work}

This work is structured as follows. We first give an overview of our contributions in \cref{ss:our_results}.
Next, in Section~\ref{section:VQA_overview} we give a high-level overview of variational quantum algorithms, in order to provide both context and notation. Readers who are familiar with variational quantum algorithms could safely skip this section. We then proceed in Section~\ref{section:ansatz} to introduce the dynamic parameterized circuit architectures we study here in more detail, and to discuss their relation to existing nonunitary parameterized quantum circuit proposals. With this established, we study the trainability of DPQC architectures in Section~\ref{section:trainability}. In particular, we begin with a discussion of what ``trainability'' actually means, and how this notion is related to barren plateaus. We then state our main analytical result, which provides sufficient conditions for the absence of barren plateaus in DPQC architectures. We then introduce the ingredients for the proof, namely a statistical mechanics model, and show how it can be used to derive existing barren plateau results.  Given these theoretical foundations, we then provide the aforementioned numerical experiments for both ground and thermal state preparation in Section~\ref{section:utility}. Finally, we discuss in Section~\ref{section:classical_hardness} the classical simulability of DPQC architectures.

\subsection{Our results and contributions}\label{ss:our_results}
\subsubsection{Dynamic parameterized quantum circuits}\label{sss:intro-dynamical-parameterized-quantum circuits}

\begin{figure}
\begin{center}
\resizebox{\columnwidth}{!}{\begin{quantikz}[row sep={1cm,between origins},wire types = {q,q,q,q,q,q}]
\lstick{$\ket{0}$} & \gate[2]{U(\btheta)}  & \gate[style={fill=blue!20}]{\mathcal{F}(\btheta)} \setwiretype{q} & \gate[2]{U(\btheta)} &  & \gate[2]{U(\btheta)}   & & \gate[2]{U(\btheta)}&\ground{}  \\
\lstick{$\ket{0}$} & & \gate[2]{U(\btheta)}  &  & \gate[2]{U(\btheta)}  &     & \gate[style={fill=blue!20}]{\mathcal{F}(\btheta)}  & & \meter{} \\
\lstick{$\ket{0}$} & \gate[2]{U(\btheta)}  & & \gate[style={fill=blue!20}]{\mathcal{F}(\btheta)}   & & \gate[2]{U(\btheta)}  &   &\gate[2]{U(\btheta)} &\ground{}   \\
\lstick{$\ket{0}$} & & \gate[2]{U(\btheta)} & \gate[style={fill=blue!20}]{\mathcal{F}(\btheta)}  &  \gate[2]{U(\btheta)} &   & \gate[style={fill=blue!20}]{\mathcal{F}(\btheta)} && \meter{} \\
\lstick{$\ket{0}$} & \gate[2]{U(\btheta)}  &  & \gate[2]{U(\btheta)} &  & \gate[2]{U(\btheta)} & &\gate[2]{U(\btheta)} & \ground{} \\
\lstick{$\ket{0}$} & & &  & &    & \gate[style={fill=blue!20}]{\mathcal{F}(\btheta)} & &\meter{}
\end{quantikz}}
\end{center}
\vspace{1em}
\begin{center}
\begin{quantikz}[row sep={1cm,between origins}]
\qw & \gate[style={fill=blue!20}]{\mathcal{F}(\btheta)} & \qw
\end{quantikz}
$:=$
\begin{quantikz}[row sep={1cm,between origins},transparent]
\lstick{$\ket{0}$} & \gate{R_X(\btheta)} & \ctrl[vertical wire=c]{1} & \ground{} \\
\qw & & \gate{\mathcal{F}} & \qw
\end{quantikz} \begin{quantikz}[row sep={1cm,between origins},transparent]
\qw & \gate{\mathcal{F}} & \qw
\end{quantikz}
$:=$
\begin{quantikz}[row sep={1cm,between origins},transparent]
\qw & \meter{} \gategroup[1,steps=2,style={dashed,rounded
corners},background]
\qw & \gate{\text{If 1, } U_1} \setwiretype{c} & \setwiretype{q}
\end{quantikz}
\end{center}
\caption{An illustration of the dynamic parameterized quantum circuit (DPQC) architectures that we consider in this work. These circuits consist of parameterized two-qubit unitary gates~$U(\btheta)$, as well as parameterized nonunitary single-qubit dynamic operations, which are denoted as $\mathcal{F}(\btheta)$-gates. Each such $\mathcal{F}(\btheta)$ operation is a probabilistic implementation of a feedforward operation $\mathcal{F}$.
}\label{fig:dpc-intro}
\end{figure}

First, we describe the dynamic parameterized quantum circuit (DPQC) architectures that we consider as variational ans\"atze in this work. They consist of layers of parameterized two-qubit gates $U(\btheta)$ interspersed with parameterized dynamic operations $\mathcal{F}(\btheta)$ (see also \cref{fig:dpc-intro}). We will denote the full parameterized dynamic circuit as a channel $\mathcal{C}(\bm{\theta})$ with parameters $\btheta$. Note that we generally work with operations that only depend on one or a few components of $\btheta$ though they are denoted as functions of the entire vector $\btheta$.
While dynamic operations can in principle be applied across many qubits, in this work, we choose to focus on single-qubit dynamic operations. More specifically, as illustrated in \cref{fig:dpc-intro}, here each parameterized dynamic operation $\mathcal{F}(\btheta_i)$ is a probabilistic implementation  of a \textit{feedforward} operation $\mathcal{F}$ with probability $\sin^2(\btheta_i/2)$.
That is,
\begin{align}
 \mathcal{F}(\btheta_i)(\cdot) = \cos^2(\btheta_i/2) (\cdot) + \sin^2(\btheta_i/2) \mathcal{F}(\cdot).
\end{align}
The feedforward operations themselves consist of a measurement on the respective qubit followed by a conditional gate: if the measurement outcome was 0, then apply $U_0=I$, if instead it was 1, apply $U_1$ given by
\begin{equation}
U_1 =
\begin{pmatrix}
\cos {\varphi}e^{-i\phi} & -i \sin \varphi \\
- i \sin \varphi & \cos{\varphi} e^{i \phi}
\end{pmatrix} \,.
\end{equation}
We note that the ancilla qubits that control the probability of implementing $\mathcal{F}$ operations are unentangled with the rest of the circuit, and can be simulated classically---i.e.\,one does not need an additional physical qubit for each $\mathcal{F}(\btheta)$ operation. At the end of the circuit, an observable supported on some subset of the qubits is measured in order to calculate the value of some loss function. We call the qubits on which this observable is supported \textit{system} qubits, and refer to the remaining qubits as ancilla qubits. We stress that the position of the dynamic operations in the circuit, the nature of the parameterized and conditional operations within the dynamic operations, and the number and allocation of system and ancilla qubits, are all design choices that one can make freely. Finally, given a parameterized dynamic quantum circuit $\mathcal{C}(\btheta)$, at initialization we draw the parameters of the two-qubit gates $U(\bm{\theta})$ from a locally scrambling ensemble, which refers to an ensemble that is invariant under conjugation via single-qubit unitaries drawn from a unitary 2-design.

\subsubsection{Expressivity}\label{sss:intro-expressivity}

Having defined the parameterized dynamic circuit architectures that we study in this work, we start with two simple observations concerning the expressivity of these architectures.

\begin{obs}[Expressivity of DPQC architectures with probabilistic feedforward---informal]\label{obs:expressivity_informal_probabilistic} Note that $\mathcal{F}(\btheta_i = 0)$ is the identity channel. Therefore, by setting all the circuit parameters that control the probability of implementing an $\mathcal{F}$ gate to 0, one obtains a purely unitary ansatz. With this in mind, given a DPQC architecture $\mathcal{C}(\btheta)$ with connectivity graph $\mathsf{G}$, let  $d$ be the depth of the architecture with all feedforward operations removed. This architecture can realize \emph{all} unitary operations of depth $d$ on $\mathsf{G}$.
\end{obs}

Said another way, if one starts from a unitary parameterized quantum circuit and then adds parameterized feedforward operations, the resulting DPQC architecture is at least as expressive as the unitary architecture from which one started. This observation is helpful, as it ensures that one can in principle prepare interesting pure states using such an architecture. In the observation below, we highlight that even DPQC architectures with only \emph{deterministic} feedforward operations (i.e.\,$\mathcal{F}(\theta) = \mathcal{F}(\pi)$) can realize nontrivial pure states.

\begin{obs}[Expressivity of DPQC architectures with deterministic feedforward---informal]\label{obs:expressivity_informal_deterministic} Consider any depth-$d$ DPQC architecture $\mathcal{C}(\btheta)$, in which all feedforward operations are deterministic (i.e. happen with probability 1) and  only occur on ancilla qubits (qubits on which the loss function has no support). Denote the connectivity subgraph of the system qubits in $\mathcal{C}(\btheta)$ as $\mathsf{G}$. Then, for every unitary circuit $U_{(d,\mathsf G)}$ on $\mathsf{G}$ of depth $d$, there exists a setting of the parameters $\btheta$ such that $U_{(d,\mathsf G)}$ is implemented by $\mathcal{C}(\btheta)$.
\end{obs}

The proof of the above observation is straightforward: simply observe that one can ``disconnect'' all ancilla qubits on which a dynamic operation occurs by setting any parameterized gate between a system qubit and such a qubit to be non-entangling. One is then left precisely with a depth-$d$ parameterized circuit with connectivity graph $\mathsf{G}$. Said another way, if one starts from a unitary parameterized quantum circuit of a certain depth and connectivity, and then adds ancilla qubits and feedforward operations on these ancilla qubits, the resulting parameterized dynamic quantum circuit is at least as expressive as the original unitary parameterized quantum circuit.
We use parameterized dynamic circuit architectures with precisely such a structure for our ground state preparation experiments (See \cref{sss:intro-numerical-results} and \cref{section:utility}).
In these experiments, we indeed observe convergence of our model to pure states. While we have not investigated whether this is the mechanism through which the model reaches pure states, the observation above again guarantees us that this is at least possible in principle.

\subsubsection{Absence of barren plateaus}\label{sss:intro-absence-of-BP}
One of the primary contributions of this work is to provide sufficient conditions for the absence of barren plateaus in parameterized dynamic quantum circuit architectures. In purely unitary random circuits, the onset of barren plateaus is closely linked to the size of the observable’s \textit{backward lightcone} \cite{letcher2024tightefficientgradientbounds,larocca2025barrenplateausvariationalquantuma}. Specifically, the loss function's variance decays exponentially with the lightcone's size due to the scrambling effect of the random unitary layers, which make the observable increasingly insensitive to individual parameter changes. We show that inserting feedforward operations $\mathcal{F}$ fundamentally alters this behavior: each feedforward operation counteracts the scrambling, providing a mechanism to prevent barren plateaus without necessarily sacrificing expressivity. To quantify this, we introduce the \textit{feedforward distance} $f$, which measures an observable's distance to the nearest feedforward operation $\mathcal{F}$ within its backward lightcone (see \Cref{fig:ff_distance} for a schematic explanation). We show a lower bound on the variance of the loss function that decays exponentially with feedforward distance $f$, rather than the size of the backward light cone.

\begin{figure}[t]
\includegraphics[width = \columnwidth]{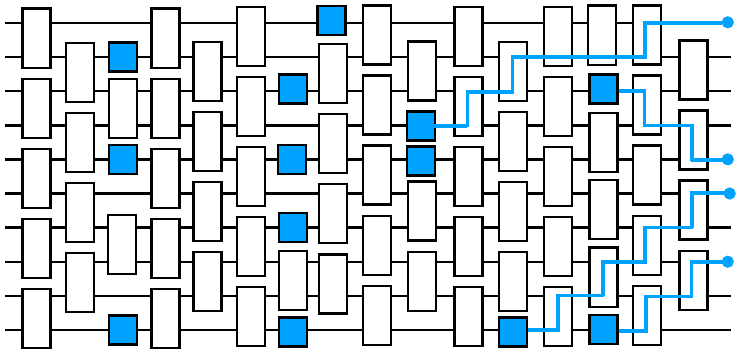}
\caption{An illustration of the shortest paths from a qubit measurement to a feedforward operation through the backwards light cone of the measurement. The \textit{feedforward distance} of an observable, with respect to a specific DPQC architecture, is the maximum length of such paths, over all qubits on which the observable is supported. Theorem~\ref{thm:BP_DPQC_informal} provides an upper bound on the variance of a local observable in terms of the feedforward distance.}\label{fig:ff_distance}
\end{figure}

\begin{thm}[Absence of barren plateaus in DPQCs for $k$-local Hamiltonians---informal]\label{thm:BP_DPQC_informal}
Let $\rho(\bm{\theta})=\mathcal{C}(\bm{\theta})\left(\ketbra{0^n}\right)$ be the output state of the parameterized circuit ensemble introduced above and let $H $ be a $k$-local Hamiltonian.
Then, the variance of the loss function $L = \Tr \rho(\bm{\theta}) H $ is lower bounded as
\begin{align}
\Var_{\btheta }{L} \geq \left( \frac{\alpha}{5} \right)^{k(f +1)} \cdot  \norm{H}^2_{HS},
\end{align}
where $\norm{H}_{HS} := \sqrt{\Tr H^2} $ is the Hilbert-Schmidt norm of $H$ and $\alpha$ is a constant that depends on the entangling power of the ensemble of two-qubit gates\footnote{For Haar-random two-qubit gates, we have $\alpha = 1$.}.
\end{thm}
We note that for $k,f=O(1)$, the variance is non-vanishing. That is, we prove absence of barren plateaus for local observables under the condition that the feedforward distance $f$ is constant. We emphasize that using this result, together with Observations~\ref{obs:expressivity_informal_probabilistic} and~\ref{obs:expressivity_informal_deterministic}, one can construct DPQC architectures which are both highly expressive and barren plateau free.

Additionally, we also show that the absence of barren plateaus in DPQC architectures is robust to noise present after every gate, including nonunital noise.
Specifically, we show that when there is a single-qubit noise channel after every operation in the circuit with an average infidelity of $\gamma/2$ and nonunitality\footnote{For a definition of these quantities, refer to \cref{eq_nonunital} in Appendix~\ref{app:SM}.} $\delta \leq \gamma \leq 1/2$, a similar lower bound on the variance holds.
\begin{thm}[Noise robustness of \cref{thm:BP_DPQC_informal}---informal]\label{thm:robustness-informal}
Let the output state of the noisy circuit be $\tilde{\rho}(\btheta) = \tilde{\mathcal{C}}(\btheta) \left( \ketbra{0^n} \right)$.
The variance of the loss function $L = \Tr \tilde{\rho}(\btheta) H$ for a $k$-local Hamiltonian is given by
\begin{align}
\Var_{\btheta }{L} \geq \left( \frac{\alpha'}{5} \right)^{k(f +1)} \cdot  \norm{H}^2_{HS},
\end{align}
where $\alpha' = \alpha (1-\gamma-\delta) + 5\delta$.
\end{thm}
Again, for $k,f=O(1)$ and $\gamma<1/2$, there is a nonvanishing variance of the loss function.
This illustrates that the nonunitary feedforward operations help fight against noise-induced barren plateaus \cite{wang2021noiseinducedbarrenplateausvariational}.

On the technical side, Theorems~\ref{thm:BP_DPQC_informal} and~\ref{thm:robustness-informal} are proven via the so called ``stat-mech'' model~\cite{hunter-jones2019unitarydesignsstatisticalmechanics,dalzell2022randomquantumcircuitsanticoncentrate,dalzell2021randomquantumcircuitstransform,napp2022quantifyingbarrenplateauphenomenon,ware2023sharpphasetransitionlinear}, which allows one to compute second moments of statistical quantities of ensembles of random quantum circuits. In addition to using the stat-mech model to prove the above theorems, we also show how a variety of previous barren plateau results can be obtained via the stat-mech model, which may be of independent interest.

We stress that while a barren plateau result for a specific architecture indicates a significant obstacle to trainability~\cite{larocca2025barrenplateausvariationalquantuma}, the absence of a barren plateau result does \textit{not} immediately imply that an architecture is trainable in a meaningful sense~\cite{gil-fuster2024relationtrainabilitydequantizationvariational}. We discuss this issue at length in \cref{section:trainability}.

Finally, we summarize the above insights into both expressivity and absence of barren plateaus with the following observation.

\begin{obs}[DPQC architectures: connecting expressivity and absence of BPs]\label{obs:tuning_between_expressive_and_BP_free}
Taking together \cref{thm:BP_DPQC_informal} on absence of BPs, and \cref{obs:expressivity_informal_probabilistic,obs:expressivity_informal_deterministic} on expressivity, we note that DPQC architectures allow one to interpolate smoothly between highly expressive unitary architectures and BP-free nonunitary architectures with a constant feedforward depth. 
\end{obs}

\subsubsection{Numerical results}\label{sss:intro-numerical-results}

In order to explore the potential utility of DPQC architectures for practically relevant problems, we perform numerical experiments for both ground and thermal state preparation problems.

\textbf{Ground state preparation:} For ground state preparation 
we study the perturbed toric code Hamiltonian
\begin{align}
H_{\text{toric}} = (1-h)H_0-\sum\limits_{j=1}^nhZ_j
\end{align}
with open boundary conditions, as studied in Ref.~\cite{zhang2024absencebarrenplateausfinite}. The first term of the Hamiltonian corresponds to the unperturbed toric code $\textstyle{H_0 = -\sum\limits_vA_v-\sum\limits_pB_p}$,
where $v$ and $p$ runs over all vertices and plaquettes of a 2D square lattice, and $A_v$ and $B_p$ represent the standard vertex and plaquette operators of the toric code. 
As discussed in Ref.~\cite{zhang2024absencebarrenplateausfinite}, this system is of interest as a test case as the ground state contains long-range entanglement.
In Section~\ref{ss:ground_states} we provide the results of our experiments for a square lattice of 12 system qubits and 4 ancilla qubits for different values of the perturbation $h$. In particular, we find that parameterized dynamic quantum circuits perform competitively with state-of-the-art variational parameterized quantum circuit architectures, such as those based on finite local-depth parameterized quantum circuits~\cite{zhang2024absencebarrenplateausfinite}. We note that, despite what one might expect, optimization over DPQC architectures often leads to \textit{pure states}.

\textbf{Thermal state preparation:} We study both the transverse field Ising model
\begin{align}
H_{\text{TFI}} =  -\sum\limits_{j=1}^nX_jX_{j+1}-\frac{1}{2}\sum\limits_{j=1}^nZ_j,
\end{align}
and an XY model
\begin{align}
    H_{\text{XY}} = -\sum\limits_{j=1}^n\left[\frac{3}{4}X_jX_{j+1} + \frac{1}{4}Y_jY_{j+1}\right]-\frac{1}{2}\sum\limits_{j=1}^nZ_j,
\end{align}
with periodic boundary conditions, as has been studied in prior work on thermal state preparation with dissipative variational quantum algorithms~\cite{ilin2024dissipativevariationalquantumalgorithms}.
In Section~\ref{ss:thermal_states}, we show the results of numerical experiments aimed at preparing the thermal states of the above models at inverse temperature $\beta=2$, for a variety of circuit depths, and system sizes of up to 10 qubits.  These results provide evidence that DPQC architectures can indeed provide good approximations to the thermal states of interesting models.
We wish to highlight that the results corresponding to good approximations of the target were obtained by using the infidelity with respect to the target thermal state as a loss function, which is \textit{not} a scalable approach and is not covered by our results on the absence of BPs.
Our results provide evidence that DPQC architectures are capable of representing good approximations to the thermal states of interesting systems, but more work is required to understand the performance of such architectures with respect to other loss functions that can be measured efficiently on a quantum computer.

Finally, we note that for both the ground and thermal state preparation experiments above, the gradients were obtained by automatic-differentiation executed via classical simulation \cite{zhang2023tensorcircuitquantumsoftwareframework}. However, by considering the purification of the feedforward operations, one could also obtain a gradient estimator via a parameter shift rule which may be evaluated on a quantum computer, under some assumptions on the parameterized gates~\cite{wierichs2022generalparametershiftrulesquantuma}. Further details are provided in Sec.~\ref{sec:gradient}.

\subsubsection{Classical simulability}\label{sss:intro-classical-simulability}

Lastly, in order to understand the potential utility of any variational quantum algorithm, it is crucial to understand the extent to which this variational quantum algorithm can or cannot be efficiently classically simulated when used to solve relevant problems. For the DPQC architectures studied in this work, we find the following: 

\textbf{Worst-case hardness:} It follows as an immediate consequence of \cref{obs:expressivity_informal_probabilistic,obs:expressivity_informal_deterministic} that there are DPQC architectures that are both barren-plateau-free and worst-case hard to simulate classically. In particular, one can consider starting from any universal unitary circuit architecture, and then either introducing probabilistic feedforward operations, or deterministic feedforward operations on ancilla qubits, in such a way that the resulting DPQC architecture has constant feedforward depth. We stress that there are a variety of other strategies one could use for avoiding BPs, such as deep circuits with fixed or small-angle initializations, that would also lead to architectures which are BP-free and worst-case hard to simulate. However, intuitively, one expects randomness in initialization to be useful for optimization, and therefore for these alternative strategies to be disadvantageous from an optimization perspective.

\textbf{Average-case easiness:} For a wide class of DPQC circuit architectures $\mathcal{C}(\btheta)$, with high probability over the choice of $\btheta$ at initialization, the circuit $\mathcal{C}(\btheta)$ \textit{can} be efficiently classically simulated---in the sense that expectation values of Pauli observables can be efficiently classically estimated---via low-weight Pauli path propagation~\cite{angrisani2024classicallyestimatingobservablesnoiseless,mele2024noiseinducedshallowcircuitsabsence}.

\textbf{Average-case easiness does not rule out quantum utility:} Despite being average-case easy to simulate, DPQC architectures could still provide quantum utility if worst-case hard instances occur during training. Whether or not this is the case for DPQC architectures, when used to solve practically relevant problems, remains an open and interesting direction for future research.

We discuss all of these issues at length in Section~\ref{section:classical_hardness}.

\section{Variational quantum algorithms }\label{section:VQA_overview}
In this section, we provide an overview of \emph{variational quantum algorithms} (VQAs), in order to both set notation and provide a unifying framework for the results of this work. Readers familiar with variational quantum algorithms could skip this section and return to it when they encounter any notation which is defined here. The starting point for any variational quantum algorithm is a problem, defined by a loss function $L:\mathcal{X}\rightarrow\mathbb{R}$, for some set of objects $\mathcal{X}$. The solution to the problem $L$ is given by
\begin{equation}
x^* = \argmin_{x\in\mathcal{X}}L(x).
\end{equation}
Some concrete examples of problems we might be interested in are:
\begin{enumerate}
\item Ground state search: Let $\mathcal{S}_n$ represent the set of $n$-qubit quantum states. The ground state search problem for an $n$-qubit Hamiltonian $H$ is defined by the loss function $L_H:\mathcal{S}_n\rightarrow\mathbb{R}$, where $L_H(\rho) = \mathrm{Tr}(\rho H)$. The solution to the problem is given by the ground state of the Hamiltonian.
\item Thermal state preparation: For a given Hamiltonian $H$ and inverse temperature $\beta$, we aim to prepare the state that minimizes the free energy, i.e. the problem is defined by the loss function $L_{(H,\beta)}:\mathcal{S}_n\rightarrow\mathbb{R}$ where $L_{(H,\beta)}(\rho) = \text{Tr}\left[\rho H\right] + \frac{1}{\beta}\text{Tr}\left[\rho\log\rho\right]$. The solution to the problem is the Gibbs state $\rho_{\text{Gibbs}}=\frac{e^{-\beta H}}{\text{Tr}\left[e^{-\beta H}\right]}$. 
\item Distribution learning: Let $\mathfrak{D}_n$ represent the set of all discrete distributions over length-$n$ bit strings.
The distribution learning problem for a target distribution $\mathcal{D}\in\mathfrak{D}_n$ is defined by the loss function $L_{\mathcal{D}}:\mathfrak{D}_n\rightarrow \mathbb{R}$ where $L_{\mathcal{D}}(\mathcal{D}') = d(\mathcal{D}',\mathcal{D})$, where $d$ is some metric on distributions. The solution to the problem is the target distribution~$\mathcal{D}$.
\item Quantum process learning: Let $\mathfrak{T}_n$ represent the set of all $n$-qubit quantum channels, and let $d_{\diamond}$ be the distance induced by the diamond norm. The quantum process learning problem for a target channel $\mathcal{T}\in\mathfrak{T}_n$ is defined by the loss function $L_{\mathcal{T}}:\mathfrak{T}_n\rightarrow\mathbb{R}$ where $L_{\mathcal{T}}(\mathcal{T}',\mathcal{T}) = d_{\diamond}(\mathcal{T},\mathcal{T}')$. The solution to the problem is the target channel $\mathcal{T}$.
\end{enumerate}

\begin{algorithm}[H]
\caption{Variational quantum algorithm}\label{alg:variational_quantum_algorithm}
\begin{algorithmic}
\State \textbf{Given:} Loss function $L:\mathcal{X}\rightarrow\mathbb{R}$
\State 
\State \textbf{Choose:} PQC architecture $\mathfrak{C} = \{\mathcal{C}({\btheta})\,|\,\bm\theta\in\Theta\}$ 
\State \textbf{Choose:} Circuit-to-object map $M:\mathfrak{C} \rightarrow \mathcal{X}$ 
\State \textbf{Choose:} initialization strategy $\mathsf{IS}$ 
\State \textbf{Choose:} parameter update rule $\mathsf{PU}$
\State \textbf{Choose:} convergence criteria $\mathsf{CC}$
\State
\State Define $L_M:\Theta\rightarrow \mathbb{R}$ via $L_M(\btheta) = L\left(M\left(\mathcal{C}({\btheta})\right) \right)$
\State Use initialization strategy $\mathsf{IS}$ to set initial $\btheta_0\in\Theta$ \Comment Initialize parameters via $\mathsf{IS}$
\State $\mathrm{converged} \gets \mathrm{false}$
\State $i\gets 0$
\State
\While{not $\mathrm{converged}$}
\State Run circuit $\mathcal{C}({\btheta_i})$ (possibly multiple times) to evaluate $L_M(\btheta_i)$ \Comment{Evaluate loss function}

\State
\If{$L_M(\btheta_i)$ converged} \Comment{Evaluate convergence using $\mathsf{CC}$}
\State $\mathrm{converged}\gets \mathrm{true}$
\ElsIf{$L_M(\btheta_i)$ not converged}
\State Use parameter update rule $\mathsf{PU}$ to determine $\btheta_{i+1}\in\Theta$ \Comment{Update circuit parameters}
\State $i \gets i+1$
\EndIf
\EndWhile
\State
\State Return $\btheta_i$       \Comment{Output hypothesis circuit parameters}
\end{algorithmic}
\end{algorithm}

A standard approach to solving such problems is to perform optimization with respect to $L$, over some parameterized subset of $\mathcal{X}$, which we will refer to as a model class. Variational quantum algorithms follow this approach, using a {parameterized quantum circuit} (PQC) to define the model class. In this work, we will define a parameterized quantum circuit to be a circuit consisting of parameterized quantum channels.
We describe the set of circuits as $\mathfrak{C} = \{\mathcal{C}({\btheta})\,|\,\bm\theta\in\Theta\}$, where $\Theta$ represents the set of all possible channel parameters, and $\mathcal{C}({\btheta})$ represents the concrete circuit instance associated with parameters $\btheta$.
In this work, the class of circuits has both unitary gates and nonunitary channels, which may or may not be due to noise.

To define an appropriate model class from a parameterized quantum circuit $\mathfrak{C}$, we also require a ``circuit-to-object'' map $M:\mathfrak{C}\rightarrow \mathcal{X}$, which maps from quantum circuits to the set of objects $\mathcal{X}$ on which the loss function $L$ is defined. For the example problems given above, natural choices for the map $M$ would be:
\begin{enumerate}
\item Ground state search and thermal state preparation: $M:\mathfrak{C}\rightarrow \mathcal{S}_n$ via  $M\left(\mathcal{C}({\btheta}) \right) = \mathcal{C}({\btheta})\left(\ketbra{0^n}\right)$, i.e., $M\left(\mathcal{C}({\btheta}) \right)$ is the output state of the quantum circuit on the fixed input state $\ketbra{0^n}$.
\item Distribution learning: $M:\mathfrak{C}\rightarrow \mathfrak{D}_n$ via $M\left(\mathcal{C}({\btheta}) \right) = \mathcal{D}_{\btheta}$, where for all $x\in\{0,1\}^n$ one has $\mathcal{D}_{\btheta}(x) = \bra{x} \mathcal{C}({\btheta}) \left(\ketbra{0^n}\right)\ket{x}$---i.e., $M\left(\mathcal{C}({\btheta})\right)$ is the Born distribution associated with the output state of the quantum circuit on the fixed input state $\ketbra{0^n}$.
\item Quantum process learning: $M:\mathfrak{C}\rightarrow \mathfrak{T}_n$ via $M\left(\mathcal{C}({\btheta})\right) = \mathcal{C}({\btheta})$---i.e., $M\left(\mathcal{C}({\btheta})\right)$ is simply the channel associated with the circuit.
\end{enumerate}
Taking the parameterized quantum circuit together with the map $M$ we can then define the model class $\mathcal{M}_{(\mathfrak{C},M)}\subseteq\mathcal{X}$ via
\begin{equation}
\mathcal{M}_{(\mathfrak{C},M)} = \{M\left(\mathcal{C}({\btheta})\right)\,|\,\btheta\in\Theta\}.
\end{equation}
We would now like to find the optimal model $m^*\in\mathcal{M}_{(\mathfrak{C},M)}$ with respect to $L$. This is equivalent to identifying the optimal circuit parameters $\btheta\in\Theta$ with respect to the loss function $L_M:\Theta\rightarrow \mathbb{R}$ defined via $L_M(\btheta) = L\left(M\left(\mathcal{C}({\btheta})\right)\right)$. A natural way to do this is via Algorithm~\ref{alg:variational_quantum_algorithm}, which provides an abstract template for any {variational quantum algorithm}. We stress that Algorithm~\ref{alg:variational_quantum_algorithm} can be instantiated with a wide variety of different parameterized quantum circuit architectures, circuit-to-object maps~$M$, initialization strategies, parameter update rules and convergence strategies.  With this in hand, we proceed to introduce the class of {dynamic} parameterized quantum circuit architectures that we study in this work.

\section{Dynamic parameterized quantum circuit architectures}\label{section:ansatz}
In this section, we describe our proposed architectures and discuss their relation to existing nonunitary architectures.

\subsection{Dynamic quantum circuits}\label{ss:dynamic_circuits}
A \emph{dynamic} quantum circuit is one with intermediate measurements followed by one or more future operations conditional on the intermediate measurement results.
In principle, any intermediate measurement can always be deferred to the end of the circuit. However, the standard way to do so incurs a cost of extra ancilla lines, taking up quantum space, a precious resource on current devices. As we will discuss below, the ability to natively perform dynamic operations on an architecture without introducing ancilla qubits is extremely valuable in a practical setting where there are limitations on the space and time overhead that may be tolerated.

In the setting of noisy quantum computation, the ability to remove entropy through intermediate measurements is extremely valuable and is fundamentally what enables quantum error correction.
As shown in Ref.~\cite{aharonov1996limitationsnoisyreversiblecomputation}, when there is not an ability to pump in fresh ancilla qubits, depolarizing noise cannot be error-corrected.
Ben-Or et al.~\cite{ben-or2013quantumrefrigerator} characterize classes of single-qubit noise based on their expressivity in the setting where intermediate measurements are not allowed.

Another classic setting where measurements and feedforward operations play an important role is that of measurement-based quantum computation \cite{briegel2009measurementbasedquantumcomputation}.
In this model, we prepare a standard, fixed, entangled resource state, which is then measured qubit by qubit in a basis that depends on the desired computation to be performed as well as the past sequence of measurement results.

Lastly, we discuss the idea of using intermediate measurements to spread correlations faster than otherwise unitarily possible.
The crucial idea is that classical communication of measurement results from one part of the system to another in most hardware architectures is much faster than the time needed to apply a gate and can be assumed to be almost free.
As an example of the power this affords, an unbounded quantum fanout gate may be applied in constant depth, thus creating correlations between arbitrarily far qubits in $O(1)$ time \cite{baumer2024efficientlongrangeentanglementusing,baumer2024measurementbasedlongrangeentanglinggates,song2024realizationconstantdepthfanoutrealtime,baumer2024quantumfouriertransformusing,corcoles2021exploitingdynamicquantumcircuits}.
This would not be possible using only unitary operations, which would need $\Omega(\mathrm{diam}(\mathsf{G}))$ depth to create nonlocal correlations between two vertices of the architectural graph $\mathsf{G}$ separated by $\mathrm{diam}(\mathsf{G})$.
Recently, this topic has received heightened interest, specifically with the viewpoint of classifying long-range entangled states \cite{piroli2021quantumcircuitsassistedlocal,lu2022measurementshortcutlongrangeentangled,bravyi2022adaptiveconstantdepthcircuitsmanipulating,tantivasadakarn2023hierarchytopologicalorderfinitedepth,li2023symmetryenrichedtopologicalorderpartially} in many-body physics and strategies to prepare classes of states using adaptivity \cite{buhrman2024statepreparationshallowcircuits,malz2024preparationmatrixproductstates,smith2024constantdepthpreparationmatrixproduct,piroli2024approximatingmanybodyquantumstates}.
For experimental demonstrations, see \cite{corcoles2021exploitingdynamicquantumcircuits,verresen2022efficientlypreparingschrodingerscat,tantivasadakarn2023shortestroutenonabeliantopological,foss-feig2023experimentaldemonstrationadvantageadaptive,baumer2024efficientlongrangeentanglementusing,baumer2024measurementbasedlongrangeentanglinggates,song2024realizationconstantdepthfanoutrealtime,baumer2024quantumfouriertransformusing}.
\subsection{Dynamic parameterized quantum circuit architecture proposals}\label{ss:proposals}

We now present the class of parameterized dynamic circuits that we argue are a promising class of circuits to consider for variational quantum algorithms.
We present the full class of parameterized dynamic circuits for which our proofs hold in Appendix \ref{app:definitions}.

We work with quantum circuits composed of two-qubit gates over $n$ qubits with a total depth $d$.
The total number of gates is denoted $m$.
In this work, we consider circuits with nonunitary operations, which we describe through channels, denoted by calligraphic letters such as $\mathcal{U, C}$.
For unitary channels, $\mathcal{U}(\rho) := U \rho U^\dag$.

We begin with the notion of an architecture.
\begin{defn}[Architecture]
A quantum computing architecture is defined by a family of graphs whose vertices represent qubits.
Edges connect vertices whose representative qubits may participate in a two-qubit operation in a layer of the circuit.
\end{defn}

\begin{defn}[Dynamic operation]
A dynamic operation is a quantum channel that may be implemented via a projective measurement followed by a future operation conditioned on the measurement result.
\end{defn}
As an example, the simplest nontrivial dynamic operation is the following:

\begin{align}
\begin{quantikz}[row sep={1cm,between origins}]
\qw & \gate{\mathcal{F}} & \qw
\end{quantikz}
 = 
\begin{quantikz}[row sep={1cm,between origins},transparent]
\qw & \meter{} & \setwiretype{c} & \gate{\text{If 1, } U_1} & \setwiretype{q}
\end{quantikz}. \label{eq:basic_ff}
\end{align}

Here, the state is measured, followed by a conditional gate $U_1$ if the measurement result is $1$ and $I$ otherwise.
If we would like to perform a different operation $U_0$ when the result is a 0, we can account for it by adding a fixed gate after the feedforward operation as follows:
\begin{center}
\begin{quantikz}[row sep={1cm,between origins},transparent]
\qw & \meter{} & \setwiretype{c} & \gate{\text{If 1, } U_0^\dag U_1} & \gate{U_0} \setwiretype{q} & \qw
\end{quantikz}.
\end{center}
We can also simulate the probabilistic application of $\mathcal{F}$ through a parameter $\theta$ that controls the probability with which it is applied, written out in purified form as:

\begin{center}
\begin{quantikz}[row sep={1cm,between origins}]
\qw & \gate[style={fill=blue!20}]{\mathcal{F}(\theta)} & \qw
\end{quantikz}
$:=$
\begin{quantikz}[row sep={1cm,between origins},transparent]
\lstick{$\ket{0}$} & \gate{R_X(\theta)} & \ctrl[vertical wire=c]{1} & \ground{} \\
\qw & & \gate{\mathcal{F}} & \qw
\end{quantikz}.
\end{center}
The only role played by the parameter $\theta$ is to tune the state of the control qubit and hence the probability with which the target qubit is measured, given by $\sin^2(\theta/2)$.

In this work, we study the power of dynamic circuits with simple single-qubit feedforward operations $\mathcal{F}$.
In our theoretical analysis, we study the case where $U_0= I$ and $U_1$ is fixed to be
\begin{align}
U_1 =
\begin{pmatrix}
\cos {\varphi}e^{-i\phi} & -i \sin \varphi \\
- i \sin \varphi & \cos{\varphi} e^{i \phi}
\end{pmatrix}.
\end{align}
In some of our numerical results in \cref{section:utility}, we also allow for these gates $U_0$ and $U_1$ to be parameterized.

\begin{defn}[Parameterized dynamic circuit]\label{def:parameterized_dynamic_circuit}
A parameterized dynamic circuit $\mathcal{C}$ of depth $d$ on an architecture is a sequence of channels $(\mathcal{U}_1(\btheta), \mathcal{U}_2(\btheta), \ldots \mathcal{U}_d(\btheta))$, where each channel $\mathcal{U}_t(\btheta)$ is a completely positive, trace preserving map on $n$ qubits that describes the operations in time step $t$.
The operations in each layer $t$ can be written as a composition of single- and two-qubit operations such that the two-qubit operations act on non-overlapping qubits and only connect qubits with an edge in the associated graph describing the circuit architecture.
The parameters $\btheta$ come from a set $\Theta \in \mathbb{R}^{p}$ for $p \in \poly(n)$.
The action of the entire circuit is described by the channel $\mathcal{C}(\btheta) = \mathcal{U}_d(\btheta) \circ \ldots \circ \mathcal{U}_1(\btheta)$.
\end{defn}

Consider a loss function for which the circuit-to-object map defined earlier necessitates a circuit on $n'$ output qubits.
One way to define a parameterized dynamic circuit is as follows.
Let $n=n'+n_a$ for a suitable integer $n_a$ and let $\mathsf{G}$ be the architecture graph on $n$ qubits describing the hardware connectivity.
We consider edge colorings of the graph.
An edge coloring is described by a list of colors $i = \left\{0, 1,2,\ldots \right\}$ together with the set of edges assigned the color $i$.
The edge coloring defines a sequence of parameterized two-qubit unitary channels in the natural way: instantiate a parameterized two-qubit gate for every edge in a coloring and apply them in parallel.
Then cycle through all the colors in the order $1,2, \ldots$ (we choose not to apply gates on the edges assigned the color $0$).
For a given graph, there can be several valid edge colorings, and each of these could potentially lead to a different ansatz.
After every unitary layer ${U}_j(\btheta)$, we apply the feedforward operation $\mathcal{F}$ on a subset $F_j$ of the $n_a$ ancillary qubits.
The $j$'th channel of the parameterized dynamic circuit is given by $\mathcal{U}_j(\btheta)(\cdot) = ({\bigcirc_{i \in F_j}^n} \mathcal{F}_i) \circ \left( U_j(\btheta) (\cdot) U_j(\btheta)^\dag \right) $, where the symbol $\bigcirc_{i \in F_j}^n$ denotes a sequential composition of channels.

\begin{figure}
\begin{center}
\resizebox{\columnwidth}{!}{\begin{quantikz}[row sep={1cm,between origins}]
\lstick{$\ket{0}$} & \gate[6]{U_1(\btheta)}  & \gate[style={fill=blue!20}]{\mathcal{F}(\btheta)} \slice{layer 1} \setwiretype{q} & \gate[6]{U_2(\btheta)} & \slice{layer 2}  &  \ \ldots \ & \ground{}  \\
\lstick{$\ket{0}$} & & &  &  &   \ \ldots \  & & \meter{} \\
\lstick{$\ket{0}$} & & \gate[style={fill=blue!20}]{\mathcal{F}(\btheta)} \setwiretype{q}  & & \gate[style={fill=blue!20}]{\mathcal{F}(\btheta)} & \ \ldots \ &\ground{}   \\
\lstick{$\ket{0}$} & & &  &  & \ \ldots \  & & \meter{} \\
\lstick{$\ket{0}$} & & \setwiretype{q} &  & \gate[style={fill=blue!20}]{\mathcal{F}(\btheta)} &  \ \ldots \ & \ground{} \\
\lstick{$\ket{0}$} & & &  & &   \ \ldots \   &  & \meter{}
\end{quantikz}}
\end{center}
\caption{A schematic dynamic parameterized circuit on $6$ qubits with $3$ of them corresponding to ancillas.} \label{fig:dpc} 
\end{figure}

We illustrate in \cref{fig:dpc} a schematic dynamic circuit with feedforward operations.
Here, the qubit lines 1, 3, and 5 are the ancillas that we trace over at the end of the circuit.
The qubit lines 2, 4, and 6 are designated to be the system qubits on which the circuit-to-object map and the loss function are defined.
The specific pattern of when and where we apply the feedforward operations $\mathcal{F}(\btheta)$ is fixed beforehand during ansatz selection, along with the choice of an edge coloring and an order in which to apply gates in each block $U_1(\btheta), U_2(\btheta) \ldots U_d(\btheta)$.

Lastly, we define a property of parameterized dynamic circuits, in terms of which we will state our main result on barren plateaus. This definition will use the notion of distance of an observable from a feedforward operation, which we describe more formally in \cref{def_ff_distance} (Appendix \ref{app:definitions}).
Informally, for a site $j$, the \emph{feedforward distance} is the minimum number of entangling two-qubit gates that are encountered in the backwards light cone of qubit $j$ before hitting a feedforward operation or the initial state.
\Cref{fig:ff_distance} illustrates this definition.

\begin{defn}[Informal version of \cref{def_ff_distance}]
A parameterized dynamic circuit is said to have a worst-case feedforward distance of $f$ if, for every qubit $j$, the feedforward distance of $j$ is at most $f$.
\end{defn}
In \cref{ss:absence} we state our results on the sufficient conditions needed for a parameterized dynamic circuit to avoid barren plateaus.

\subsection{Relation to existing nonunitary architectures}\label{ss:relation}

As mentioned in the introduction, there already exist in the literature a variety of proposals for nonunitary circuit architectures \cite{cong2019quantumconvolutionalneuralnetworks,beer2020trainingdeepquantumneural,beer2021dissipativequantumgenerativeadversarial,mele2024noiseinducedshallowcircuitsabsence,heredge2025nonunitaryquantummachinelearning,ilin2024dissipativevariationalquantumalgorithms,yan2024variationalloccassistedquantumcircuits}. Here we discuss the relation of our work to these existing proposals. 

Perhaps the most prominent of existing nonunitary PQC architectures are Quantum Convolutional Neural Networks (QCNNs)~\cite{cong2019quantumconvolutionalneuralnetworks}. These are particularly interesting in light of their provable absence of barren plateaus~\cite{pesah2021absencebarrenplateausquantum}. QCNNs are designed around a very specific nonunitary operation, and as such can be seen as a particular subset of the PQC architectures that we study in this work.
QCNNs are effectively classically simulable and are not believed to be able to express arbitrarily deep unitary quantum circuits \cite{cerezo2023doesprovableabsencebarren,bermejo2024quantumconvolutionalneuralnetworks}.
Indeed, one of our contributions is showing that nonunitary circuit operations can be used much more generally to avoid barren plateaus, without sacrificing circuit expressivity.

Another prominent class of nonunitary PQC architectures are the so-called dissipative quantum neural networks~\cite{beer2020trainingdeepquantumneural,bondarenko2020quantumautoencodersdenoisequantum,beer2021dissipativequantumgenerativeadversarial,poland2020nofreelunchquantum}. The trainability of these architectures has also been studied to some extent~\cite{sharma2022trainabilitydissipativeperceptronbasedquantum}, with both positive and negative results for specific architectural choices. Special cases of dissipative quantum neural networks coincide with special cases of dynamic PQC architectures we formulate and study here.

A related work is also that of Ref.~\cite{mele2024noiseinducedshallowcircuitsabsence}, who study PQC architectures subject to \emph{nonunital} noise. Indeed, one can view these nonunital noise channels as a specific instance of the nonunitary operations we allow for in DPQC architectures. However it is crucial to note that in the DPQC architectures we consider, one has control over the density, type and probability of nonunitary operations in the architecture, which is not the case for the PQC architectures with nonunital noise studied in Ref.~\cite{mele2024noiseinducedshallowcircuitsabsence}. As we observed in Observation~\ref{obs:tuning_between_expressive_and_BP_free}, the DPQC architectures we consider allow one to continuously interpolate between expressive unitary architectures and BP-free architectures, which is not the case for circuits subject to environmental nonunital noise.
The extra control we allow for in our setup is a natural assumption to model current quantum hardware that has the capability of applying dynamic operations such as resets.
As a consequence of the extra control, there is also a straightforward worst-case hardness result we can claim, even in the presence of noise, utilizing error-correction arguments.

Very recently, Refs.~\cite{ilin2024dissipativevariationalquantumalgorithms,yan2024variationalloccassistedquantumcircuits} introduced nonunitary PQC architectures involving mid-circuit measurements, similar to those we explore here. However, their focuses differ from ours. Specifically, Ref.~\cite{ilin2024dissipativevariationalquantumalgorithms} focuses on the noise resilience of these architectures, providing numerical evidence for their utility in thermal state preparation under noise but leaving the question of trainability and barren plateaus unresolved—issues we directly address here.
Concurrently and independently, Ref.~\cite{yan2024variationalloccassistedquantumcircuits} explores state preparation using mid-circuit measurements to achieve shallow-depth circuits. To ensure absence of barren plateaus, they require constant circuit depth. By contrast, our result on barren plateau absence (see \Cref{thm:BP_DPQC_informal}) does not impose depth restrictions; rather, we characterize the absence of barren plateaus via the feedforward distance. Additionally, we apply a distinct technique, the statistical mechanics model, which may be of independent interest. Finally, we offer new insights into trainability beyond barren plateaus (c.f.~\cref{section:trainability}) and classical simulability (c.f.~\cref{section:classical_hardness}) of these architectures.

\section{Trainability}\label{section:trainability}

\subsection{What does it mean for a PQC architecture to be trainable?}\label{ss:what_is}

Having defined the dynamic parameterized quantum circuit (PQC) architectures we will study in this work, we now move on to a discussion of the ``trainability'' of such architectures within the context of variational quantum algorithms. 
Informally, and similarly to what is suggested in Ref.~\cite{gil-fuster2024relationtrainabilitydequantizationvariational}, we consider a PQC architecture $\mathcal{C}$ to be \textit{trainable} if with high probability Algorithm~\ref{alg:variational_quantum_algorithm} converges, in a reasonable amount of time, to a set of circuit parameters $\btheta$ which is almost as good as the optimal circuit parameters 
\begin{equation}
\btheta^* := \argmin_{\btheta\in\Theta} L_M(\btheta).
\end{equation}
Before trying to make the above notion more formal, there are a few points worth highlighting:
\begin{enumerate}
\item As the above notion of trainability relies on properties of Algorithm~\ref{alg:variational_quantum_algorithm}, it clearly depends not just on the PQC architecture $\mathcal{C}$, but also on the circuit-to-object map $M$, loss function $L$, parameter update rule $\mathsf{PU}$, initialization strategy $\mathsf{IS}$ and convergence criteria $\mathsf{CC}$. In particular, a PQC architecture $\mathcal{C}$ may be trainable with respect to some set of these choices, and not trainable with respect to another set.
\item Importantly, the notion of trainability we have given above only requires that the circuit parameters $\btheta$ are almost as good as the best possible parameters in $\Theta$. It \textit{does not} put any additional ``absolute'' requirement on the optimal parameters. Said another way, it could be that even though $\btheta^*$ are the best possible parameters for our model class, the corresponding model $M(C_{\btheta^*})$ is still a poor model. To get any stronger guarantee on the quality of the solution, we have to ensure that in addition to being trainable, $\mathcal{C}$ also contains good solutions.
\end{enumerate}
Essentially, we say that a PQC architecture $\mathcal{C}$ is trainable, if we can efficiently and reliably find a set of circuit parameters defining a circuit which is almost as good, with respect to $L$, as the best solution we could hope to find within $\mathcal{C}$. With this in hand, a natural way to formalize the notion that a candidate set of parameters $\btheta$ is ``almost as good'', would be if $L_M(\btheta) \leq L_M(\btheta^*) + \epsilon$ for some desired small $\epsilon$. Additionally, we could say that Algorithm~\ref{alg:variational_quantum_algorithm} converges ``efficiently'', if it requires at most $\mathrm{poly}(n,1/\epsilon,1/\delta)$ update steps, where $n$ is the size of the problem, $\epsilon$ is the desired accuracy with respect to the optimal parameters, and $1-\delta$ is the desired probability of success\footnote{We note that this definition of trainability looks extremely similar to the definition of agnostic learning~\cite{kearns1992efficientagnosticlearning}. However, here $L_M$ is the loss function which is actually evaluated and minimized during Algorithm~\ref{alg:variational_quantum_algorithm}---i.e., the \textit{empirical} risk---whereas in the agnostic learning setting the requirement is with respect to the \textit{true} risk.}.

Unfortunately, proving any formal trainability statement of the above type for a nontrivial PQC architecure $\mathcal{C}$, realistic problem $L$, and state-of-the-art parameter update rule $\mathsf{PU}$, is formidably difficult~\cite{gil-fuster2024relationtrainabilitydequantizationvariational}. On the other hand, it is often easier, at least for gradient-based parameter update rules, to prove \emph{negative} trainability results via barren plateaus~\cite{larocca2025barrenplateausvariationalquantuma,nietner2023unifyingquantumstatisticalparametrized}. At a high level, one says that a PQC architecture $\mathcal{C}$ admits a barren plateau under initialization strategy $\mathsf{IS}$ if, with high probability when choosing the circuit parameters according to $\mathsf{IS}$, one finds that the expectation value and variance of the gradient $\nabla L_M(\btheta_0)$ are exponentially close to zero. Intuitively, if this is the case then one expects the standard gradient descent parameter update rule $\btheta_{i+1} = \btheta_i - \alpha \nabla L_M(\btheta_i)$ to fail, as it can only lead to very small parameter changes in any polynomial number of update steps. We will discuss barren plateaus in much more detail in \cref{ss:BP}, however for our discussion here it is important to stress the following:
\begin{center}
\textit{For a given PQC architecture $\mathcal{C}$, absence of a barren plateau is not sufficient for trainability.}
\end{center}
More specifically, while the presence of a BP can be used to rule out trainability, the absence of a barren plateau does not prove trainability. To give an example, one could imagine a loss landscape filled with many local minima, none of which is $\epsilon$-close to the global minimum with respect to $L_M$. In such a landscape, we would expect a gradient based update rule to efficiently find a local minimum, but this local minimum may not be good enough to satisfy the requirements of trainability.

Given all of the above, and in particular the difficulty of formally proving rigorous trainability statements for meaningful problems, we do not in this work prove trainability of dynamic parameterized quantum circuits. Instead, our approach in this work is to provide as much well motivated \emph{evidence} for the trainability of dynamic parameterized quantum circuit architectures as possible. We do this by:
\begin{enumerate}
\item Proving an absence of barren plateaus for a natural distribution over circuit parameters---i.e. we show that the necessary but not sufficient condition for trainability is satisfied for dynamic parameterized circuit architectures that satisfy certain explicit criteria. We do this in Section~\ref{ss:absence}.
\item Providing numerical experiments in Section~\ref{section:utility} that show that, at least for one meaningful ground state search problem, gradient-based VQAs using dynamic parameterized quantum circuits can reach good solutions in a reasonable amount of time. We also leverage the nonunitary nature of DPQCs and additionally study problems where the target state is a mixed state, such as Gibbs states of certain Hamiltonians.
Importantly, we also acknowledge and discuss caveats and criticisms of these numerical experiments in Section~\ref{ss:criticisms}. 
\end{enumerate}
With this in hand, we proceed to provide a brief overview of barren plateaus, before proving our central absence of barren plateau result in Section~\ref{ss:absence}.

\subsection{Barren plateaus}\label{ss:BP}
Loss landscapes that are on average exponentially flat and gradients that are on average exponentially small in the number of system qubits---\emph{barren plateaus}---represent a substantial bottleneck in the applicability of variational quantum algorithms to a practically relevant problem (size)~\cite{mcclean2018barrenplateausquantumneural, cerezo2021costfunctiondependentbarren, cerezo2021higherorderderivativesquantum, holmes2022connectingansatzexpressibilitygradient, ortizmarrero2021entanglementinducedbarrenplateaus, wang2021noiseinducedbarrenplateausvariational, napp2022quantifyingbarrenplateauphenomenon, uvarov2021barrenplateauscostfunction}.
In fact, they are one of the central challenges that might block the scalability of these algorithms. The causes for the occurrence of this phenomenon are multifold: ansatz depth or expressivity~\cite{mcclean2018barrenplateausquantumneural, holmes2022connectingansatzexpressibilitygradient}, entanglement~\cite{ortizmarrero2021entanglementinducedbarrenplateaus}, unital hardware noise~\cite{wang2021noiseinducedbarrenplateausvariational}, and loss functions induced by global observables, namely, observables acting on most system qubits~\cite{cerezo2021costfunctiondependentbarren}.
More specifically, a loss function $\mathcal{L}({\btheta})$ is said to suffer from a barren plateau if, for all parameters $\btheta_k$, the variance of the gradient decays exponentially---i.e. $\Var_{{\btheta}}\left[ \partial_k \mathcal{L} \right] \in O(1/b^n)$---with respect to $\btheta$ drawn from some natural distribution over $\Theta$.
As proven in Ref.~\cite{arrasmith2022equivalencequantumbarrenplateaus}, under some assumptions, this gradient behavior is often directly equivalent to an exponential concentration of the loss function itself, namely, $\Var\left[\mathcal{L} \right] \in O(1/b^n)$ for some $b > 1$. 
Bounds on the concentration of both loss and gradients which may be computed efficiently with classical resources are given in Refs.~\cite{napp2022quantifyingbarrenplateauphenomenon,letcher2024tightefficientgradientbounds, uvarov2021barrenplateauscostfunction}.

While strategies to circumvent barren plateaus are known, most of them suffer from drawbacks that limit their utility in practice. Firstly, one may use smart parameter initialization strategies \cite{grant2019initializationstrategyaddressingbarren, rudolph2023synergisticpretrainingparametrizedquantum,  wang2023trainabilityenhancementparameterizedquantum, zhang2022escapingbarrenplateaugaussian}. These, however only hold for the beginning of the training and do not help to exclude the occurrence of barren plateaus throughout the training.
Moreover, parameter initialization according to distributions concentrated on small subsets of parameter space can restrict the ability to explore the space of good solutions.
Another strategy to avoid exponentially flat loss landscapes is to choose ansatz classes which have been proven to not suffer from barren plateaus~\cite{pesah2021absencebarrenplateausquantum, sharma2022trainabilitydissipativeperceptronbasedquantum}.
However, while being barren plateau free, the previously suggested ansatz classes suffer from restricted expressivity, e.g., due to a constant circuit depth or restricted dynamical lie algebra, which in turn can easily lead to classical simulability of the respective model~\cite{cerezo2023doesprovableabsencebarren}.
For a detailed review on barren plateaus, we refer the interested reader to Ref.~\cite{larocca2025barrenplateausvariationalquantuma}.

\subsection{Absence of barren plateaus in dynamic parameterized quantum circuits}\label{ss:absence}

We now state our first main result.
Let $\rho(\bm{\theta})=\mathcal{C}(\bm{\theta})\left( \ketbra{0^n} \right)$ be the output state of the parameterized DPQC ensemble. In particular, assume that the ensemble $\mathcal{C}(\bm{\theta})$ satisfies the following properties:
\begin{enumerate}
\item Every component $\btheta_i$ of $\btheta$ parameterizes only a single operation in $\mathcal{C}$.
\item $\mathcal{C}$ is locally scrambling, i.e., is invariant under single-qubit random gates from a unitary 2-design after every gate.
\item $\mathcal{C}$ has constant worst-case feedforward distance $f$.
\item $\mathcal{C}$ has an average entangling power of $\alpha \in [0,1]$ \cite{zanardi2000entanglingpowerquantumevolutions,wang2003entanglingpoweroperatorentanglement} and average swapping power $\beta$ with the conditions $\beta \in [0, 1-\alpha]$ (see, for example, Ref.~\cite{ware2023sharpphasetransitionlinear} for a definition of these parameters).
\end{enumerate}
Under these conditions, we state two results:
\begin{subtheorem}{thm}
\begin{thm}[Variance bound for $k$-local Hamiltonians---formal version of \Cref{thm:BP_DPQC_informal}] \label{thm:variance-bound-formal}
Let $H$ be a $k$-local Hamiltonian and suppose it has an expansion $\sum_{\bm{\alpha}} c_{\bm{\alpha}} \bm{\alpha} $ in the Pauli basis $\bm{\alpha} \in \mathbb{P}_n$.
Then, the variance of the loss function $L = \Tr \rho(\bm{\theta}) H $ with respect to the ensemble of circuits $\mathcal{C}$ is lower bounded as
\begin{align}
\Var_{\btheta }{L} \geq \sum_{\bm{\alpha}} c_{\bm{\alpha}}^2 \left( \frac{\sin^2 \varphi}{3} \right)^{\abs{\bm{\alpha}}}  \left( \frac{\alpha}{5} \right)^{kf}.
\end{align}
For Hamiltonians with locality at most $k$, this simplifies to 
\begin{align}
\Var_{\btheta }{L} \geq \norm{H}_{HS}^2 \left( \frac{\sin^2 \varphi}{3} \right)^{k}  \left( \frac{\alpha}{5} \right)^{kf},
\end{align}
where $\norm{H}_{HS} := \sqrt{{\Tr H^2}} $ is the Hilbert-Schmidt norm of $H$. \newline
\end{thm}
\begin{thm}[Formal version of \Cref{thm:robustness-informal}]
Under the same conditions as in \cref{thm:variance-bound-formal}, in the case of noisy quantum circuits where each gate is followed by a noise channel with nonunitarity $\gamma$ and nonunitality $\delta \leq \gamma \leq \frac{1}{2}$ (see Appendix~\ref{app:SM} for definitions), the lower bound on the variance is
\begin{align}
\Var_{\btheta }{L} \geq \norm{H}_{HS}^2 \left( \frac{\sin^2 \varphi}{3} \right)^{k}   \left(\frac{\alpha}{5}(1-\gamma-\delta) + \delta \right)^{kf}.
\end{align}
\end{thm}
\end{subtheorem}

In the above, the local scrambling condition implies that the distribution over circuits is identical to the distribution where, after every operation, one inserts random single-qubit gate from any 2-design ensemble.
The quantity $\alpha$ is related to the entangling power of the ensemble of 2-qubit gates \cite{ware2023sharpphasetransitionlinear}.
We may also fix the 2-qubit gates deterministically.

Note the crucial point that these bounds are \emph{independent} of the number of qubits $n$ and the depth $d$ of the circuit.
This fact is what enables us to consider deep circuits of this form, restoring expressivity.
Remarkably, even in the presence of noise, which causes noise-induced barren plateaus unless the noise is nonunital \cite{wang2021noiseinducedbarrenplateausvariational,mele2024noiseinducedshallowcircuitsabsence}, the feedforward operations help fight noise-induced barren plateaus.

The proof idea of these results relies on the stat-mech model mapping described in the following section.
We extend these techniques to the case of more general nonunitary operations.
The technique involves analyzing a biased random walk over a configuration space of identity and swap operators, interpreted as bitstrings $\{0,1\}^n$.
The walk is biased towards the 0 state.
On the other hand, the variance of the loss function is related to the probability of obtaining a significant number of 1s, or large Hamming weight, in the output of the random walk.

In the presence of nonunitary operations, such as the simple measurement and feedforward operation mentioned above, we uncover a mechanism that changes the operator walk dynamics.
Specifically, these operations add a new bias term in the reverse direction, causing a 0 to flip to a 1 with a certain nonzero probability.
As we show in \cref{lem_singlequbit_bound,lem_multiqubit_bound}, this mechanism is enough to lead to a lower bound on the probability of observing a nonzero Hamming weight at the output of the walk.
This, in turn, yields a lower bound on the variance of the loss function.

\subsection{The stat-mech model}\label{ss:SM}
We briefly review here the formalism of the statistical mechanical model (``the stat-mech model'') that allows us to compute second moments of statistical quantities over ensembles of random quantum circuits~\cite{hunter-jones2019unitarydesignsstatisticalmechanics,dalzell2022randomquantumcircuitsanticoncentrate,dalzell2021randomquantumcircuitstransform,napp2022quantifyingbarrenplateauphenomenon,ware2023sharpphasetransitionlinear}.
We review this more thoroughly in Appendix \ref{app:SM}.
Concretely, the stat-mech model helps evaluate quantities of the form $\mathbb{E}_U[\Tr \rho O]^2$ for an output state $\rho$ of a quantum circuit with potentially nonunitary elements, averaged over the choice of random unitary gates $U$.
This is done by observing that $\mathbb{E}_U[\Tr \rho O]^2 = \mathbb{E}_{U} [\Tr \rho^{\otimes 2} O^{\otimes 2}]$, which can be rewritten as $\Tr \mathbb{E}_{U}[\rho^{\otimes 2}] O^{\otimes 2}$.
This illustrates that the quantity $\mathbb{E}_{U}[\rho^{\otimes 2}]$, which we call the 2-copy average state, is the fundamental object of interest for calculating second moment quantities.
Henceforth, we will denote this state by $\bar{\rho}$.

We can derive the stat-mech model by observing that the average two-copy state $\bar{\rho}$ can be classically tracked using the well-known Weingarten calculus for computing moments of the Haar measure.
Specifically, if we assume that every 2-qubit gate or noise channel in the circuit is followed by a single-qubit Haar-random gate, then the state $\bar{\rho}$ lies in the symmetric subspace spanned by the operators $\{I, S\}^n$, where $I$ is the identity operation and $S$ the SWAP operation between the copies of the state.
Using the trace-1 normalized versions of these operations $\mathtt{I}:= \frac{I}{4}$, and $\mathtt{S}:=\frac{S}{2}$, we can write
\begin{align}
  \bar{\rho} = \sum_{\bm x \in \{0,1\}^n } c_{\bm x} \mathtt{I}^{1-\bm{x}_1} \cdot \mathtt{S}^{\bm{x}_1} \otimes \ldots \otimes \mathtt{I}^{1-\bm{x}_n} \cdot \mathtt{S}^{\bm{x}_n},
\end{align}
where $\sum_{\bm x} c_{\bm x} = 1$.
We thus see that the average two-copy state is characterized by a (quasi)-probability distribution over the space $\{\mathtt{I},\mathtt{S}\}^n$.
We will henceforth speak about the average two-copy state and its associated distribution $\mathcal{X}$ over $\{\mathtt{I},\mathtt{S}\}^n$ interchangeably.
We also denote by $\bm{x}^t$ the bitstring that specifies an operator in $\{\mathtt{I},\mathtt{S}\}^n$ at time $t$, viewed as a random variable, and by $\bar{\rho}^t$ the average two-copy state at time $t$.

In order to motivate the study of the stat-mech mapping in the context of barren plateaus in variational quantum algorithms, we state here a key lemma due to Napp \cite{napp2022quantifyingbarrenplateauphenomenon} relating the variance of the loss function to a quantity related to the average two-copy state.
In \cref{lem_ham_variance} (Appendix \ref{app:SM}), we derive a generalization of  Napp's lemma to a more general setting where the two-qubit gates need not be chosen from a 2-design.

\begin{lemma}[Informal version of \cref{lem_ham_variance}]
The variance of a $k$-local Pauli observable ${\bm \alpha} \in \mathbb{P}_n$ supported on a region $A = \supp({\bm \alpha})$ over a locally scrambling ensemble of quantum circuits of depth $d$ is
\begin{align}
\Var_{\btheta}[L] = \Pr_{\mathcal{X}_d} \left[\bm{x}_{A} = 11\ldots 1_A \right].
\end{align}
\end{lemma}
In the above, $\mathcal{X}_d$ is the distribution over bitstrings at the end of the circuit.
We see, therefore, that the variance of the loss function directly depends on the distribution $\mathcal{X}_d$ and whether it has high probability weight on strings $\in \{\mathtt{I},\mathtt{S}\}^n$ that lead to the all $\mathtt{S}$ string on the subregion $A$.

\subsection{A unified derivation of barren plateaus from the stat-mech model}\label{ss:unified}
In this subsection, we give a unified derivation of all known sources of barren plateaus in the stat-mech picture.
This is done by examining the physical behavior of the stat-mech model on arbitrary geometries.
This exercise will enable us to form intuition on the causes of barren plateaus, which will inform strategies to mitigate against them.

Let the reduced two-copy average state on two qubits at time $t$ be $\bar{\rho}^t = a \mathtt{II} + b \mathtt{IS} + c \mathtt{SI} + d \mathtt{SS}$.
Then the reduced two-copy average state after application of a unitary channel, given by $\bar{\rho}^{t+1} = \mathbb{E}_{V_1, V_2}[(V_1 V_2 \otimes V_1 V_2) \mathcal{U} \otimes \mathcal{U}(\bar{\rho}^t) (V_1 V_2 \otimes V_1 V_2)^\dag]$, can be expressed as $a'\mathtt{II} + b' \mathtt{IS} + c' \mathtt{SI} + d' \mathtt{SS}$, where
\begin{align}
\begin{pmatrix}
  a' \\
  b' \\
  c' \\
  d'
\end{pmatrix} = T 
\begin{pmatrix}
a \\
b\\
c\\
d
\end{pmatrix}
\end{align}
for a $4 \times 4$ stochastic matrix $T$.
For a Haar-random two-qubit unitary, the transfer matrix is given by \cite{ware2023sharpphasetransitionlinear}
\begin{align}
T_\mathrm{Haar} = 
\begin{pmatrix}
1 & \frac{4}{5} & \frac{4}{5} & 0 \\
0 & 0 & 0 & 0 \\
0 & 0 & 0 & 0 \\
0 & \frac{1}{5} & \frac{1}{5} & 1
\end{pmatrix}.
\end{align}

From these facts, we see that for deep, random circuits on a well-connected architecture, the fixed point of $\bar{\rho}$ is given by $\frac{2^n}{2^n+1} \mathtt{II\ldots I} + \frac{1}{2^n+1} \mathtt{SS\ldots S}$.
This implies that the probability of seeing any substring $\mathtt{S\ldots S}$ is at most $\frac{1}{2^n+1}$, giving barren plateaus for Hermitian observables of any locality.
This observation also applies to circuits with fixed entangling gates, as long as the architecture is well-connected and every operation is followed by random single-qubit unitaries from a 2-design.

We can also easily derive in this formalism the statement on barren plateaus for random circuits of any depth with a global observable.
The key is that when the initial state is a product state across all $n$ qubits, its two-copy average state, $\bar{\rho} = \left(\frac{2}{3}\mathtt{I} + \frac{1}{3}\mathtt{S} \right)^n$ only has an inverse-exponential mass $\frac{1}{3^n}$ on the string $\mathtt{S\ldots S}$.
Also observe that even in the limit $d\to \infty$, the mass $\frac{1}{2^n+1}$ remains inverse-exponentially small.
Informally speaking, the distribution over $\{\mathtt{I},\mathtt{S}\}^n$ is highly constrained and does not allow for a significant probability mass on the $\mathtt{S\ldots S}$ string at any depth.
See Ref.~\cite{napp2022quantifyingbarrenplateauphenomenon} for a more complete proof.

We now turn our attention to circuits with noise, as studied by Ref.~\cite{wang2021noiseinducedbarrenplateausvariational}.
With unital noise, the stat-mech model has an additional update rule given by the transfer matrix
\begin{align}
T_\mathrm{noise} =
\begin{pmatrix}
1 & \gamma \\
0 & 1 - \gamma
\end{pmatrix},
\end{align}
as detailed in Appendix \ref{app:SM}.
This means that any $\mathtt{S}$ string has a probability $\gamma$ to ``decay'' into an $\mathtt{I}$ string, which is a fixed point.
In this case, as analyzed in Refs.~\cite{dalzell2021randomquantumcircuitstransform,ware2023sharpphasetransitionlinear}, the distribution rapidly converges to the $\mathtt{II\ldots I}$ fixed point.
The probability of seeing an $\mathtt{SS}$ string on any subregion decays exponentially in the number of gates applied on that subregion, or the local depth of the circuit.

Finally, we also discuss entanglement-induced barren plateaus \cite{ortizmarrero2021entanglementinducedbarrenplateaus}.
For this, we note that the relevant quantity we consider, the probability mass on any $\mathtt{SS \ldots}$ string, is exactly the expected purity of the reduced density matrix on those sites.
It often turns out for Lipschitz-continuous functions that one can get extremely good concentration bounds for random circuits.
From concentration bounds, one can show the average logarithm of the purity is very well approximated by the logarithm of the average purity.
A small mass on a $\mathtt{SS \ldots}$ string thus implies a large negative logarithm of the purity on average, which in turn is related to the average R\'enyi-2 entropy.

\subsection{Gradient evaluation}
\label{sec:gradient}

In this subsection, we show that the partial derivative of the loss function with respect to parameterized angles, as in \cref{def:parameterized_dynamic_circuit}, can be evaluated using the parameter-shift rule \cite{mitarai2018quantumcircuitlearning,schuld2019evaluatinganalyticgradientsquantum}. Recall that a feedforward operation, $\mathcal{F}(\btheta)$ can be parameterized using a unitary operation $R_X(\btheta)$ on an ancillary qubit, followed by an operation on a system qubit. Therefore, the overall loss function can be defined as a unitary circuit, followed by a measurement of a Hermitian observable. 

Let $H$ denote a Hermitian observable and let $\btheta_{j}$ denote a parameter such that the unitary circuit $\mathcal{C}(\btheta)$ can be expressed as  $\mathcal{C}(\btheta) = V_R(\btheta) \circ \exp(-i(\btheta_j/2)P)\circ V_L (\btheta)$, where $P$ 
is a Pauli operator, and $V_L(\btheta)$ and $V_R(\btheta)$ parameterized unitaries positioned on either side of $\exp(-i(\btheta_j/2)P)$. Note that we denote all parameterized angles together using $\btheta$. 

Given an input state $\sigma$, the loss function $L$ becomes 
\begin{align}
    L(\theta_j) = \Tr\left(H V_R \exp(-i \frac{\btheta_j}{2}P)V_L \sigma V^{\dagger}_L\exp(i \frac{\btheta_j}{2}P)V_R^{\dagger} \right),
\end{align}
where we denoted $L$ as a function of $\btheta_j$ alone, though it depends on other parameters $\btheta$ as well.

Define $B(\theta_j) :=  \exp(-i \frac{\btheta_j}{2}P)$, $\Xi := V_R^{\dagger}HV_R$ and $\zeta := B(\btheta_j) V_L \rho V_L^{\dagger}B(\btheta_j)^{\dagger}$. The partial derivative of $L$ with respect to $\btheta_j$ is given by 
\begin{align}
    \partial_j L: = \frac{\partial L}{\partial \btheta_j} &= -(i/2) \left[ \Tr\left( [P, \zeta]\Xi  \right) \right]\\
    & =  -(i/2) \left[ -i \Tr\left( \left[B^{\dagger}(\pi/2)\zeta B(\pi/2) \right. \right. \right. \nonumber \\
    & \left. \left. \left. - B^{\dagger}(-\pi/2)\zeta B(-\pi/2)\right]\Xi \right)  \right]\\
    & = \frac{1}{2} \left[L(\btheta_j+ \pi/2) - L(\btheta_j -\pi/2)\right].
\end{align}

Thus, the gradient with respect of $\btheta_j$ can be estimated by evaluating the loss function at two parameter-shifted values. When feedforward operations are modeled by a classical random variable $q(\phi_j)$ instead of a unitary on an ancillary qubit, the partial derivative involves updating the distributions $q(\phi_j +\pi/2)$ and $q(\phi_j -\pi/2)$. 

We now summarize the implications of our results on the variance of the gradient of the cost function.  
From \cref{thm:variance-bound-formal}, we know that under certain conditions, the variance of the cost function does not vanish exponentially. This directly implies that the variance of the gradient with respect to some parameters also does not vanish exponentially. The reason is that the variance of the loss function is upper bounded by a quantity that depends on the variances of the gradients with respect to each parameter. Specifically, from Ref.~\cite[Eq.~(C5)]{arrasmith2022equivalencequantumbarrenplateaus}, the variance of the loss function difference between two points is bounded above by $m^2 \max_i \int_0^M \int_0^M \operatorname{Var}_{\bm{\theta}_A}(\partial_i C(\bm{\theta}_A + \ell \hat{\ell})) \, d\ell \, d\ell'$, where $m$ is the number of parameters and $M$ is the distance in parameter space between the two points.
Since we have a $1/\poly(n)$ lower bound on the loss function's variance\footnote{Technically, we have not stated our theorems in this language, but our proofs can be slightly modified to give lower bounds in terms of these quantities.}, it follows that the variance of the gradient with respect to at least one parameter $i$ is also lower bounded by $1/\poly(n)$. In other words, the gradient retains sufficient signal in at least some directions in parameter space, enabling effective navigation toward a local minimum.

\section{Potential utility}\label{section:utility}

Here we give evidence that the parameterized circuit architectures we propose indeed contain good solutions to interesting problems \textit{and} that they can find these solutions in practice. In other words, we provide evidence that at scale they may provide ``quantum utility''. We focus on both ground state and thermal state preparation.
The setting that is explored in the context of ground state preparation almost\footnote{The numerical experiments for both ground states and thermal states feature a slightly different ensemble over single-qubit gates. Specifically, we consider single-qubit gates $U_3$ (see \cref{fig:vqe_setup} for example) to have the Euler form $U_3 = R_Y(\theta_1) R_Z(\theta_2) R_Y(\theta_3)$ with a uniform distribution over angles $\theta_{1,2,3} \sim \mathcal{U}[0,\pi)$. In order to have a single-qubit 2-design, the middle angle $\theta_2$ should instead be chosen from a distribution such that $\cos \theta_2 \sim \mathcal{U}[0,1]$.} satisfies the assumptions of Theorem \ref{thm:variance-bound-formal}, with deterministically applied feedforward operations at fixed locations in the circuit. 

For the setting of thermal state preparation, we perform numerical experiments based on both the minimization of a global loss function (the infidelity) and a local loss function (the gradient of parameters under variational quantum imaginary time evolution).
For the former, we do not expect \cref{thm:variance-bound-formal} to hold, while we do expect it to hold for the latter.
These experiments are meant to explore the expressivity of the ansatz, which can naturally express impure states.
Our numerical experiments can be replicated using the publicly available codebase \cite{zoufal2025dynparqcirclearning}.

\subsection{Ground state preparation}\label{ss:ground_states}

\begin{figure*}[t]
\begin{center}
    \includegraphics[width = \textwidth]{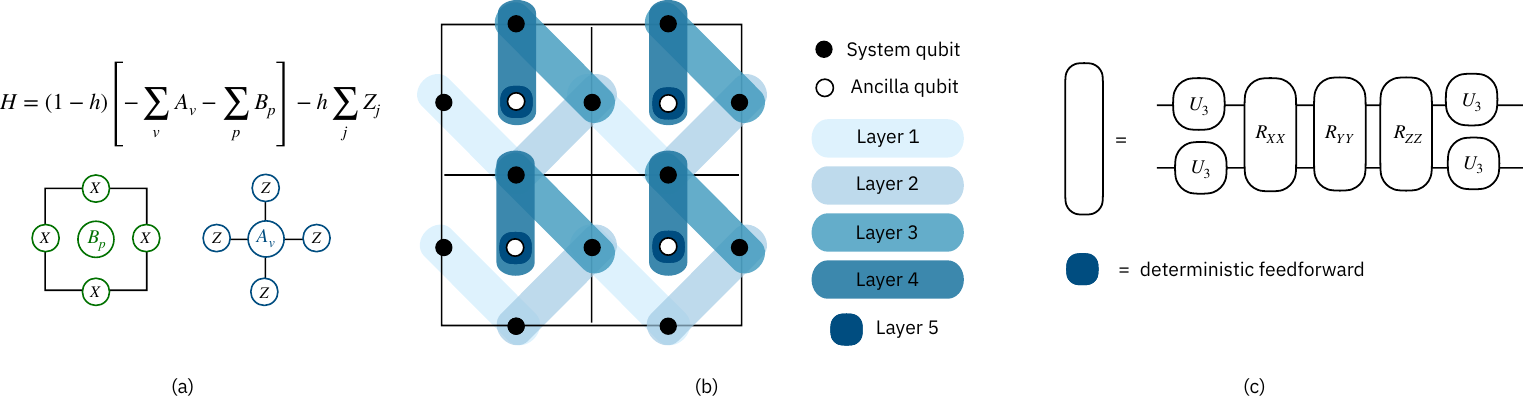}
    \caption{An illustration of the Hamiltonian and DPQC architecture used for ground state experiments. (a) The Hamiltonian is the perturbed toric code, for a system on a square lattice, with system qubits on the edges. (b) One layer of the ansatz consists of five sublayers. The first four are parameterized two-qubit unitary entangling gates, and the last sublayer consists of deterministic reset gates on all ancilla qubits. Gates are applied from lightest to darkest, so that all gates of the same opacity are applied in parallel. (c) Structure of the gates within each sublayer. $U_3$ denotes a generic single-qubit rotation gate with 3 Euler angles. }\label{fig:vqe_setup}
\end{center}
\end{figure*}

In the following we investigate the ground state problem for a Hamiltonian with $12$ qubits that corresponds to a perturbed toric code.
More specifically, the Hamiltonian is
\begin{align}
H_{\text{toric}} = (1-h)H_0-\sum\limits_{j=1}^nhZ_j,
\end{align}
where the first term corresponds to the unperturbed toric code
$\textstyle{H_0 = -\sum\limits_vA_v-\sum\limits_pB_p}$,
with $v$ and $p$ running over all vertices and plaquettes of the corresponding 2D square lattice and $A_v$ and $B_p$ representing  vertex and plaquette operators, respectively, which are given by products of $X$ and $Z$ Pauli operators.
Notably, this model has also been investigated in
Ref.~\cite{zhang2024absencebarrenplateausfinite}.
The ground state of this system exhibits long-range entanglement---a property that makes the representation, e.g., with tensor networks difficult. In order to entangle any two arbitary qubits using a circuit consisting of single and two-qubit gates---as most hardware efficient ans\"atze---one needs, even if all-to-all connectivity is given, at least $\mathcal{O}(\log(n))$ two-qubit gates.

We build our ansatz layers as shown in \cref{fig:vqe_setup}. In particular, each layer of the ansatz consists of five sublayers. The first three sublayers consist of parameterized unitary two-qubit entangling gates applied between system qubits, organized in such a way that in each sublayer all two-qubit gates are applied in parallel. The fourth sublayer consists of an entangling gate between the top qubit of each lattice block, and the ancilla in the centre of each lattice block. The structure of each two qubit gate is shown in \cref{fig:vqe_setup}(c). Finally, the fifth sublayer consists of deterministic resets applied to each ancilla qubit. In the presented experiments, we employed a total of two layers, of five sublayers each.
At this point it should be noted that our ansatz is different from the ansatz employed in Ref.~\cite{zhang2024absencebarrenplateausfinite} as 
all two-qubit gates in a sublayer are applied in parallel, hence reducing the ansatz depth. A deep version of this structure \emph{without ancillary resets} was tested in the above mentioned reference but did not perform well on account of barren plateaus. As illustrated in \cref{fig:training_ground}, we do not observe these performance issues, which we conclude to be circumvented via the insertion of ancilla qubits which are regularly reset.

\begin{figure}[t]
        \centering
        \subfloat[Difference between system and target energy density during training]{\includegraphics[width=0.48\textwidth]{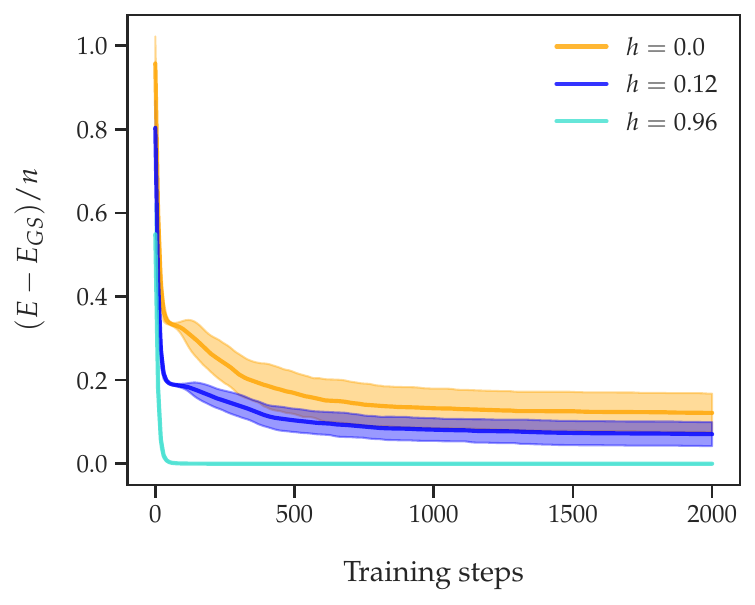}}
        
        \centering
        \subfloat[State purity during training]{\includegraphics[width=0.48\textwidth]{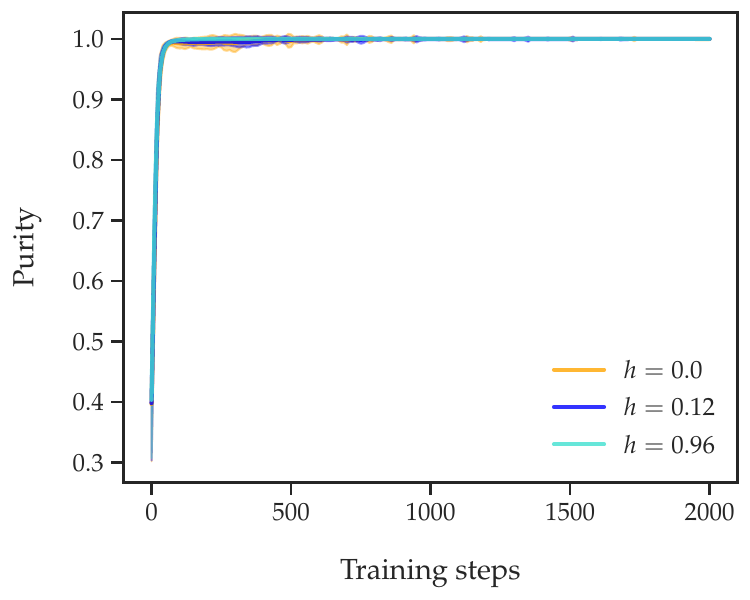}}
    \caption{
    Variational training of a DPQC architecture for a perturbed toric code Hamiltonian acting on 12 qubits. The perturbation strength is controlled by the parameter $h$. All settings were run for $100$ different seeds. The solid lines represent the average over these runs and the shaded, translucent lines illustrate the standard deviation.
    (a) The training dynamics of the loss function, for $h=0,0.12$ and $0.96$. The results indicate quick convergence for all trials. (b) The dynamics of the state's purity during training, depicted for the same values of $h$. Note that the purity does not explicitly enter into the loss function.}
    \label{fig:training_ground}
\end{figure}

The initial parameters are drawn at random from a uniform distribution $[0, \pi)$. We use the ADAM optimizer \cite{kingma2017adammethodstochasticoptimization} with a learning rate of $10^{-2}$ and at most 2000 iterations. The statistics for the different values of $h$ are evaluated from $100$ trial runs. The only exception of these settings are the experiments for $h=1.0$, which had already converged after 500 steps and with $25$ trials.
All experiments were simulated with a tensor network simulator \textit{TensorCircuit} \cite{zhang2023tensorcircuitquantumsoftwareframework}, using a purification on 24 qubits for simplicity of code integration.

The examples shown in \cref{fig:training_ground} illustrate that the training converges quickly for all trials and the different values of $h$. The instances with larger $h$ lead to faster convergence in terms of energy and purity.

\begin{figure}[t]
\subfloat[System and ground-state energy]{
\includegraphics[width=0.48\textwidth]{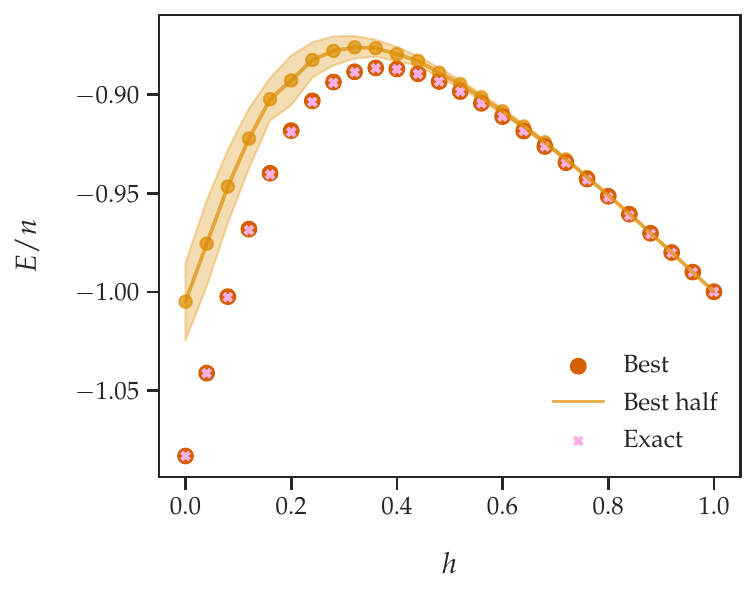}}
\newline
\subfloat[State purity after training]{
\includegraphics[width=0.48\textwidth]{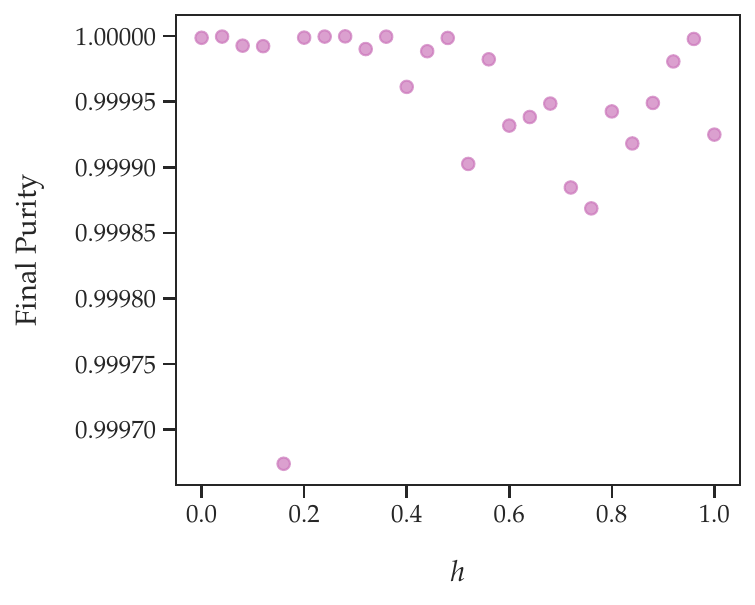}}
\label{fig:new_figure}
\caption{(a) The exact ground-state energy compared to that output by the variational algorithm, measured in terms of the best out of 100 trials (``best'') and the average over the best 50 out of 100 trials (``best half''). For each $h$, the best estimate coincides with the exact ground-state energy. (b) The purity of the final output state for different values of $h$, averaged over all trials.}
\label{fig:purity_energy}
\end{figure}

As illustrated in \cref{fig:purity_energy}(a), the average best found energy  matches the exact energy well for $h\geq 0.5$ with a standard deviation $\leq 0.022$. The best energies found per value for $h$ matches the exact energy quite well with differences $\leq 0.078$ for all tested values of $h$.
Furthermore, it is particularly interesting to note that Fig.~\ref{fig:purity_energy}(b) illustrates how the states at initialization are mostly mixed. However, the purity rapidly increases during the course of training.

\subsection{Thermal state preparation}\label{ss:thermal_states}

\begin{figure}[t]
\begin{center}
    \includegraphics[width = \columnwidth]{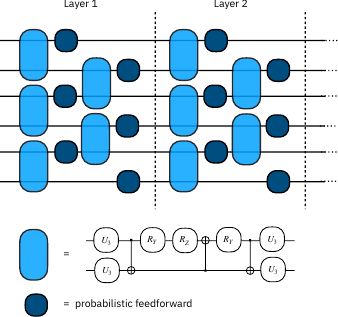}
    \caption{An illustration of the ansatz used for thermal state experiments, for a 6 qubit example. 
    $U_3$ denotes a generic single-qubit rotation gate with 3 Euler angles.
    }\label{fig:thermal_state_setup}
\end{center}
\end{figure}

Next, we investigate the applicability of our ansatz structure to thermal state preparation---which, unlike ground state preparation, typically leads to a mixed state. More specifically, we are looking to prepare a state of the form
\begin{align}
    \rho_\text{Gibbs} \left(H, \beta\right) = \frac{e^{-\beta H}}{\Tr\left[e^{-\beta H}\right]},
\end{align}
with $\beta$ representing the inverse temperature.
While it is possible to generate a mixed state in a unitary parameterized quantum circuit with auxiliary qubits, the DPQC architecture we work with is a natural fit for the task at hand.

\begin{figure*}[t]
\subfloat[$H_{\text{TFI}}$]{
\includegraphics[height=0.8\textwidth]{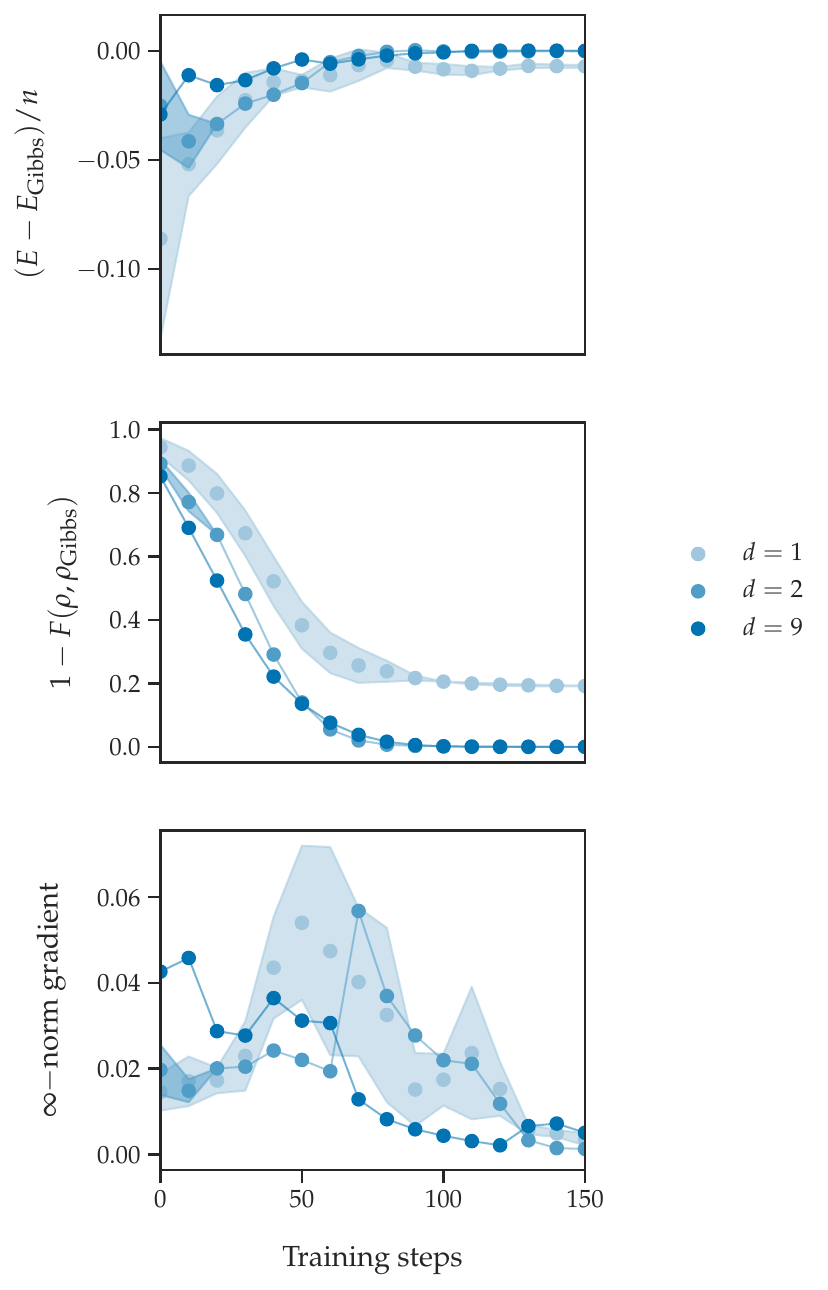}} 
\subfloat[$H_{\text{XY}}$]{
\includegraphics[height=0.8\textwidth]{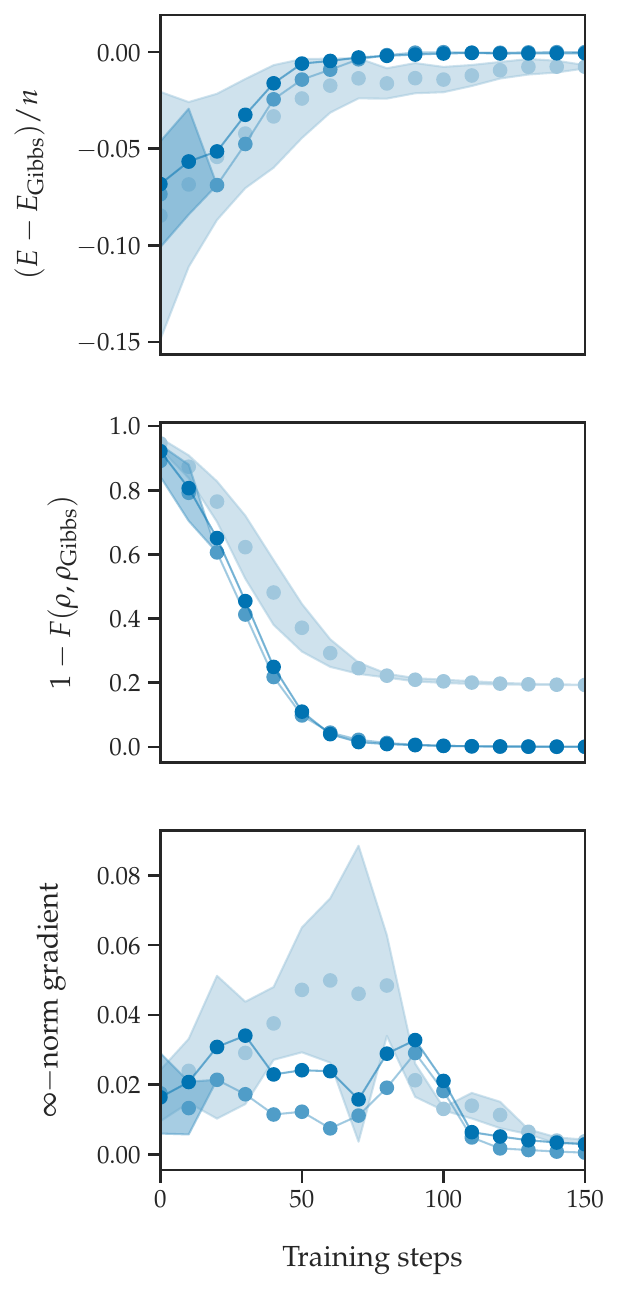}}
\caption{The behavior of the energy density difference with respect to the target state, infidelity to the target state, and the $\infty$-norm of the gradient throughout $150$ training iterations for (a) the transverse field Ising and (b) the XY model for  $n=10$ and number of layers $d\in \{1, 2, 9\}$. The dots mark the average over 5 runs and the filled lines represent one standard deviation.}
\label{fig:thermal_state_res}
\end{figure*}

\begin{figure}[t]
\subfloat[$H_{\text{TFI}}$]{
    \includegraphics[width=0.48\textwidth]{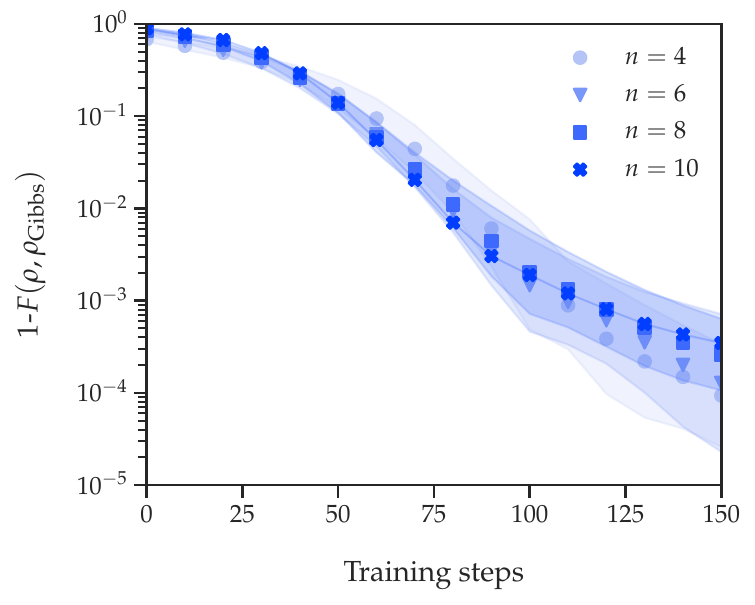}
} \newline
\subfloat[$H_{\text{XY}}$]{
\includegraphics[width=0.48\textwidth]{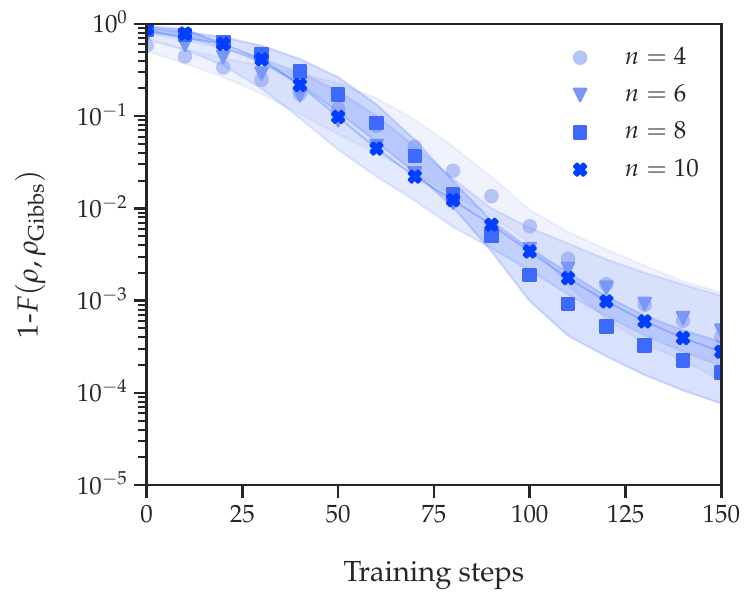}
}
\caption{The behavior of the infidelity with respect to the target state throughout $150$ training iterations for $d=2$ for the (a) transverse field Ising (b) XY model. The points are marked by $e^{\langle \log (1-F) \rangle}$, i.e.\ the exponential of the average log-infidelity; and the filled lines represent the standard deviation of this quantity.}
\label{fig:thermal_state_qubits}
\end{figure}

We consider a transverse field Ising model
\begin{align}
H_{\text{TFI}} =  -\sum\limits_{j=1}^nX_jX_{j+1}-\frac{1}{2}\sum\limits_{j=1}^nZ_j,
\end{align}
and an XY model
\begin{align}
    H_{\text{XY}} = -\sum\limits_{j=1}^n\left[\frac{3}{4}X_jX_{j+1} + \frac{1}{4}Y_jY_{j+1}\right]-\frac{1}{2}\sum\limits_{j=1}^nZ_j,
\end{align}
both on a periodic 1D chain for up to $10$ qubits.
These systems were also investigated for up to 6 qubits in Ref.~\cite{ilin2024dissipativevariationalquantumalgorithms}. 
Notably, the dissipative ansatz suggested in Ref.~\cite{ilin2024dissipativevariationalquantumalgorithms} is compatible with the model discussed in this work---provided that the dissipative gates are set to be probabilistic feedforward operations $\mathcal{F}(\btheta)$ where the application probability is controlled by a trainable parameter.

Notably, the action of $\mathcal{F}(\btheta)$ corresponds to applying the identity operation with probability $p\left(\btheta\right)$ and a reset on to the $\ket{0}$ state with probability $1-p\left(\btheta\right)$.
The reset on to $\ket{0}$ can be implemented straightforwardly with a feedforward operation by setting in \cref{eq:basic_ff} the correction operation $U_1$ after measurement to be the $X$ gate.
An illustration of the ansatz structure is depicted in \cref{fig:thermal_state_setup} for 6 qubits. More specifically, the figure illustrates the layers of the ansatz and their decomposition into sublayers. Notably, after the application of $d$ layers as illustrated in \cref{fig:thermal_state_setup}, a final layer of unitary gates is applied, which is a special case of a $d+1$-layer DPQC up to the removal of the feedforward operations.
The experiments shown in Ref.~\cite{ilin2024dissipativevariationalquantumalgorithms} demonstrate that this ansatz can prepare thermal states up to small errors with a number of layers equal to $n$ . We study the induced training behavior for layer numbers $d=\left\{1,\ldots, n-1\right\}$ which helps us understand whether a shallower form of this ansatz class is already sufficiently expressive to approximate the thermal states with good accuracy. 

First, we train the model $\tilde\rho(\btheta)$ by minimizing the infidelity to the target state $\rho\left(H, \beta\right)$, i.e.\,$1-F(\tilde\rho(\btheta), \rho_{\text{Gibbs}}(H,\beta))$.
While this type of loss function cannot be efficiently probed with small error on a quantum computer, it enables us to study the expressivity of our ansatz.
We choose the initial parameters at random from a uniform distribution $[0, 1]$
and the target inverse temperature $\beta=2$.
Similar to the example on ground state preparation, we optimize with ADAM using a learning rate of $10^{-2}$ and simulate the systems with the density matrix simulator of \textit{TensorCircuit} \cite{zhang2023tensorcircuitquantumsoftwareframework}.
All experiments were executed with $5$ different randomly chosen seeds. The results presented in Fig.~\ref{fig:thermal_state_res} show the average of those runs as well as the respective standard deviation thereof.
The plots illustrate that in both cases the training leads to very small infidelities for $d>1$ and the norm of the gradient magnitude for the model parameters does not decrease for larger $d$---despite the global form of the infidelity as a loss function.
In \cref{fig:thermal_state_qubits}, we examine the infidelity achieved at $d=2$ as a function of the number of training steps, for system sizes $n \in \{4,6,8\}$.
We see that infidelities of the order $10^{-4}$ can be achieved, illustrating that the ansatz has sufficient expressivity for representing the target thermal state.

Based on the observed convergence behavior and the magnitude of the loss function gradients, we may now test whether the ansatz can also work with a scalable loss function. More specifically, we employ 
McLachlan's variational principle \cite{mclachlan1964variationalsolutiontimedependentschrodinger} to realize an approximate imaginary time evolution \cite{zoufal2021variationalquantumboltzmannmachines, yuan2019theoryvariationalquantumsimulation} following $H_{\text{XY}}$ for times (which directly correspond to the inverse temperatures) $0.1$ and $0.25$.
The underlying method is described in more detail in Appendix~\ref{app:VarQITE}. 
At this point, we would only like to mention that the method is based on an ordinary differential equation (ODE) which is informed by McLachlan's variational principle and, as such, requires a resource scaling that is at most quadratic in the number of ansatz parameters.
Note further that in this set of experiments, each layer of the ansatz presented in Fig.~\ref{fig:thermal_state_setup} is supplemented by a parameterized local depolarizing channel
\begin{align}
    \mathcal{D}_\lambda: \rho \rightarrow (1-\lambda)\rho + \frac{\lambda}{2} I.
\end{align}

Now, all initial parameters are chosen at random from a uniform distribution $[0, 1]$, except for the parameters that control the application of the final depolarizing channels---these are chosen as $\lambda=1$ such that the initial state for variational imaginary time evolution corresponds to $\frac{I}{2^n}$.
\begin{figure}[t]
\subfloat[Infidelity]{
        \includegraphics[width=0.48\textwidth]{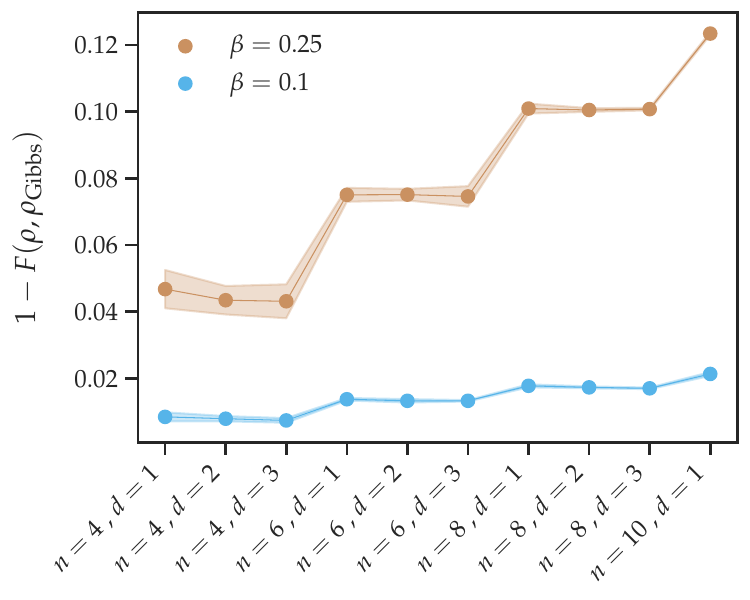}
}
\newline
 \subfloat[Energy Difference]{
        \includegraphics[width=0.48\textwidth]{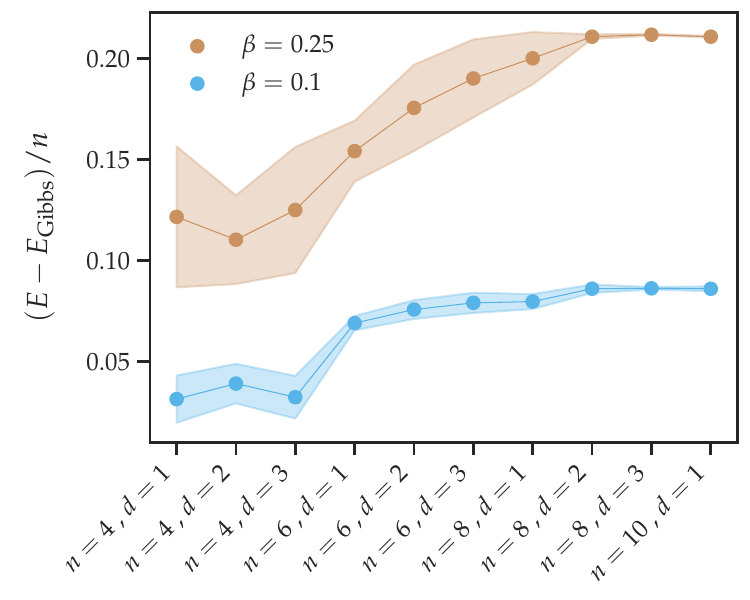}
}
\caption{The plots show the mean and standard deviation of (a) the infidelity and (b) the energy density difference with respect to the target state for the XY model using a scalable variational quantum imaginary time simulation for $\beta=0.1$ and $\beta=0.25$, for $n\in\left\{4, 6, 8, 10\right\}$ and $10$ random seeds per setting.}
\label{fig:thermal_state_qubits_vaQTE}
\end{figure}

\Cref{fig:thermal_state_qubits_vaQTE} presents the infidelity and the energy difference between the target Gibbs state and the state approximated with the DPQC ansatz using a forward Euler method for discrete time steps of size $0.01$.
The infidelity curve in Fig.~\ref{fig:thermal_state_qubits_vaQTE}(a)
illustrates that the method works better for the thermal state with $\beta=0.1$ compared to $\beta=0.25$---which is aligned with the fact that lower temperature states are usually more difficult to prepare. 
It is evident that the resulting infidelities are not as good as the ones achieved with the global infidelity loss function.
Given the performance of the infidelity based training results, we can conclude that this is not because of a lack of ansatz expressivity. 
Instead, it may be due to ODE-induced errors such as integration errors and time discretization as well as the fact that, unlike in the infidelity minimization, errors conducted at individual time steps directly accumulate.
Notably, one may further improve this methodology by using higher-order or implicit ODE solvers---which increases the required measurement resources---or by investigating the use of a regularization method targeted towards lowering the system energy.

\subsection{Caveats and criticisms}\label{ss:criticisms}
We note some caveats on these numerical experiments that may limit their applicability for other practical situations.
\begin{enumerate}
\item \textbf{Small-scale experiments:} We study relatively small system sizes, with a number of system qubits at most 12. In particular, this means that even an exponentially small gradient might not be very small.
\item \textbf{No shot noise:}
It should be noted that since we used a tensor network simulator to simplify the evaluation of expectation values, the results are devoid of shot noise. Hence, we could in principle resolve exponentially small gradients faithfully should they occur throughout the training.
\item \textbf{No hardware noise:} Using numerical simulations with tensor networks means that the results are not affected by the noise typically present in actual quantum hardware. If we were to run our experiments on quantum hardware, we would expect the results to be influenced by this hardware-induced noise.
\item \textbf{Scalability vs. performance trade-off in thermal state preparation: } The thermal state preparations are either based on a loss function that requires the evaluation of the fidelity or an ODE-based approach. Given that the former might require exponential measurement resources to estimate this quantity, the training pipeline is not scalable in its current form. While the latter approach is scalable in the sense that it may be realized with polynomial measurement resources, the resulting infidelities to the target states are significantly larger than the ones achieved with the infidelity training. This might also be due to the fact that there is no theoretical guarantee on closeness of the obtained state to the Gibbs state.
\end{enumerate}

We note a few reasons to be optimistic about DPQCs despite these caveats.
For example, note that Ref.~\cite{zhang2024absencebarrenplateausfinite} studied the perturbed toric code model at the same system size. 
Even at these system sizes, they noted that linear-depth architectures for the preparation of the toric code state failed on account of barren plateaus.
Furthermore, we note that as shown in \cref{fig:thermal_state_res}, the norm of the gradient for the DPQC architecture is not very small in general, raising the possibility that it can be meaningfully estimated.
There is also hope that in implementations of variational quantum algorithms, if hardware noise is consistent across runs, then its effect may be surmounted when implementing the algorithm \cite{sharma2020noiseresiliencevariationalquantum}.
Lastly, finding a suitable and scalable loss function for thermal state preparation with the property that the optimum is close to the target Gibbs state is still an open problem.
Improving the hyper-parameters in the investigated ODE based approach or testing additional scalable loss functions such as those suggested in Ref.~\cite{wang2021variationalquantumgibbsstate} could help make conclusive statements about the capabilities of our ansatz for thermal state preparation.
Therefore, despite the caveats we point out, our results in this section may be viewed as promising first steps towards establishing the general utility of DPQC architectures in practical situations.

\section{Classical hardness}\label{section:classical_hardness}

In the section above, we have seen evidence that, at least for small problem sizes, variational quantum algorithms using dynamic parameterized quantum circuits can feasibly provide meaningful results for interesting problems. In this section we address the question when the dynamic parameterized quantum circuits could be classically simulated.

\begin{figure}[htbp]
\includegraphics[width = 0.9\columnwidth]{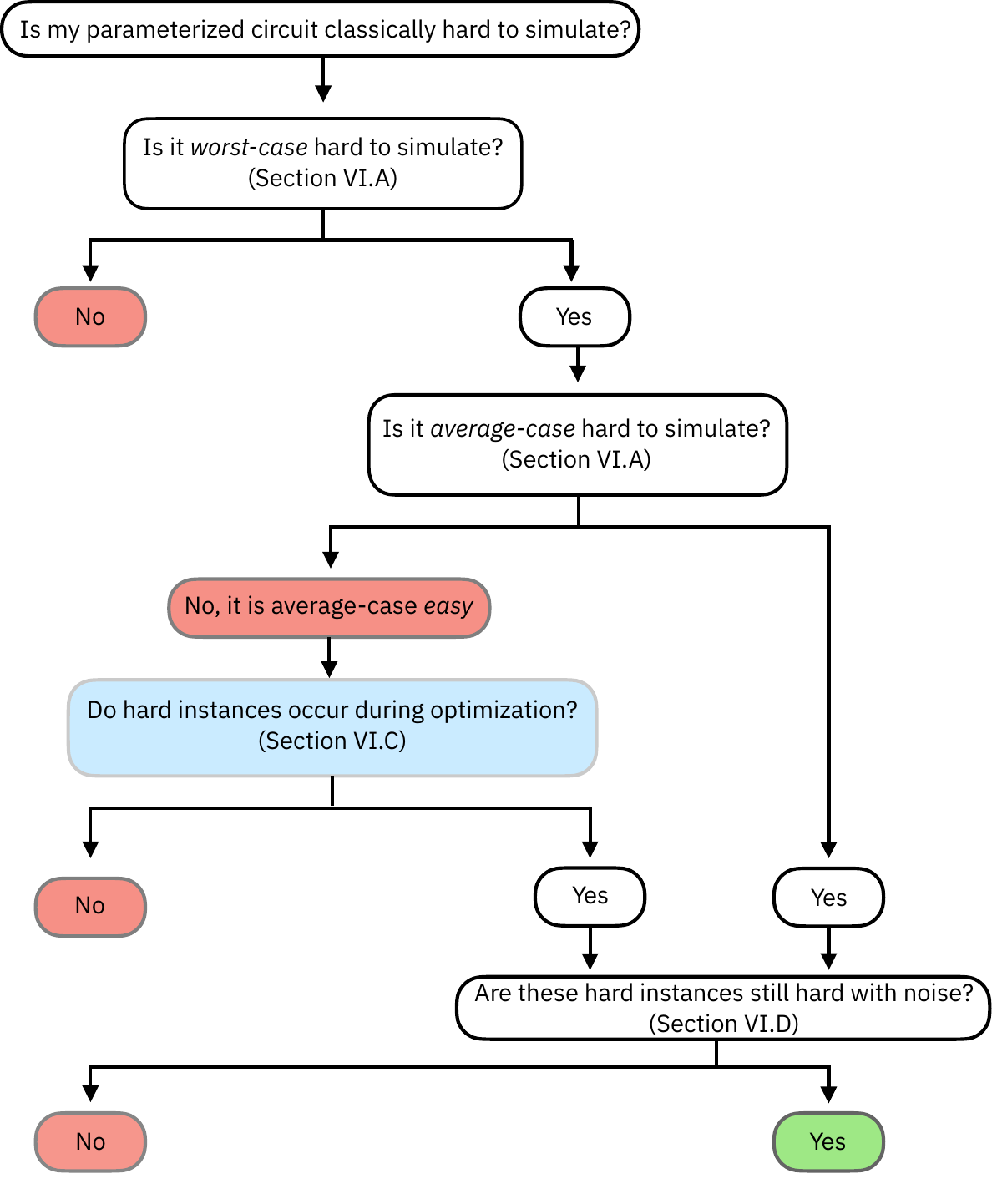}
\caption{\label{fig:classical_hardness} A methodology for determining whether a parameterized circuit class is classically hard to simulate. For the dynamic parameterized circuits considered in this work, we attempt to provide answers to the necessary questions in the sections indicated.
}
\end{figure}

In order to classically simulate \textit{a single iteration} of Algorithm~\ref{alg:variational_quantum_algorithm}, it is sufficient to be able to classically simulate:
\begin{enumerate}
\item The evaluation of the loss function $L_M(\btheta_i)$ for current circuit parameters $\btheta_i\in\Theta$.
\item The parameter update rule $\mathsf{PU}(\btheta_i) = \btheta_{i+1}$.
\end{enumerate}
Since for several gradient based parameter update rules, it is sufficient to be able to evaluate the loss function $L_M$ at $\btheta_i$ and small perturbations of $\btheta_i$~\cite{wierichs2022generalparametershiftrulesquantuma}, we restrict our attention to classical simulation of the loss function $L_M:\Theta\rightarrow\mathbb{R}$.
Moreover, we focus on the setting where the loss function can be calculated from the expectation value of local observables after running the circuit $\mathcal{C}({\btheta})\in\mathcal{C}$. Therefore, the question of classical simulation of the loss function reduces to whether one can obtain the expectation value of $O$, with respect to the output state of the circuit $\mathcal{C}({\btheta})$ given some local observable $O$ and circuit parameters $\btheta$.  

With this in mind, we proceed to analyze the classical simulability question, by following the methodology illustrated in Figure~\ref{fig:classical_hardness}. More specifically, we begin by first exploring in Section~\ref{ss:worst_case} whether or not DPQCs can be simulated in a \textit{worst-case} sense. As per the decision tree in Figure~\ref{fig:classical_hardness}, if the DPQC is \textit{not} worst-case hard to simulate, then by the arguments discussed above this simulation algorithm can be used to simulate any variational quantum algorithm using the DPQC. However, in Section~\ref{ss:worst_case} we leverage the expressivity of DPQC's to show that they are indeed worst-case hard to efficiently classically simulate, and that they therefore pass this first obstacle for classical simulations.

In light of this worst-case classical hardness, we then proceed to study the \textit{average-case} hardness of classical simulations, with respect to circuit instances drawn from some natural distribution over $\Theta$. Here we will show that, contrary to what one might hope, DPQCs are in fact \textit{easy} to simulate on average, via the low weight Pauli-path algorithm recently studied in Ref.~\cite{angrisani2024classicallyestimatingobservablesnoiseless} (see Ref.~\cite{rudolph2025paulipropagationcomputationalframework} for an overview). This average-case easiness of DPQCs then forces us to consider in Section~\ref{ss:hardness_in_optimization} the subtle question of whether or not instances which are hard to simulate might occur during the execution of variational quantum algorithms using DPQCs, for interesting and relevant problems. In particular, here we argue that there could exist problems for which hard-instances for classical simulation might occur during the execution of a DPQC-based variational quantum algorithm.

\subsection{Worst-case hardness}\label{ss:worst_case}

In this section we observe that the DPQC architectures are worst-case hard to simulate classically (under standard complexity theoretic assumptions). To this end, consider a DPQC architecture constructed in either one of the following two ways:

\begin{enumerate}
\item Start with any universal unitary parameterized circuit architecture, and add \emph{probabilistic} feedforward operations in such a way that ensures the feedforward distance of the resulting DPQC architecture is constant. 
\item Start with any universal unitary parameterized circuit architecture, and add \emph{deterministic} feedforward operations on ancilla qubits, together with entangling operations between system and ancilla qubits, in such a way that ensures the feedforward distance of the resulting DPQC architecture is constant.
\end{enumerate}

From Observations~\ref{obs:expressivity_informal_probabilistic} and~\ref{obs:expressivity_informal_deterministic} both DPQC architectures above will be at least as expressive as the unitary architecture from which one started, and therefore worst-case hard to simulate classically unless $\mathsf{BPP} = \mathsf{BQP}$. Additionally, under some additional easy to satisfy assumptions on the parameterized gates, both architectures will also be barren-plateau free via Theorem~\ref{thm:BP_DPQC_informal} and the assumed constant feedforward distance. Taken together we have:

\begin{center}
\textit{There exist DPQC architectures that are \emph{both} worst-case hard to classically simulate efficiently, \emph{and} barren-plateau free.}
\end{center}

We stress that one could straightforwardly make a similar statement about worst-case hardness and absence of BPs by simply considering universal unitary architectures with a fixed (or highly constrained) initialization strategy---eg, initializing to the identity. However, as optimization with such an initialization strategy would always start from the same region (or point) in the cost landscape, one would expect such a strategy to be disadvantageous from an optimization perspective. Ultimately, however, large-scale numerical experiments are necessary to distinguish the practical utility of these two approaches to balancing expressivity and trainability.

\subsection{Average-case hardness (easiness)}\label{ss:average_case}

As illustrated in \cref{fig:classical_hardness}, in order to have any potential of providing utility via quantum devices, it is necessary but not sufficient for a parameterized quantum circuit architecture to be worst-case hard to classically simulate efficiently. Indeed, it could be the case that there exist hard circuit instances, but that these are never encountered during the execution of a variational quantum algorithm, which can therefore be efficiently classically simulated despite the worst-case hardness. 

Unfortunately, contrary to what one might hope, a large class of DPQC architectures---including the ones we use for numerical experiments in Section~\ref{section:utility}---are in fact average-case \textit{easy} to efficiently classically simulate, via the low weight Pauli-path algorithm recently studied in Refs.~\cite{gao2018efficientclassicalsimulationnoisy,aharonov2022polynomialtimeclassicalalgorithmnoisy,shao2023simulatingquantummeanvalues,fontana2023classicalsimulationsnoisyvariational,schuster2024polynomialtimeclassicalalgorithmnoisy,angrisani2024classicallyestimatingobservablesnoiseless}. More specifically, we can make the following observation:

\begin{obs}[Average-case easiness of DPQC simulation via low-weight Pauli paths] \label{obs:avg_case_easy} Consider DPQC architectures in which (a) the only nonunitary operations are single-qubit feedforward operations as studied in \cref{section:ansatz}, (b) there are at most a polynomial number of such feedforward operations, and (c) the parameterized gates are ``locally scrambling'' as defined in Ref.~\cite{angrisani2024classicallyestimatingobservablesnoiseless}, these DPQC architectures can be efficiently classically simulated with high probability with respect to the distribution of parameters via the low weight Pauli paths algorithm of Ref.~\cite{angrisani2024classicallyestimatingobservablesnoiseless}.
\end{obs}

While the DPQC architectures from Observation~\ref{obs:avg_case_easy} are not explicitly considered in Ref.~\cite{angrisani2024classicallyestimatingobservablesnoiseless}, one can extract the observation above from the following reasoning. Firstly, as noted in Appendix~\ref{app:definitions}, we can replace each nonunitary operation with an ancilla qubit and some unitary interactions between this ancilla qubit and the wire on which the feedforward operation takes place---i.e. we can replace the feedforward with a unitary ``feedforward gadget'' in the purified picture, at the cost of at most two ancillas per feedforward operation. Given that the local scrambling property is satisfied by assumption, we can now apply the algorithm of Ref.~\cite{angrisani2024classicallyestimatingobservablesnoiseless} to this circuit. The requirement that there are at most a polynomial number of feedforward operations is to ensure that at most polynomially many ancilla qubits are added. Next we note that each feedforward gadget does not increase the weight of a propagated Pauli string on the system qubits. While the Pauli strings can spread on to the ancilla qubits, since only identity operations happen on ancilla qubits after a gadget, they can be easily handled.
Also see Ref.~\cite{bermejo2024quantumconvolutionalneuralnetworks} for a discussion.

At first glance, one might think that Observation~\ref{obs:avg_case_easy} renders DPQC architectures unsuitable for offering some advantage over classical methods. However, this is not the case. More specifically, as shown in \cref{fig:classical_hardness}, even if a circuit architecture is average-case easy to simulate, it could still be the case that there exist meaningful and relevant problems for which variational quantum algorithms encounter hard instances during execution, and therefore \emph{cannot} be efficiently classically simulated, despite the average-case simulability.

Finally, before moving on to the question of whether or not hard instances occur during optimization, we make some brief comments on the limitations of the low-weight Pauli paths algorithm from Ref.~\cite{angrisani2024classicallyestimatingobservablesnoiseless}. In particular, we note that this algorithm works in the Heisenberg picture, by backwards evolving an observable through the quantum circuit (while keeping track of a suitably truncated Pauli expansion of the observable). As such, while this algorithm can provide expectation values of observables, it cannot provide a succinct representation of the output state of the quantum circuit.

\subsection{Do hard instances occur during optimization?}\label{ss:hardness_in_optimization}
In this section, we are particularly interested in the case when:
\begin{enumerate}
\item The PQC architecture is barren-plateau free.
\item The PQC architecture is also efficiently classically simulable on average.
\item The VQA is only allowed to make a polynomial number of measurements and run for polynomial time.
\end{enumerate}

We note that as a consequence of the third assumption, we can also restrict ourselves to VQAs which only ever encounter instances with non-negligible gradients. In particular, if a VQA reaches a circuit instance with a negligible gradient, then there are two options: The first option is to take enough measurements to resolve the gradient to an accuracy which allows for a meaningful update, but this is not possible with the polynomial constraint. The second option is to use only a polynomial number of measurements to estimate the gradient. In this case however the estimated gradient is essentially random, and the subsequent optimization step is essentially a guess. VQAs with such behaviour are unlikely to succeed in optimization, and we can exclude them from our analysis.

As such, in order to understand whether or not a meaningful polynomially constrained VQA can reach hard-to-simulate instances, we would like to characterize the classical simulability of those circuit instances with non-negligible gradients.
In particular, if \emph{all} such instances are efficiently classically simulable, then clearly, under our assumptions on the VQA, no classically hard instances can occur during polynomial-time optimization.

\begin{figure}[t]
\includegraphics[width = \columnwidth]{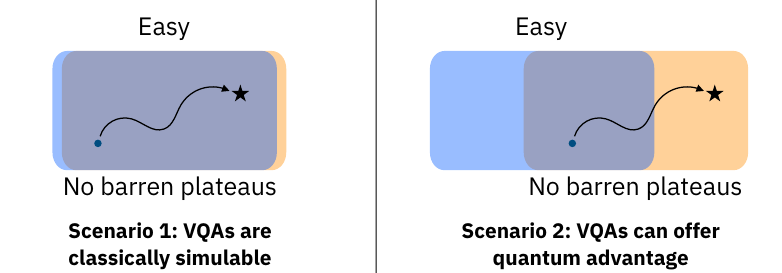}
\caption{\label{fig:optimization_possibilities} An illustration of two different scenarios for the simulability of variational quantum algorithms. The set of instances visited by a quantum algorithm are depicted by an arrow.
In the first scenario, every instance visited by the quantum algorithm remains easy to simulate, providing no quantum advantage.
In the second scenario, the variational quantum algorithm visits some instances that are hard to classically simulate, potentially offering a quantum advantage.}
\end{figure}

For the case of interest in which the PQC architecture is both barren-plateau free, and efficiently classically simulable on average, the situation is illustrated in \cref{fig:optimization_possibilities}.
In particular, the absence of barren plateau result implies that \emph{some} large fraction of instances do have non-negligible gradients, and are therefore reachable under our assumptions. Simultaneously, the average-case simulation result implies that \textit{some other} large fraction of instances can be classically efficiently simulated.
The important question is then how these fractions of instances are related to each other, and there are two distinct possibilities:

\begin{enumerate}
\item \textbf{Possibility 1:} All instances that have non-negligible gradients (and are therefore reachable) are easy to efficiently simulate classically.
\item \textbf{Possibility 2:} There exist instances that have non-negligible gradients (and are therefore reachable), which are hard to efficiently simulate classically.
\end{enumerate}

Note that if the first possibility is true, then this rules out the utility of running variational quantum algorithms with these architectures for any problem where local expectation values of the target state are the desired outcome. If possibility 2 is true, this implies that VQAs \textit{could} potentially reach hard instances during optimization, but \textit{does not} imply that they {will} always reach hard instances in practice. Which of these two possibilities is the case?

On the one hand, one could reasonably conjecture that possibility 1 is the case, because the mechanism that makes a particular instance efficient to classically simulate (low-weight Pauli path) is very similar to the mechanism which leads to non-negligible gradients. We view formalizing this connection as an important and immediate open problem. Additionally, we note that Ref.~\cite{cerezo2023doesprovableabsencebarren} has accumulated evidence for the truth of this conjecture for certain PQC architectures, by reasoning about properties of an architecture's dynamical lie algebra (DLA), which is also intimately linked to both non-negligible gradients and efficient classical simulations. However, this DLA picture is not immediately applicable to the dynamic parameterized quantum circuit architectures we study in this work.
Understanding the relation between barren plateaus in the stat-mech picture and the ensembles of average-case easy circuits for classical algorithms is crucial to the larger project of understanding the extent to which ``absence of barren plateaus implies classical simulability''~\cite{cerezo2023doesprovableabsencebarren}.

On the other hand, the numerical experiments we presented in Section~\ref{section:utility} provide some evidence that possibility 2 might be true. More specifically, in the discussion of worst-case hardness given in Section~\ref{ss:worst_case}, we argued that DPQC architectures are worst-case hard to simulate because they contain poly-depth \textit{unitary} architectures (which produce pure states). A skeptic might argue that, because DPQC architectures are generically nonunitary by design, VQAs using DPQC  architectures never converge in practice to unitary circuits, and therefore never have a chance of reaching the instances we used to prove worst-case hardness. However, we have seen in Section~\ref{section:utility} that for a meaningful problem, VQAs with DPQC architectures \textit{can} indeed converge to unitary circuits which prepare pure states! 
This constitutes evidence that VQAs with DPQCs might be able to reach instances that are structurally more similar to the (pure) worst-case hard instances than (impure) generic states for which average-case easiness holds\footnote{Of course, not all unitary instances are hard to simulate.}.

In summary, it is as of yet unclear whether possibility 1 or 2 is the case when using variational quantum algorithms with dynamic parameterized quantum circuit architectures for meaningful problems. Above, we have sketched two potential routes for resolving this issue, which we believe are concrete and important avenues for future research.

\section{Outlook and conclusions}\label{section:outlook}
This work focuses on a \textit{dynamic} parameterized quantum circuit class that is constructed by unitary gates, intermediate measurements, and feedforward operations. 
The study was motivated by the question of whether this circuit class provides a \textit{good} ansatz for variational quantum algorithms. This presents a particularly important question given that most known variational quantum ansatz classes suffer from drawbacks that result in them being either 
untrainable, classically simulable, or insufficiently expressive.
We present an ansatz that may offer an avenue for training scalable variational quantum algorithms, capable of finding good solutions to interesting problems, and challenging for classical computers.

Our theoretical analysis indicates that the studied DPQC class corresponds to a promising model for variational quantum algorithms. Evidence for this intuition comes from  the fact that DPQC models can be both expressive and free from barren plateaus. Specifically, we have shown that one can construct DPQC models which contain arbitrarily deep unitary quantum circuits---and are therefore worst case hard to simulate classically---while at the same time not suffering from exponentially vanishing gradients.  These models can interpolate smoothly between being highly expressive and barren-plateau free, making them a convenient design choice. We stress, however, that absence of barren plateaus does not imply trainability, and it remains an open question whether worst-case hard instances for classical simulation can be encountered during the training of DPQC models.

The potential of the DPQC architecture is additionally supported by numerical experiments on ground state and thermal state preparation problems.
While the numerical results demonstrate the capability to learn interesting states, it remains an open task for future research to investigate whether the observed good training behavior persists for larger system sizes.
Additionally, in order to reliably argue about the capabilities of our ansatz for thermal state preparation, it remains to be investigated whether the experiments may be reproduced with high fidelity using scalable loss functions \cite{zoufal2021variationalquantumboltzmannmachines, wang2021variationalquantumgibbsstate}.

The DPQC architecture provides a promising model that is worth studying in future research. Answering the open questions about trainability and classical simulability outlined above are crucial steps towards understanding its \textit{quantum utility}.

\textbf{Author contributions}.
AD: Conceptualization, Software, Formal analysis, Investigation, Writing--Original Draft, Supervision.
MH: Formal analysis, Writing--Review and Editing.
SN: Writing--Review and Editing.
KS: Writing--Review and Editing.
RS: Conceptualization, Writing--Original Draft, Writing--Review and Editing, Supervision.
CZ: Conceptualization, Software, Investigation, Writing--Review and Editing.

\textbf{Acknowledgments.}
We acknowledge discussions with Zo\"e Holmes, Manuel S.~Rudolph, and Armando Angrisani. We are also grateful to Marius Krumm for technical discussions of our proofs.

\appendix

\addcontentsline{toc}{section}{Appendix}
\section{Definitions and Notation}\label{app:definitions}
We set up here a few notations and recap some definitions from the main text.
We work with quantum circuits composed of two-qubit gates over $n$ qubits with a total depth $d$.
The total number of gates is denoted $m$.
We denote the space of linear operators acting on $n$-qubits $\mathcal{L}(n)$, and the space of Hermitian observables on $n$ qubits $\mathrm{Herm}(n)$.
The space of valid density matrices on $n$ qubits, or equivalently, the space of positive semidefinite trace 1 Hermitian matrices, is $\mathrm{Dens}(n) \subset \mathrm{Herm}(n)$.
We consider circuits with nonunitary operations, which we describe through channels, denoted by calligraphic letters such as $\mathcal{U}(\rho_{\mathrm{init}}): \mathrm{Dens}(n) \to \mathrm{Dens}(n)$.
Denote by $\mathbb{P}_n$ the set of all Pauli observables on $n$ qubits and by $\bm{\alpha}$ a particular Pauli observable.

We denote bitstrings in boldface, e.g.~$\bm{x} \in \{0,1\}^n$.
We use subscripts to denote subsets of bitstrings, for example $\bm{x}_j$ denotes a single component of $\bm{x}$, while $\bm{x}_A$ denotes the bitstring restricted to components in a subset $A \subseteq [n]$.

\begin{defn}[Haar measure]
    The Haar measure $\mathcal{H}$ on the unitary group $\mathrm{U}(N)$ is the unique probability measure that is both left and right invariant under the group action. That is, for any integrable function $f$ and for all $V\in \mathrm{U}(N)$, it holds that
    \begin{align}
        \int_{U\in \mathrm{U}(N)} f(U)\, d \mathcal{H} (U) = \int_{U\in \mathrm{U}(N)} f(UV)\, d \mathcal{H} (U) \nonumber\\
        = \int_{U\in \mathrm{U}(N)} f(VU) \,d \mathcal{H} (U) \,.
    \end{align}
\end{defn}

In this work, we are interested in systems of $n$ qubits such that we consider $N=2^n$. 

\begin{defn}[Global unitary $t$-design]
    Let $\mathcal{E}$ be an ensemble of $n$-qubit unitaries. Then, $\mathcal{E}$ is a unitary $t$-design if and only if for all $O\in \mathcal{L}(n)^{\otimes t}$, it holds that
    \begin{align}
    \label{eq:def-t-design}
        \mathbb{E}_{V\sim \mathcal{E}} [V^{\otimes t} O V^{\dagger \otimes t} ] =  \mathbb{E}_{V\sim \mathcal{H}} [V^{\otimes t} O V^{\dagger \otimes t} ] \,.
    \end{align} 
\end{defn}

In this work, we study a more relaxed notion of designs called \emph{local} designs.
Here, we only require that each $k$-local operation is drawn randomly from a $k$ design.
In the following, for any Haar average like the one on the RHS of \Cref{eq:def-t-design}, we usually omit explicitly mentioning the measure $\mathcal{H}$ and simply write $\mathbb{E}_{V}$.

\begin{defn}[Locally scrambling ensemble.]
Consider a distribution over quantum circuits $\mathcal{D}$.
The ensemble $\mathcal{D}$ is \emph{locally scrambling} if the distribution is invariant to the insertion of random single-qubit gates $\bm V = (V_1, V_2, \ldots V_m)$ drawn from a 2-design $\mathcal{E}$.
Mathematically, 
\begin{align}
 \Pr_{\mathcal{D}} [C] = \Pr_{\mathcal{D}} [C_{\bm{V}}], 
\end{align}
where $C_{\bm{V}}$ is the circuit obtained by interspersing single-qubit gates $V_1 \sim \mathcal{E}, V_2 \sim \mathcal{E}, \ldots V_m \sim \mathcal{E}$ after each gate of $C$.
\end{defn}

An oft-occurring calculation in the study of random circuits is that of $t$ copies of a state, which means studying the object where the initial state is $\rho_{\mathrm{init}}^{\otimes t}$ and applying copies of the channel on the initial state: $\mathcal{U}^{\otimes t}(\rho_{\mathrm{init}}^{\otimes t})$.
Let $\mathbb{S}_t$ be the permutation group on $t$ objects labeled by integers $[t]$, with group elements $\sigma : [t] \to [t]$.
Consider a representation of $\mathbb{S}_t$ where each permutation $\sigma$ is associated with the map that permutes copies of quantum states through conjugation, i.e.\ $R(\sigma) (\cdot) R(\sigma)^\dag : \mathrm{Dens}(n \times t) \to \mathrm{Dens}(n \times t)$.
The symmetric subspace $P^t$ over operators on $n\times t$ qubits is defined by operators $\{O: R(\sigma) O R(\sigma)^\dag = O\}$ invariant under permutations $\sigma \in \mathbb{S}_t$.

For $t=2$, the group $\mathbb{S}_2$ has the elements identity $e$ and SWAP $s$, satisfying
\begin{align}
e(1) &= 1; \quad s(1) = 2 \\
e(2) &= 2; \quad s(2) = 1.
\end{align}
The representation of these elements for $1$-qubit density matrices is $R: \mathrm{Dens}(2) \to \mathrm{Dens}(2)$, which has elements we denote through tensor network diagrams as:
\begin{align}
R(e) &= I  = 
\begin{quantikz}
 & & &\\
 & & &
\end{quantikz} \\
R(s) &= S =
\begin{quantikz}
 &\permute{2,1} \\
 &
\end{quantikz}.
\end{align}

The first observation that enables the stat-mech model mapping is the ``replica trick'':
\begin{align}
(\Tr \rho O)^t &= \Tr \rho^{\otimes t}O^{\otimes t} \nonumber \\
\implies \mathbb{E}_\mathcal{B}\left[(\Tr \rho O)^t\right] &= \mathbb{E}_\mathcal{B} \left[\Tr \rho^{\otimes t}O^{\otimes t}\right].
\end{align}
Exchanging the order of the expectation and the trace, we get
\begin{align}
\mathbb{E}_\mathcal{B}\left[(\Tr \rho O)^t\right] = \Tr \mathbb{E}_\mathcal{B} \left[\rho^{\otimes t} \right] O^{\otimes t}.
\end{align}
Thus, it suffices to know the average $t$-copy density matrix $\mathbb{E}_\mathcal{B} \left[\rho^{\otimes t} \right]$.
For any state $\rho$ and some distribution over unitaries $V$, we call the quantity $\mathbb{E}_{V}[\rho^{\otimes 2}]$ the 2-copy average state corresponding to $\rho$.
This is the fundamental object of interest for calculating second moment quantities, which we denote by $\bar{\rho}$.

The next basic fact we need is that performing the Haar-average over single-qubit gates $\mathbb{E}_V[V^{\otimes t} A^{\otimes t} {V^\dag}^{\otimes t}] $ for any single-qubit operator $A \in \mathrm{Herm}(1)$  projects it down to the symmetric subspace over $t$ copies:
\begin{align}
\mathbb{E}_V[V^{\otimes t} A^{\otimes t} {V^\dag}^{\otimes t}] \in P^t,
\end{align} 
where $P^t$ is the symmetric subspace $\{O: \sigma O \sigma^\dag = O\}$ defined by operators invariant under permutations $\sigma \in \mathbb{S}_t$.
When $t =1$, the above reduces to
\begin{align}\label{trace_example}
\mathbb{E}_V[VA {V^\dag}] = \Tr A \frac{I}{2}.
\end{align}
For $t=2$, we use as basis elements for the symmetric subspace the $4\times 4$ identity gate $I$ and the SWAP gate
\begin{align}
S =
\begin{pmatrix}
    1 & 0 & 0 & 0 \\
    0 & 0 & 1 & 0 \\
    0 & 1 & 0 & 0 \\
    0 & 0 & 0 & 1
\end{pmatrix}.
\end{align}
Note that $S^2 = I$ and $\Tr I = 4$, $\Tr S = 2$.
We also use the trace-1 normalized versions of this basis set, denoted in typewriter font: $\mathtt{I}:= \frac{I}{4}$, and $\mathtt{S}:=\frac{S}{2}$.
We obtain
\begin{align}
  \mathbb{E}_V[V^{\otimes 2} A^{\otimes 2} {V^\dag}^{\otimes 2}] = a I + b S, \label{eq_statmech_span_i_s}
\end{align}
where $a$ and $b$ can be obtained by solving the linear equations
$\Tr A^{\otimes 2} = (\Tr A)^2 = 4 a + 2 b$, $\Tr A^{\otimes 2} S = \Tr (A^2) = 2a + 4b$, giving:
\begin{align}
  \mathbb{E}_V[V^{\otimes 2} A^{\otimes 2} {V^\dag}^{\otimes 2}] = 
 \frac{(\Tr A)^{2} - \frac{1}{2} \Tr (A^2)}{3} I + \nonumber \\
 \frac{\Tr (A^2)- \frac{1}{2} (\Tr A)^{2}}{3} S. \label{eq_singleqbthaar}
\end{align}

For a single-qubit Pauli observable $\bm{\alpha} \neq I$, this simplifies to:
\begin{align}
 \mathbb{E}_V[V^{\otimes 2} \bm{\alpha}^{\otimes 2} {V^\dag}^{\otimes 2}] = \frac{2S}{3} - \frac{I}{3} = \frac{4}{3}\left(\mathtt{S} - \mathtt{I}\right). \label{eq_singleqbtpauli}
\end{align}

We also note the relations for $n$-qubit Pauli observables:
\begin{align}
\Tr (\bm{\alpha} \otimes \bm{\beta}) = \delta_{\bm{\alpha},\bm{I}} \delta_{\bm{\beta},\bm{I}} \\
\Tr (\bm{\alpha} \otimes \bm{\beta} \cdot S) = \delta_{\bm{\alpha},\bm{\beta}}, \label{eq_product_paulis}
\end{align}
where $\bm{I}$ denotes the $n$-qubit identity Pauli word.

\section{The stat-mech model}\label{app:SM}
We derive the basics of the technique here.
Our discussion closely follows that of Refs.~\cite{hunter-jones2019unitarydesignsstatisticalmechanics,dalzell2022randomquantumcircuitsanticoncentrate,dalzell2021randomquantumcircuitstransform,napp2022quantifyingbarrenplateauphenomenon}.
For second moment observables, the stat-mech model involves computing and keeping track of the two-copy average state.
For an $n$-qubit state $\rho$, the two-copy average state resides in $\text{span}(\{\mathtt{I},\mathtt{S}\}^n)$.
Let us denote by a bitstring $\bm{x}$ whether we pick out an operator $\mathtt{I}$ ($\bm{x}_j = 0$) or $\mathtt{S}$ ($\bm{x}_j=1$) at site $j$.
In other words, let $\mathtt{T}_{x_j} := \mathtt{I}^{1-\bm{x}_j} \cdot \mathtt{S}^{\bm{x}_j}$ and $\mathtt{T}_{\bm{x}} := \prod_{j}^{n} \mathtt{T}_{x_j}$.
We can then write the two-copy average state as 
\begin{align}
  \bar{\rho} &= \sum_{\bm{x} \in \{0,1\}^n } c_{\bm{x}} \mathtt{I}^{1-\bm{x}_1} \cdot \mathtt{S}^{\bm{x}_1} \otimes \ldots \otimes \mathtt{I}^{1-\bm{x}_n} \cdot \mathtt{S}^{\bm{x}_n} \nonumber 
  \\
  & = \sum_{\bm{x}} c_{\bm{x}} \mathtt{T}_{\bm{x}},
\end{align}
where $\sum_{\bm{x}} c_{\bm{x}} = 1$.
In this way, we can also associate a two-copy average state with a (quasi)-probability distribution $\mathcal{D}$ over bitstrings $\bm{x}\in \left\{ 0,1 \right\}^n$.

Lastly, we use the fact that the channels are drawn independently of each other, which lets us perform the average over many channels in sequence:

\begin{defn}
For a circuit with initial state $\rho_0$, consider $\rho^t(\btheta) = \mathcal{U}_t(\btheta) \dots \mathcal{U}_2(\btheta) \circ \mathcal{U}_1(\btheta)(\rho_0)$

the state of the system at time step $t$ when fixing the parameters of the circuit $\btheta$.
The two-copy average state at time $t$, denoted $\bar{\rho}^t$, is the average two-copy state of $\rho^t(\btheta)$ according to the distribution $\btheta \sim \mathcal{D}_p$ over parameters, or equivalently, according to the distribution $\mathcal{B}$ over channels.
\end{defn}
It is easy to see that this average $\mathbb{E}_{\mathcal{B}}$ does not depend on the operations occurring after time step~$t$ since they are not in the backwards light cone of $\rho^t$.

\begin{claim}
The average 2-copy state $\bar{\rho}^t$ at time $t$ can be obtained from the average $\bar{\rho}^{t-1}$ at time $t-1$.
\end{claim}
This is because we have
\begin{align}
\mathbb{E}_{\mathcal{B}} \left[ (\mathcal{C}^t (\rho_0))^{\otimes 2} \right] &= \mathbb{E}_{\mathcal{U}_1, \mathcal{U}_2,\ldots \mathcal{U}_t} \left[ ( \mathcal{U}_t \ldots \mathcal{U}_2 \circ \mathcal{U}_1(\rho_{0}))^{\otimes 2} \right] \\
& = \mathbb{E}_{\mathcal{U}_t} \mathbb{E}_{\mathcal{U}_1, \mathcal{U}_2,\ldots \mathcal{U}_{t-1}} \left[ \left(\mathcal{U}_t (\rho^{t-1})\right)^{\otimes 2} \right] \\
& = \mathbb{E}_{\mathcal{U}_t} \left[ \mathcal{U}_t^{\otimes 2} (\bar{\rho}^{t-1}) \right].
\end{align}
Therefore, in order to get a handle on properties of $\bar{\rho}_d$, the 2-copy average of the circuit output, it suffices to keep track of the average state $\bar{\rho}^t$ in time.

\begin{claim}[Linearity]
For a two-copy average state $\bar{\rho}$, we can define the reduced density matrix in the usual way by tracing out the appropriate subregion over both copies.
By linearity, this coincides with the two-copy average of the reduced density matrix, i.e.\
$\mathbb{E}_{\mathcal{B}} [(\Tr_A \rho)^{\otimes 2}] = \Tr_{A\times A} \mathbb{E}_\mathcal{B} \left[ \rho^{\otimes 2} \right]$.
\end{claim}

These two properties imply that we can obtain a description of the 2-copy average state at time $t+1$ from the one at time $t$ using \emph{local} update rules.
Crucially, it suffices to understand the map $\bar{\rho}^t \mapsto \mathbb{E}_{V_1, V_2}[(V_1 V_2 \otimes V_1 V_2)\ \mathcal{U} \otimes \mathcal{U}(\bar{\rho}^t) (V_1 V_2 \otimes V_1 V_2)^\dag]$ for Haar-random single-qubit gates $V_1 V_2$ and various channels $\mathcal{U}$.
We will restrict our attention to channels acting on at most two qubits at a time.
Suppose the channel $\mathcal{U}$ has the Kraus form $\mathcal{U}(\cdot) = \sum_i K_l(\cdot) K_l^\dag $.
We write this out in tensor notation as
\begin{align}\label{eq:local_update_rule}
\bar{\rho}^{t +1} = 
\mathbb{E}_{V_1, V_2} \sum_{l,m}
\resizebox{0.7\columnwidth}{!}{\begin{quantikz}[row sep={1cm,between origins},transparent]
& \gate[][8mm][8mm]{V_1}  &\gate[2][1cm][9mm]{K_l} & \gate[4]{\bar{\rho}^t} & \gate[2][1cm][9mm]{K_l^\dagger} & \gate[][8mm][8mm]{V_1^\dagger} & \\
& \gate[][8mm][8mm]{V_2}  &   &  &  & \gate[][8mm][8mm]{V_2^\dagger} & \\
& \gate[][8mm][8mm]{V_1}  &  \gate[2][1cm][9mm]{K_m} & & \gate[2][1cm][9mm]{K_m^\dag}   &\gate[][8mm][8mm]{V_1^\dagger} & \\
& \gate[][8mm][8mm]{V_2}   &  & & &   \gate[][8mm][8mm]{V_2^\dagger} &
\end{quantikz},}
\end{align}
where the expressions are read left to right (i.e.\ time flows from center out to each side).
The third and fourth qubit lines are the copies of the first and second, and $K_l$ acts as a two-qubit operator on qubits 1 and 2 or their copies 3 and 4.

\begin{claim}[Stat-mech update rule]
Consider a two-qubit operation $\mathcal{U}$.
Let the reduced two-copy average state on the qubits at time $t$ be $\bar{\rho}^t = a \mathtt{II} + b \mathtt{IS} + c \mathtt{SI} + d \mathtt{SS}$.
Then the reduced two-qubit average state after application of the channel satisfies  $\bar{\rho}^{t+1} = \mathbb{E}_{V_1, V_2}[(V_1 V_2 \otimes V_1 V_2)\ (\mathcal{U} \otimes \mathcal{U})(\bar{\rho}^t) (V_1 V_2 \otimes V_1 V_2)^\dag] = a'\mathtt{II} + b' \mathtt{IS} + c' \mathtt{SI} + d' \mathtt{SS}$, where
\begin{align}
\begin{pmatrix}
  a' \\
  b' \\
  c' \\
  d'
\end{pmatrix} = T 
\begin{pmatrix}
a \\
b\\
c\\
d
\end{pmatrix}
\end{align}
for a $4 \times 4$ stochastic matrix $T$.
We call $T$ the transfer matrix.
\end{claim}

We now state some results on the transfer matrices for some common operations.
Reference \cite{ware2023sharpphasetransitionlinear} derived general stat-mech rules for fixed two-qubit unitaries in terms of their local unitary invariants, the entangling power and the swapping power.

\begin{claim}[Transfer matrices of some gates]
For a two-qubit gate, $T$ takes the form
\begin{align}
T =
\begin{pmatrix}
1 & \frac{4\alpha}{5} & \frac{4\alpha}{5} & 0 \\
0 & 1-\alpha - \beta & \beta & 0 \\
0 & \beta & 1-\alpha  - \beta & 0 \\
0 & \frac{\alpha}{5} & \frac{\alpha}{5} & 1
\end{pmatrix}, \label{eq_twoqubit_transfermat}
\end{align}
with $\alpha \in \left[ 0, \frac{10}{9} \right]; \beta \in \left[ -\frac{\alpha}{5}, 1- \frac{4\alpha}{5} \right]$, subject to the constraint $\beta + \frac{\alpha}{5} \leq \left( \beta + \frac{\alpha}{2} \right)^2$.
Furthermore, for Haar-random two-qubit gates, we have $\alpha = 1, \beta=0$.
\end{claim}
These parameters can be derived from a calculation analogous to \cref{eq_statmech_span_i_s}.
Suppose that the gate acts on qubits $i,j$.
We expand $\mathcal{U}( \Tr_{[n] \backslash \{i,j\}} \bar{\rho}^{t})$ in the basis\footnote{Note that it suffices to only consider the reduced density matrices supported on the qubits the gate acts on.} $\{\mathtt{I}, \mathtt{S}\}^2$ and use linearity to infer the parameters in $T$.

We will restrict our attention to transfer matrices $T$ that have nonnegative entries.
For the two-qubit gate transfer matrix in \cref{eq_twoqubit_transfermat}, this corresponds to requiring $\beta \geq 0$, $\alpha \leq 1-\beta$.
Recall that the case of Haar-random two-qubit gates is covered by setting $\beta = 0$, $\alpha = 1$ and is included in this analysis.

\subsection{Dynamics of the stat-mech model}
We now interpret the dynamics of the two-copy average state as a random walk over bitstrings.

First, we calculate the two-copy average state of any product initial state {$\overline{\rho}_0$}, where we average over a single layer of single-qubit gates on all qubits:
\begin{align}
 \bar{\rho} = \left(a\mathtt{I} + b\mathtt{S}\right)^{\otimes n},
\end{align}
where
$a + b = 1$ and $\Tr \left(a\mathtt{I} + b\mathtt{S}\right)\cdot S = 1$, giving us $\frac{a}{2} + 2b = 1$.
Therefore, $a= \frac{2}{3}, b=\frac{1}{3}$, meaning 
\begin{align}
\bar{\rho} =\left( \frac{2}{3} \mathtt{I} + \frac{1}{3} \mathtt{S} \right)^n.
\end{align}
This two-copy state $\bar{\rho}$ can be interpreted as a probability distribution on strings $\bm{x} \in \{0,1\}^n$ by taking $\Pr[\bm{x}] = \prod_{i=1}^n \left( \frac{2}{3}  \right)^{1-\bm{x}_i} \left( \frac{1}{3} \right)^{\bm{x}_i }$.
Denote this distribution $\mathcal{X}_0$.

\begin{claim}
There is a one-to-one correspondence between a two-copy average state $\bar{\rho}^t$ and its associated distribution $\mathcal{X}^t$ over $n$-bit strings.
\end{claim}
This claim follows because we can write the vector of probabilities $p^t_{\bm{x}} := \Pr_{\mathcal{X}^t}[\bm{x}]$ as
\begin{align}
 p^t_{\bm{x}} = \left(T_t \ldots T_2 T_1 p^0\right)_{\bm{x}}
\end{align}
for a sequence of appropriate transfer matrices corresponding to that of the sequence of operations $\mathcal{U}_1,\ldots \mathcal{U}_t$.
The transfer matrices have nonnegative entries, and so does $\bm{p}^0$, implying $\bm{p}^t$ has nonnegative entries as well.
Furthermore,
\begin{align}
 \sum_{\bm{x}} {p}^t_{\bm{x}} = \sum_{\bm{x},\bm{y}} {(T_t)}_{\bm{x},\bm{y}} {p}^{t-1}_{\bm{y}}
 = \sum_{\bm{y}} {p}^{t-1}_{\bm{y}}
\end{align}
since the matrix $T_{t}$ is stochastic.
Proceeding inductively, we get $\sum_{\bm{x}} {p}^t_{\bm{x}} = \sum_{\bm{x}} {p}^0_{\bm{x}} = 1$.
Therefore, ${p}^t $ is a probability distribution.
From this claim, we can view the dynamics of the two-copy average state as a random walk over bitstring configurations $\bm{x} \in \{\mathtt{I}, \mathtt{S}\}^n$.
The statistical properties of the classical random walk entirely determine the dynamics of all second moment quantities in $\rho$, such as the behavior of the fidelity, the linear cross-entropy metric, and the collision probability.

For the rest of this section, we will illustrate the physics of the model by considering Haar-random two-qubit gates as an example.
However, recall that the analysis also holds for arbitrary two-qubit gates as long as they are followed by random single-qubit gates from a 2-design.
We will state a fact about the steady state of the two-copy average:
\begin{claim}[Convergence to global Haar average \cite{dalzell2022randomquantumcircuitsanticoncentrate}]
For a sufficiently well-connected architecture, in the limit $t\to \infty$, the two-copy average state $\bar{\rho}^t$ converges to
\begin{align}
 \lim_{t\to \infty} \bar{\rho}^t = \bar{\rho}_H = \frac{2^n}{2^n+1} \mathtt{II\ldots I} + \frac{1}{2^n+1} \mathtt{SS\ldots S}. \label{eq_haar_equilibrium}
\end{align}
\end{claim}
This is the same two-copy average as that of a global Haar-random state on $n$ qubits.
This fact gives barren plateaus for deep random quantum circuits on any architecture, since one can often show for certain architectures that the convergence is exponentially fast in the depth.
Morally speaking, for unitary random circuits, any second moment quantity such as the variance of the loss function, is close to that of global Haar-random unitaries for sufficient depth.
The global Haar-average state $\bar{\rho}_H$ and the identity state $\mathtt{II\ldots I}$ are the two fixed points under random local unitaries.
Note also that this claim is not strongly dependent on the ensemble of gates used---as long as the distribution over single-qubit gates forms a two-design (so that the stat-mech formalism applies) and the transfer matrices of the two-qubit gates have nonzero entangling power $\alpha \geq 0$, the fixed point is $\bar{\rho}_H$.

\begin{claim}[Noisy update rules]
In the presence of noise, which we model as local stochastic noise acting on every qubit after every layer of gates, there is an additional stat-mech rule.
We have~\cite{dalzell2021randomquantumcircuitstransform}
\begin{align}
T_\mathrm{noise} =
\begin{pmatrix}
1 & \gamma \\
0 & 1 - \gamma
\end{pmatrix}, \label{eq_unital}
\end{align}
where $\gamma$ is proportional to the average infidelity of the noise channel.
For nonunital single-qubit noise channels, we have \cite{ware2023sharpphasetransitionlinear}
\begin{align}
T_\mathrm{noise} =
\begin{pmatrix}
1 - \delta & \gamma \\
\delta & 1 - \gamma
\end{pmatrix}, \label{eq_nonunital}
\end{align}
for parameters $\delta \leq \gamma$ related to the nonunitality and the nonunitarity of the channel.
\end{claim}

\subsection{Connection between the stat-mech model and barren plateaus}
Assume the loss function is given by $L(\btheta) = \Tr \rho(\btheta) H$, where $H = \sum_{i} h_i$ is a Hermitian observable and the individual terms $h_i$ of $H$, are operators in $\mathrm{Herm}(n)$ of locality at most $k$, i.e., can be written as $h_i = \tilde{h}_{A}\otimes I_{A^c}$ for a subset $A$ of size at most $k$.
We can assume without loss of generality that $\Tr H = 0$.
We are usually interested in terms of constant locality $k=O(1)$ and constant norm, since for these observables, we can estimate $\Tr \rho h_i$ to additive error $\epsilon$ using $O(1/\epsilon^2)$ copies of $\rho$ with high probability.

When studying barren plateaus, we are interested in the typical variance of the loss function when initializing the parameters $\btheta$ randomly from a distribution $\mathcal{D}_p$.

\begin{lemma}[Variance in terms of marginal distribution] \label{lem_ham_variance}
Consider a $k$-local Hamiltonian $H = \sum_{\bm{\alpha} \in \mathbb{P}_n} c_{\bm{\alpha}} \bm{\alpha}$
and a parameterized dynamic quantum circuit $\mathcal{C}$ taking in parameters $\btheta \in \Theta$.
For a probability distribution $\mathcal{D}_p$ over $\Theta$ the set of circuit parameters such that each component $\btheta_i$ of $\btheta$ controls the parameter of a single operation in $\mathcal{C}$ and is invariant under single-qubit random gates from a 2-design inserted after every gate, the variance of the loss function $L = \Tr \rho(\theta) H $ is given by
\begin{align}
 \Var_{\btheta \sim \mathcal{D}_p}{L} = \sum_{\bm{\alpha}} c_{\bm{\alpha}}^2 \Pr_{\mathcal{X}_d}[\bm{x}_{\supp (\bm{\alpha})} = {11\ldots 1}_{\supp (\bm{\alpha})}].
\end{align}
\end{lemma}
In the RHS, the quantity can be understood as the probability, when drawing a bit string $\bm{x}$ from the distribution over bitstrings $\mathcal{X}_d$ at time $d$, that all of the entries of the $\bm{x}$ corresponding to qubits on which ${\bm{\alpha}}$ is supported are equal to 1.

\begin{proof}
We have 
\begin{align}
 \mathbb{E}_{\btheta \sim \mathcal{D}_p} \Var{L} &= \mathbb{E}_{\btheta \sim \mathcal{D}_p} \left[  L(\btheta)^2 \right] - \left( \mathbb{E}_{\btheta \sim \mathcal{D}_p} \left[  L(\btheta) \right] \right)^2 \\
 &= \mathbb{E}_{\mathcal{B}}\left[  (\Tr \rho H)^2 \right] -  \left( \mathbb{E}_{\mathcal{B}}\left[\Tr \rho H\right] \right)^2.
\end{align}
We can rewrite the Hamiltonian in terms of a traceless part $H'$ and a part proportional to the identity $\lambda I$.
The latter portion does not contribute to the variance, so we focus on the traceless portion (and denote it $H$ for the rest of the proof).
The second term yields
\begin{align}
\left( \mathbb{E}_{\mathcal{B}}\left[\Tr \rho H\right] \right)^2 &= \left( \Tr \mathbb{E}_{\mathcal{B}}\left[ \rho H\right] \right)^2 \\
& = \left( \Tr \mathbb{E}_{\mathcal{B}}\left[ \rho \right] H \right)^2 \\
& = \left( \Tr \frac{H}{2^n} \right)^2 \\
& = 0
\end{align}
for any distribution with a final layer of single-qubit gates from a one-design.
We now focus on the quantity $ \mathbb{E}_{\mathcal{B}} \left[(\Tr \rho H)^2\right] $.
This quantity can be expressed in terms of the average two-copy state
\begin{align}
\mathbb{E}_{\mathcal{B}}\left[(\Tr \rho H)^2 \right] = \mathbb{E}_{\mathcal{B}}\left[\Tr \rho^{\otimes 2} \cdot H^{\otimes 2} \right] = \Tr \bar{\rho} \cdot H^{\otimes 2}
\end{align}
as seen before.
Since the two-qubit gates are from a locally scrambling ensemble, we assume that the single-qubit gates are independently drawn according to the Haar measure $\mathcal{H}$.
The two calculations coincide because we are computing a second moment quantity over the ensemble.
We find that it is often convenient to insert an additional (virtual) layer of random single-qubit unitaries $V_n\sim \mathcal{H}$ at the end of the circuit and Heisenberg-evolve the output observable:

\begin{align}
& \mathbb{E}_{\mathcal{B}}\left[\Tr \rho^{\otimes 2} \cdot H^{\otimes 2} \right] \nonumber  \\
& = \mathbb{E}_{\mathcal{B}} \mathbb{E}_{V_n \sim \mathcal{H}} \left[\Tr \rho^{\otimes 2} (V_n\otimes V_n)^\dag H^{\otimes 2} (V_n \otimes V_n) \right]. 
\end{align}
We now perform the Haar-average over each single-qubit gate $V_{n_k}$ in $V_n=\bigotimes_k V_{n_k}$.
Recall from \cref{eq_singleqbtpauli} that a single-qubit non-identity Pauli averages to:
\begin{align}
\mathbb{E}_{V_{nk} \sim \mathcal{H}} [{V_{nk}^{\otimes 2}}^\dag \bm{\alpha}^{\otimes 2} {V_{nk}^{\otimes 2}} ] &=
\frac{4}{3} (\mathtt{S} - \mathtt{I}).
\end{align} 
When generalizing this to the $n$-qubit Haar-average over $V_n$, we obtain
\begin{align}
\mathbb{E}_{V_{n} \sim \mathcal{H}^n} \left[ (V_n\otimes V_n)^\dag \bm{\alpha}^{\otimes 2} (V_n \otimes V_n) \right] = \nonumber \\
= \left(\frac{4}{3}\right)^{\abs{\bm{\alpha}}} \mathtt{(S-I)}^{\supp(\bm{\alpha})},
\end{align}
where we have defined $O^B$ for an operator $O$ and a set $B$ as $\otimes_{j \in B} O_j \otimes I_{j^c}$.
Continuing to calculate $\mathbb{E}_{\mathcal{B}}\left[\Tr \rho^{\otimes 2} H^{\otimes 2} \right]$, we have
\begin{align}
\mathbb{E}_{\mathcal{B}}\left[\Tr \rho^{\otimes 2} H^{\otimes 2} \right] &= \sum_{\bm{\alpha}, \bm{\beta} \in \mathbb{P}_n } c_{\bm{\alpha}} c_{\bm{\beta}} \mathbb{E}_{\mathcal{B}}\left[\Tr \rho^{\otimes 2} \cdot \bm{\alpha} \otimes \bm{\beta} \right] \\
& = \sum_{\bm{\alpha}} c_{\bm{\alpha}}^2 \mathbb{E}_{\mathcal{B}}\left[\Tr \rho^{\otimes 2} \cdot \bm{\alpha} \otimes \bm{\alpha} \right],
\end{align}
where we have used the relations in \cref{eq_product_paulis} that yield $\mathbb{E}_{V_n \sim \mathcal{H}}\left[\bm{\alpha} \otimes \bm{\beta} \right] = 0$.
Therefore
\begin{align}
\mathbb{E}_{\mathcal{B}}\left[\Tr \rho^{\otimes 2} H^{\otimes 2} \right] = \sum_{\bm{\alpha}}  \left(\frac{4}{3}\right)^{\abs{\bm{\alpha}}} c_{\bm{\alpha}}^2 \Tr \left[\bar{\rho} \mathtt{(S-I)}^{\supp(\bm{\alpha})}\right].
 \label{eq_var_general}
\end{align}
Now consider a particular term in \cref{eq_var_general}.
We obtain
\begin{align}
\mathbb{E}_{\mathcal{B}}\left[\Tr \rho^{\otimes 2} \bm{\alpha}^{\otimes 2} \right] &= \left(\frac{4}{3}\right)^{\abs{\bm{\alpha}}} \Tr \left[ \bar{\rho} \cdot \mathtt{(S-I)}^{\supp(\bm{\alpha})} \right] \\
& = \left(\frac{4}{3}\right)^{\abs{\bm{\alpha}}} \Tr_{A} \Tr_{A^c}\left[ \bar{\rho} \cdot \mathtt{(S-I)}^{\supp(\bm{\alpha})} \right],
\end{align}
where we have split the trace into one over the support of $\bm{\alpha}$, $A$, and another on its complement $A^c$.
Let us perform the trace over the complement first.
Let 
\begin{align}
\Tr_{A^c} \bar{\rho} = \sum_{\bm{x'} \in \{0,1\}^{n- \abs{\bm{\alpha}}} } p_{\bm{x'}}  \mathtt{T}^{\bm{x'}},
\end{align}
where the coefficients $p_{\bm{x'}}$ for any substring $\bm{x'} \in \{0,1\}^{\abs{A}}$ are obtained by marginalizing the probability distribution ${p}_{\bm{x}}$ over the bits outside $A$:
\begin{align}
 p_{\bm{x'}} = \sum_{\bm{x}: \bm{x}_{A} = \bm{x'} } p_{\bm{x}}.
\end{align}
Now, from the relations
\begin{align}
    \Tr \mathtt{I \cdot (S-I)} = 0; \quad \Tr \mathtt{S \cdot (S-I)} = \frac{3}{4}, \label{eq_orthogonality}
\end{align}
we conclude that each term ${\bm{x}}$ in $\Tr_{A^c} \bar{\rho}$ when multiplied with $\mathtt{(S-I)}$ gives a trace of 0 unless ${\bm{x}}= 11\ldots 1$, in which case it yields a trace of 1.
This yields a result relating the variance of a term in the loss function to the probability of ending with the all $\mathtt{SS\ldots S}_{\supp(\bm{\alpha})}$ string on the support of that term:
\begin{align}
\mathbb{E}_{\mathcal{B}}\left[\Tr \rho^{\otimes 2} \bm{\alpha}^{\otimes 2} \right] = p_{\bm{x'} =  11\ldots 1_{\abs{A}}} \,.
\end{align}
Putting everything together, we get
\begin{align}
 \Var_{\btheta \sim \mathcal{D}_p}{L} = \mathbb{E}_{\mathcal{B}}\left[\Tr \rho^{\otimes 2} H^{\otimes 2} \right] = \sum_{\bm{\alpha}} c_{\bm{\alpha}}^2 g_{\bm{x'} = {1\ldots1}_{\abs{\supp(\bm{\alpha})}}} \,.
\end{align}
\end{proof}

Note that this lemma is similar to that of Napp \cite{napp2022quantifyingbarrenplateauphenomenon} except for the fact that we assume arbitrary distributions over 2-qubit gates provided the single-qubit gates are randomly chosen from a 2-design; whereas Napp assumes that the 2-qubit gates are drawn Haar randomly.
It is also possible to generalize this result to qudits of dimension $q$, although for simplicity, we fix $q=2$ in this work.

\subsection{Stat-mech mapping rules for circuits with feedforward operations}
In this section, we derive the stat-mech rule for measurements with feedforward operations.
First, we recall our treatment of dynamic operations in the purified picture.
As an example, the simplest nontrivial dynamic operation is the following:
\begin{align}
\begin{quantikz}[row sep={1cm,between origins},transparent]
\qw & \meter{} & \setwiretype{c} & \gate{\text{If 1, } U_1} & \setwiretype{q}
\end{quantikz} \label{eq_basic_ff}
\end{align}
Here, the state is measured and is followed by a conditional gate $U_1$ if the measurement result was $1$ and $I$ otherwise.

It is often convenient to view these operations in a larger Hilbert space by introducing an ancilla to record the intermediate measurement results and converting the conditional statements into gates that condition on the ancilla.
We call this the purified picture.
Concretely, the purified picture of the above feedforward operation is:
\begin{align}
\begin{quantikz}[row sep={1cm,between origins},transparent]
\lstick{$\ket{0}$}  & \targ{}   & \ctrl{1}  & \meter{} & \ground{} \setwiretype{c} \\
\lstick{$\ket{\psi_0}$}   & \ctrl{-1} & \gate{U_1} &  & &
\end{quantikz},
\label{eq_feedforward_purified}
\end{align}
where we trace out the ancilla after measurement.
The ancilla measurement reveals the result of the mid-circuit measurement, but in this work we will consider observables of dynamic circuits where we average over the mid-circuit measurement results (hence, strictly speaking, the measurement on the ancilla is superfluous here).
We fix
\begin{align}
U_1 = \begin{pmatrix}
\cos {\varphi}e^{-i\phi} & -i \sin \varphi \\
- i \sin \varphi & \cos{\varphi} e^{i \phi}
\end{pmatrix}. 
\end{align}

\begin{lemma}[Transfer matrix for feedforward operation]
For the single-qubit feedforward gadget shown in \cref{eq_basic_ff}, the transfer matrix is
\begin{align}
 T_F =
 \begin{pmatrix}
 1 - \frac{\sin^2 \varphi}{3} & \frac{2}{3} \\
 \frac{\sin^2\varphi}{3} & \frac{1}{3}
 \end{pmatrix}. \label{eq_ff_transfermat}
\end{align}
\end{lemma}

\begin{proof}
We will work in the purified picture for simplicity, where the feedforward gadget is given by \cref{eq_feedforward_purified}.
We need to calculate
\begin{align}
\bar{\rho}^{t+1} = \mathbb{E}_{V_1}[(V_1 \otimes V_1) (\mathcal{F} \otimes \mathcal{F})(\bar{\rho}^t) (V_1 \otimes V_1)^\dag]
\end{align}
for Haar-random single-qubit gates $V_1$ and the feedforward channel $\mathcal{F}$.
The tensor network diagram for this calculation is:
\begin{widetext}
\begin{align}
\mathbb{E}_{V_1} \Tr_{1,2}
\begin{quantikz}[row sep={1cm,between origins},transparent]
\slice{step 3} & \slice{step 2}& \ctrl{2} & \slice{step 1} & \targ{}  &  \slice{step 0} & \gate{\ketbra{0}} \slice{step 0}  &  & \targ{} \slice{step 1} &   & \ctrl{2} \slice{step 2}& \qw &  \\
&& & \ctrl{2} &    & \targ{}  &    \gate{\ketbra{0}} & \targ{} &  & \ctrl{2} & & \qw & \\
& \gate{V_1} & \gate{U} & &\ctrl{-2}  & & \gate[2]{\bar{\rho}^t} & & \ctrl{-2} &  & \gate{U^\dag} & \gate{V_1^\dag} & \\
& \gate{V_1} & & \gate{U} &  & \ctrl{-2} &  &  \ctrl{-2} &  & \gate{U^\dag} &  & \gate{V_1^\dag} \slice{step 3} & 
\end{quantikz}.
\end{align}
\end{widetext}
Here time flows from the center outwards on either side, as denoted by the time slices.

Suppose that the initial 2-copy average state is $\bar{\rho}^t= aI +bS$, where $I$ and $S$ are now the unnormalised identity and SWAP operators.
Then, the state after step 1 when we apply the CNOTs from lines $3\to 1, 4\to 2$ is
\begin{align}
a \left( \ketbra{00} \otimes \ketbra{00} + \ketbra{01} \otimes \ketbra{01} + \right.\nonumber \\
\left.\ketbra{10} \otimes \ketbra{10} + \ketbra{11} \otimes \ketbra{11}\right) \nonumber \\
+ b  \left( \ketbra{00} \otimes \ketbra{00} + \ketbra{01}{10} \otimes \ketbra{01}{10} + \right.\nonumber \\
\left. \ketbra{10}{01} \otimes \ketbra{10}{01} + \ketbra{11} \otimes \ketbra{11}\right).
\end{align}
Next, applying the controlled unitary on both copies yields:
\begin{widetext}
\begin{align}
&(a+b) \ketbra{00} \otimes \ketbra{00} + (a+b) \ketbra{11} \otimes (U\otimes U) \ketbra{11} (U\otimes U)^\dag \nonumber \\
&+ a \ketbra{01} \otimes \left(\sin^2 \varphi \ketbra{00} - ie^{-i\phi}\sin \varphi \cos \varphi \ketbra{00}{01} + i e^{i\phi}\sin \varphi \cos \varphi \ketbra{01}{00} + \cos^2\varphi \ketbra{01} \right) \nonumber \\
&+ a \ketbra{10} \otimes \left(\sin^2 \varphi \ketbra{00} - ie^{-i\phi}\sin \varphi \cos \varphi \ketbra{00}{10} + i e^{i\phi}\sin \varphi \cos \varphi \ketbra{10}{00} + \cos^2\varphi \ketbra{10} \right) \nonumber \\
&+ b \ketbra{01}{10} \otimes (\ldots) + b \ketbra{10}{01} \otimes (\ldots),
\end{align}
\end{widetext}
where we have omitted the last two terms proportional to $b$ because they will vanish when we take the partial trace over the first two qubits.
After taking the partial trace, we will have

\begin{widetext}
\begin{align}
& (a+b) \ketbra{00} + (a+b) (U\otimes U) \ketbra{11} (U\otimes U)^\dag \nonumber \\
& + a \left(\sin^2 \varphi \ketbra{00} - ie^{-i\phi}\sin \varphi \cos \varphi \ketbra{00}{01} + i e^{i\phi}\sin \varphi \cos \varphi \ketbra{01}{00} + \cos^2\varphi \ketbra{01} \right) \nonumber \\
& + a \left(\sin^2 \varphi \ketbra{00} - ie^{-i\phi}\sin \varphi \cos \varphi \ketbra{00}{10} + i e^{i\phi}\sin \varphi \cos \varphi \ketbra{10}{00} + \cos^2\varphi \ketbra{10} \right)
\end{align}    
\end{widetext}
We now perform the average over $V_1$.
This gives us $\bar{\rho}^{t+1} = \alpha I + \beta S$, where the coefficients can be inferred from
\begin{align}
\Tr \bar{\rho}^{t+1} &= 4\alpha + 2\beta = 4a + 2b\\
\Tr \bar{\rho}^{t+1} S &= 2\alpha + 4\beta = 2a + 2b + 2a \sin^2 \varphi \\
\implies \alpha &= \frac{a(3-\sin^2\varphi) + b}{3}, \quad \beta = \frac{2a \sin^2\varphi + b}{3}.
\end{align}
This, in turn, shows that a measurement and feedforward operation leads to the rule for the stat mech model:
\begin{align}
\mathtt{I} \to & \frac{(3-\sin^2 \varphi)}{3} \mathtt{I} + \frac{\sin^2\varphi}{3} \mathtt{S} \\
\mathtt{S} \to & \frac{2}{3}\mathtt{I} + \frac{1}{3} \mathtt{S},
\end{align}
completing the proof.
\end{proof}

Let us examine some limits.
When $\varphi=0$, there is no effect of the feedforward unitary and the operation is equivalent to measuring and forgetting the result, or dephasing.
The stat mech rule would then be $\mathtt{I} \to \mathtt{I}; \mathtt{S} \to \frac{2}{3}\mathtt{I} + \frac{1}{3} \mathtt{S}$.
$\varphi = \pi/2$ corresponds to the case of controlled operation being a CNOT, when the whole gadget acts essentially as a reset.
In this case, both $\mathtt{I}$ and $\mathtt{S}$ map to the standard $\frac{2}{3}\mathtt{I} + \frac{1}{3} \mathtt{S}$ single-qubit average of a pure state.
Interestingly, there is no fundamental difference between the stat mech model for any $\varphi >0 $ and $\varphi = \frac{\pi}{2}$.

We can also study the fixed point of this single-qubit map for more intuition.
We equate $\alpha = a$, $\beta = b = \frac{1}{2} - 2a$ and obtain
\begin{align}
  a = \frac{1}{6 - 2\cos^2\varphi}, \quad b = \frac{\sin^2\varphi}{6-2\cos^2\varphi}.
\end{align}
This again shows that when the measurements only serve to dephase the state ($\varphi = 0$), the fixed point is $\mathtt{I}$, consistent with studies of unital noise \cite{deshpande2022tightboundsconvergencenoisy,ware2023sharpphasetransitionlinear}, whereas for any nonzero $\varphi$, the fixed point of the channel has a nonzero $\mathtt{S}$ component.

In sum, we have seen that from the point of view of the second moment and the stat-mech model, there is qualitatively no difference between some $\varphi>0$, e.g. $\varphi = \frac{\pi}{4}$, and $\varphi = \frac{\pi}{2}$, the case of resets.

\section{Physics of the feedforward stat-mech model}
In this Appendix, we study the dynamics of the stat-mech model in the presence of feedforward operations and derive our main result, a lower bound on the variance of the loss function.

We define here a few quantities that will be useful in our analysis.

\begin{defn}[Parameterized dynamic circuit]\label{def:parameterized_dynamic_circuit_app}
A parameterized dynamic circuit $\mathcal{C}$ of depth-$d$ on an architecture is a sequence of channels $\{\mathcal{U}_1(\btheta), \mathcal{U}_2(\btheta), \ldots \mathcal{U}_d(\btheta) \}$, where each channel $\mathcal{U}_j(\btheta): \mathrm{Dens}(n) \to \mathrm{Dens}(n)$ is a completely positive, trace preserving map on $n$ qubits that describes the operations in time step $j$.
The operations in each layer $j$ can be written as a composition of single- and two-qubit operations such that the two-qubit operations act on non-overlapping qubits and only connect qubits with an edge in the associated graph describing the circuit architecture.
The parameters $\btheta$ come from a set $\Theta \in \mathbb{R}^{p}$ for $p \in \poly(n)$.
The action of the entire circuit is described by the channel $\mathcal{C}(\btheta) = \mathcal{U}_d(\btheta) \circ \ldots \circ \mathcal{U}_1(\btheta)$.
\end{defn}
In the most general case in this definition, some operations in the circuit may depend on no parameters, and some parameters may influence multiple operations.
For simplicity, we focus on the case where each parameter $\theta$ is a rotation angle of either a single-qubit gate $e^{-i\theta \sigma}$ or a two-qubit gate $e^{-i\theta \sigma \otimes \sigma }$ for some Pauli matrix $\sigma$, and each parameter only controls a single gate.

We also need the idea of the effective channel up to time $t$:
\begin{defn}
The effective channel until time $t$ is given by $\mathcal{C}^t = \mathcal{U}_t(\btheta) \circ \ldots \circ \mathcal{U}_1(\btheta)$.
\end{defn}

With this in hand, we now define the ensemble of circuits we work with.

\begin{defn}[Ensemble of dynamic circuits]
Suppose there is a probability distribution $\mathcal{D}_p$ defined on the space of possible parameter values $\Theta$.
This distribution induces in a natural way a distribution over $\mathcal{C}(\btheta)$
by choosing $\btheta \sim \mathcal{D}_p$.
We denote this ensemble of dynamic circuits, and equivalently, channels, $\mathcal{B}$.
\end{defn}

We are interested in ensembles of dynamic ciruits $\mathcal{B}$ that have the \emph{local scrambling} property that the distribution over every operation is invariant under a single-qubit rotation chosen from a 2-design ensemble (as defined in Appendix \ref{app:definitions}).
For locally scrambling parameterized circuits, we may assume the parameterized circuit has the form $\mathcal{C} = \mathcal{V}_d(\btheta)\circ \mathcal{U}_d(\btheta) \circ \ldots  \mathcal{V}_1(\btheta) \circ \mathcal{U}_1(\btheta)$, where $\mathcal{V}_1(\btheta)(\rho) = \left(\prod_{i=1}^n V_i(\btheta)\right) \rho \left(\prod_{i=1}^n V_i(\btheta)\right)^\dag$ is a product of single-qubit unitaries $V_i \sim \mathcal{T}_2$ for a 2-design $\mathcal{T}_2$.

\begin{defn}[Backwards light cone]
Consider a parameterized dynamic circuit.
For a subregion $A \subset [n]$, define $G_d$ to be the set of operations $g$ in $\mathcal{U}_d$ that can potentially affect the subregion, i.e.\ satisfy $\supp(g) \cap A \neq \emptyset$.
Recursively define $G_{j-1} = \{g \in \mathcal{U}_{j-1}: \supp(g) \cap G_j \neq \emptyset \}$
The depth $k$-backwards light cone of $A$ is the set $L^b_k(A) := G_d \cup G_{d-1} \ldots G_{d-k}$.
Proceeding to depth 1, the set $L^b(A) := G_d \cup G_{d-1} \ldots G_1 $ is called the backwards light cone of the subregion $A$.
\end{defn}

\begin{defn}[Path]
A path in a circuit originating from time $t$ is a specification of space-time locations, i.e.\ a choice of a qubit $\ell(s)\  \forall \ s \in \{t, t+1, \ldots d\}$, such that for every time slice $s$, there is a gate or operation in $\mathcal{U}_{s} $ that connects $\ell(s)$ and $\ell(s+1)$, meaning that the operation is supported both on qubits $\ell(s)$ and $\ell(s+1)$.
We denote such a path using its sequence $\mathcal{P} := (\ell(t),\ell(t+1), \ldots\ell(d))$.
\end{defn}

\begin{defn}[Path length]
The length of a path $\mathcal{P} = (\ell(t),\ell(t+1),\ldots \ell(d))$ is the number of space-time locations $s$ such that in the associated circuit, the operation from $\mathcal{U}_s$ that acts on $\ell(s)$ and $\ell(s+1)$ is an entangling gate.
We denote the length of a path $\abs{\mathcal{P}}$.
\end{defn}

\begin{defn}[Worst-case feedforward distance]\label{def_ff_distance}
For any qubit $j$, consider all the paths arising from any feedforward operation at any time $t$ and ending in $j$.
Define the \emph{feedforward distance} of qubit $j$ to be the minimum path length of all these paths.
Now, consider the maximum feedforward distance out of all qubits $j \in [n]$.
We define this to be the worst-case feedforward distance of the parameterized dynamic circuit, denoted $f$.
More precisely,
\begin{align}
 f = \max_{j \in [n]} \min_{\ell(t): \text{feedforward} } \abs{(\ell(t), \ell(t+1), \ldots j)}.
\end{align}
\end{defn}

We have now established the tools and language we need to prove the claimed lower bound on the variance of the loss function.
Just as before, we are interested in the dynamics of the stat-mech model on bitstrings $\bm{x} \in \{0,1\}^n$.
The bitstrings often play the role of random variables, and the probability of sampling a string $\bm{x}$ at time $t$ is denoted $\Pr_{\mathcal{X}_t}[\bm{x}] = p^t_{\bm{x}}$.
We also are interested in conditional probabilities where the state (bitstring) at time $t-1$ is known.
These are denoted $\Pr_{\mathcal{X}_t}[\bm{x}^t | \bm{x}^{t-1}] = p^t_{\bm{x}| \bm{x}^{t-1} }$.
\begin{lemma}[Lower bound on probability mass on $\mathtt{S}$ after feedforward operation]
Suppose qubit $j$ undergoes a feedforward operation at time $t$.
Then, the probability that $\bm{x}_j = 1$ at time $t$, equivalently, the probability mass $\Tr_{j^c} \bar{\rho}^t \cdot \frac{4}{3} (\mathtt{S} - \mathtt{I})_{j}$, is lower bounded by $\frac{\sin^2 \varphi}{3}$. \label{lem:swapmassafterff}
\end{lemma}
\begin{proof}
If $\Tr_{j^c} \bar{\rho}^t = a \mathtt{I} + b \mathtt{S} $, we would like a lower bound on the probability $b$ of seeing an $\mathtt{S}$ after the feedforward operation.
From the orthogonality relations in \cref{eq_orthogonality}, we get
\begin{align}
 \Tr_{j^c} \bar{\rho}^t \cdot (\mathtt{S} - \mathtt{I})_{j} = a \cdot 0 +  \frac{3}{4} b.
\end{align}
This establishes that the quantity $\frac{4}{3} \Tr_{j^c} \bar{\rho}^t \cdot (\mathtt{S} - \mathtt{I})_{j}$ is the probability mass we desire.
As for the lower bound, it follows straightforwardly from observing that \cref{eq_ff_transfermat} has an entry corresponding to $\mathtt{S}$ of either $\frac{\sin^2\varphi}{3}$ or $\frac{1}{3}$.
Therefore, we get
\begin{align}
 \Pr_{\mathcal{X}_t} [\bm{x}_j = 1] = \sum_{\bm{x}: \bm{x}_j = 1} {p}^t_{\bm{x}} \geq \frac{\sin^2\varphi}{3},
\end{align}
where the sum $\sum_{\bm{x}: \bm{x}_j = 1}$ plays the role of marginalizing the probability distribution ${p}^t$ over $j^c$.
\end{proof}
While this proof is self-evident, this simple fact is crucial.
We are able to make a strong statement by giving a constant lower bound on the probability of an event, independent of all other gates in the circuit and independent of the time at which the feedforward operation takes place. 
Without intermediate measurements and feedforward operations, we typically do not have a good lower bound on the probability mass on any string or sequence.
Particularly, as seen earlier, for unitary circuits the associated distribution quickly converges to that given by \cref{eq_haar_equilibrium}.
This in turn means, for large $t$, 
\begin{align}
 p^t_{\bm{x}_j=1} \approx \frac{1}{2^n+1},
\end{align}
which is exponentially small.

\begin{lemma}[Lower bound on probability mass on $\mathtt{S}$ surviving after two-qubit gate]
Suppose we apply a two-qubit gate on qubits $j, j'$ at time $t$ with transfer matrix parameters $\alpha$ and $\beta$ following \cref{eq_twoqubit_transfermat}.
The conditional probability $\Pr_{\mathcal{X}_t} [\bm{x}^t_{j} = 1| \bm{x}^{t-1}_{j} = 1 \text{ or } \bm{x}^{t-1}_{j'} = 1] $ is lower bounded by $\frac{\alpha}{5}$. \label{lem:swapmass_remain}
\end{lemma}
\begin{proof}
The proof of this lemma is largely similar to the previous.
The main difference is that we are examining the conditional probability of there being an $\mathtt{S}$ on qubit $j$ at time $t$, with the promise that either qubit $j$ or $j'$ had an $\mathtt{S}$ operator at time $t-1$.
In this setting, we can restrict our attention to the columns of the transfer matrix \cref{eq_twoqubit_transfermat} corresponding to the initial states $\mathtt{IS}_{jj'}, \mathtt{SI}_{jj'}, \mathtt{SS}_{jj'}$, meaning the last three columns of
\begin{align}
T =
\begin{pmatrix}
1 & \frac{4\alpha}{5} & \frac{4\alpha}{5} & 0 \\
0 & 1-\alpha - \beta & \beta & 0 \\
0 & \beta & 1-\alpha  - \beta & 0 \\
0 & \frac{\alpha}{5} & \frac{\alpha}{5} & 1
\end{pmatrix}.
\end{align}
Therefore, 
\begin{align}
 \Pr_{\mathcal{X}_t} [\bm{x}^t_{j} = 1| \bm{x}^{t-1}_{j} = 1 \text{ or } \bm{x}^{t-1}_{j'} = 1] \geq \nonumber \\ 
 \min \left\{\beta + \frac{\alpha}{5}, 1- \frac{4\alpha}{5} - \beta, 1 \right\}.
\end{align}
Since we have restricted our attention to transfer matrices with nonnegative entries, $\beta \geq 0$ and $\beta \leq 1-\alpha$.
The minimum of the three entries is $\geq \frac{\alpha}{5}$.
\end{proof}
Informally, \Cref{lem:swapmass_remain} states that an $\mathtt{S}$ string on a qubit $j$ has a probability at least $\frac{\alpha}{5}$ of remaining at $j$.
By symmetry, we also have the probability that an $\mathtt{S}$ string on a qubit $j$ moves or spreads to $j'$ after a two-qubit gate on $j,j'$:
\begin{lemma}[Lower bound on probability of $\mathtt{S}$ hopping or spreading]
Suppose we apply a two-qubit gate on qubits $j, j'$ at time $t$ with transfer matrix parameters $\alpha$ and $\beta$ following \cref{eq_twoqubit_transfermat}.
The conditional probability $\Pr_{\mathcal{X}_t} [\bm{x}^t_{j'} = 1| \bm{x}^{t-1}_{j} = 1 \text{ or } \bm{x}^{t-1}_{j'} = 1] $ is lower bounded by $\frac{\alpha}{5}$. \label{lem:swapmass_move}
\end{lemma}
\begin{proof}
This proof is almost identical and follows from the fact that
\begin{align}
 & \Pr_{\mathcal{X}_t} [\bm{x}^t_{j'} = 1| \bm{x}^{t-1}_{j} = 1 \text{ or } \bm{x}^{t-1}_{j'} = 1] = \nonumber\\
 & \Pr_{\mathcal{X}_t} [\bm{x}^t_{j} = 1| \bm{x}^{t-1}_{j} = 1 \text{ or } \bm{x}^{t-1}_{j'} = 1]
\end{align}
since the transfer matrix is symmetric under the operation of exchanging the qubits $j \leftrightarrow j'$.
\end{proof}
We will now define some useful concepts that let us reason about paths from feedforward operations to the end of the circuit.

\begin{defn}[Circuit-compatible sequence]
We call a sequence of bitstrings $\bm{x}^t, \bm{x}^{t+1}, \ldots \bm{x}^{t'}$ a circuit-compatible sequence of a dynamic parameterized circuit $\mathcal{C}$ if they satisfy the consistency condition that the string $\bm{x}^s$ is obtainable in principle from the string at time $s-1$, namely $\bm{x}^{s-1}$ from the sequence of operations in $\mathcal{U}_s$.
\end{defn}
The phrase ``obtainable in principle'' means that the probability $\Pr_{\mathcal{X}_s} [\bm{x}^s | \bm{x}^{s-1}]$ is nonzero.

\begin{defn}[SWAP-active sequence]
We call a circuit-compatible sequence \emph{SWAP-active} with respect to a given path if, for all space-time locations $\ell(s)$ specified by the path, we have $\bm{x}^s_{\ell(s)} = 1$.
\end{defn}

Informally speaking, a path identifies a potential route from a space-time location with an $\mathtt{S}$ operator to an $\mathtt{S}$ operator at the end of the circuit supported on a potentially different location, and we are interested in SWAP-active sequences with respect to this path.
Specifically, we would like to lower bound the probability that an $\mathtt{S}$ operator survives at a given location $l$ at the end of the circuit.
This depends on the path length of a path starting from the initial state or a nearby feedforward operation and ending at the given location $l$.

\begin{lemma}
Consider a qubit at site $l$.
The probability that the final bitstring contains an $\mathtt{S}$ at site $l$, i.e.\ $\Pr_{\mathcal{X}_d}[\bm{x}^d_l = 1] $, is lower bounded by 
\begin{align}
\frac{\sin^2 \varphi}{3} \left( \frac{\alpha}{5} \right)^\abs{(\ell(t), \ell(t+1),\ldots \ell(d))}
\end{align}
for any path starting from a feedforward operation at $\ell(t)$ and ending at the site of interest $\ell(d)=l$. \label{lem_singlequbit_bound}
\end{lemma}
\begin{proof}
Consider a path $(\ell(t), \ell(t+1), \ldots\ell(d))$ starting after a feedforward operation at qubit $\ell(t)$ at time $t$.
This operation results in $\bm{x}^{t}_{\ell(t)} = 1$ with probability at least $\frac{\sin^2 \varphi}{3} $ from \cref{lem:swapmassafterff}.
Once we get a handle on this probability mass, we can lower bound the probability mass of the event $\bm{x}^{t+1}_{\ell(t+1)} = 1$ using the conditional probability lower bounds in \cref{lem:swapmass_remain,lem:swapmass_move}.
If there is an entangling gate between qubits $\ell(t)$ and $\ell(t+1)$ at time $t+1$, then
\begin{align}
 \Pr[\bm{x}^{t+1}_{\ell(t+1)} = 1 | \bm{x}^{t}_{\ell(t)} = 1] \geq \frac{\alpha}{5}.
\end{align}
Otherwise, if the qubit is idle and $\ell(t+1) = \ell(t)$, the probability mass is unaffected.
This step holds inductively at arbitrary time slice $s$:
\begin{align}
 \Pr[\bm{x}^{s+1}_{\ell(s+1)} = 1 | \bm{x}^{s}_{\ell(s)} = 1] \geq \frac{\alpha}{5}.
\end{align}
Chaining together everything, we get
\begin{align}
 \Pr[\bm{x}^{d}_{\ell(d)} = 1 ] \geq \Pr[\bm{x}^{t}_{\ell(t)} = 1] \times \left( \frac{\alpha}{5} \right)^{\abs{(\ell(t),\ell(t+1),\ldots \ell(d))}},
\end{align}
where we have counted the number of times a two-qubit gate is encountered along the path, which is precisely how we have defined the path length.
In order to optimize the lower bound, we look for paths to the closest feedforward or initialization operation, using the distance measure we have defined.
The intuition we have made rigorous is that qubit initializations and measurement and feedforward operations are ``sources'' of $\mathtt{S}$ probability mass.
The path between the source and the ``sink'' (the final destination at $l$) is a ``leaky pipe'' where the probability of some mass surviving can decay exponentially in the length of the pipe.
Thus, in order to get a good flow rate to the sink, we try and optimize the plumbing so that sources and sinks are as close to each other as possible.
\end{proof}

We have seen how the probability of a single $\mathtt{S}$ surviving at the end of the circuit at one location is related to the feedforward distance.
Let us also study the case $\Pr[\bm{x}^{d}_{A} = 11\ldots 1]$, the probability of ending up with the all $\mathtt{S}$ operator on a region $A$.

\begin{lemma}
For a region $A$, the probability that the final bitstring contains $\mathtt{S\ldots S}$ in the entire region $A$, i.e.\ $\Pr_{\mathcal{X}_d}[\bm{x}^d_l = 1 \forall l \in A]$, is lower bounded by 
\begin{align}
\left( \frac{\sin^2 \varphi}{3} \right)^{\abs{A}} \left( \frac{\alpha}{5} \right)^{\sum_{i \in A} \abs{(\ell(t_i),\ell(t_i+1),\ldots,i)}}
\end{align}
for any set of paths starting from some feedforward operation $t_i$ and ending at sites in A $\ell(d_i)=i$. \label{lem_multiqubit_bound}
\end{lemma}

\begin{proof}
For every qubit in $A$, consider the path with shortest path length to a nearby feedforward operation in the circuit $\mathcal{C}$.
For two distinct qubits, it is possible that the shortest paths overlap significantly and begin at the same feedforward operation.

\begin{figure}
\includegraphics[width=0.95\columnwidth]{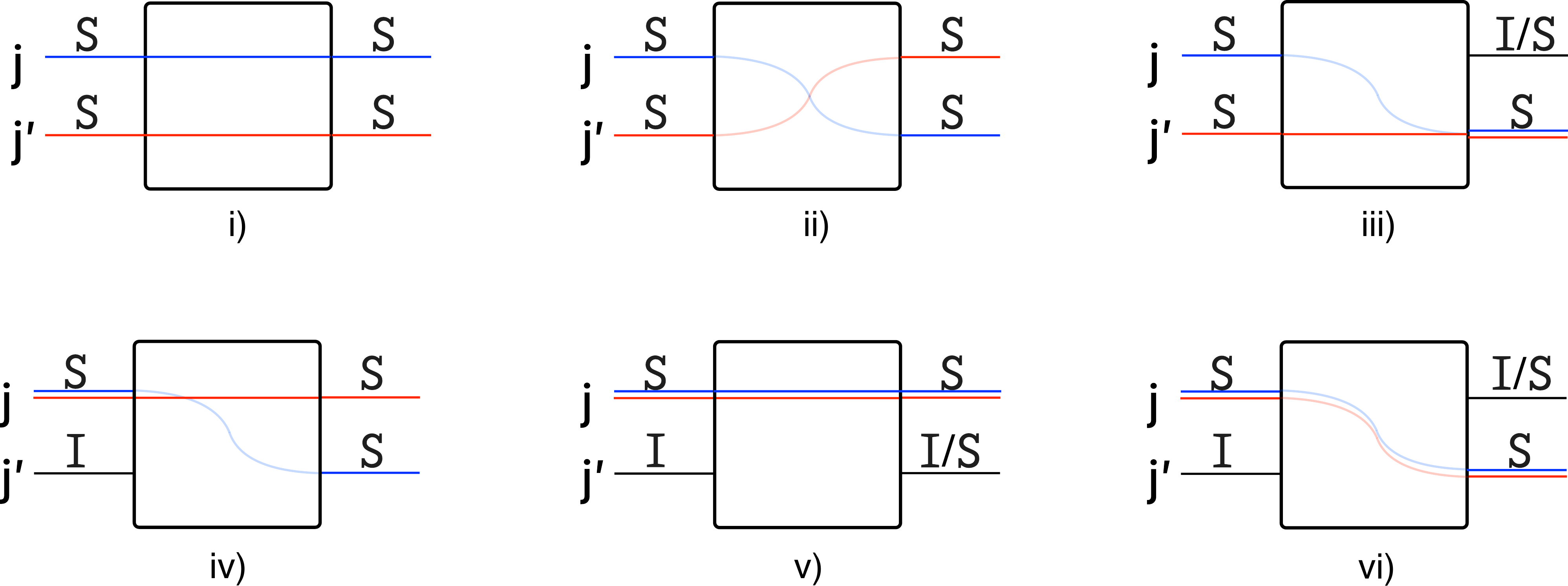}
\caption{Configurations in which two SWAP-active paths can interact at a two-qubit gate, up to permutations blue $\leftrightarrow$ red and permutations at the input. The labels on the left depict the operators at the input and the ones on the right the operator labels at the output.
The blue and red lines denote SWAP-active sequences corresponding to two different paths ending in qubits $i, i' \in A$ (not depicted).
We are only interested in the blue/red lines on SWAP-active sequences ending in $\mathtt{S}$ operators, while the operators on the black lines can be $\mathtt{I}$ or $\mathtt{S}$.
}
\label{fig:intersectingcases}
\end{figure}

One may apply \cref{lem_singlequbit_bound} to bound the mass of each $\mathtt{S}_{i\in A}$ separately.
This is despite the fact that the events are not independent---multiple paths from the same source could contribute to the event of seeing an $\mathtt{S\ldots S}_A$ string at the output.
As we show below, this can be handled since the product of the individual lower bounds are a lower bound on the combined event.

Since at every gate, at most two SWAP-active paths can interesect, it suffices to consider the case of two SWAP-active paths at a time.
We consider the six cases depicted in \cref{fig:intersectingcases}.

In cases i) and ii), the operators $\mathtt{SS}_{jj'}$ map to $ \mathtt{SS}_{jj'}$.
The probability of this event is $\Pr_{\mathcal{X}_{s+1}}[\bm{x}^{s+1}_{jj'} = 11 | \bm{x}^{s}_{jj'} = 11]= 1 \geq \left( \frac{\alpha}{5} \right)^2$.
Next for case iii), the event $\mathtt{SS}_{jj'} \to  \star_j \mathtt{S}_{j'}$ again has probability $\Pr_{\mathcal{X}_{s+1}}[\bm{x}^{s+1}_{j'} = 1 | \bm{x}^{s}_{jj'} = 11]= 1 \geq \left( \frac{\alpha}{5} \right)^2$.
In cases iv)--vi), the input is $\mathtt{SI}_{jj'}$.
For case iv), where the output is $\mathtt{SS}_{jj'}$, the probability is $\Pr_{\mathcal{X}_{s+1}}[\bm{x}^{s+1}_{jj'} = 11 | \bm{x}^{s}_{jj'} = 10]= \frac{\alpha}{5} \geq \left( \frac{\alpha}{5} \right)^2$.
The probability for case v) is $\Pr_{\mathcal{X}_{s+1}}[\bm{x}^{s+1}_{j} = 1 | \bm{x}^{s}_{jj'} = 10]= \frac{\alpha}{5} + 1-\alpha - \beta \geq \left( \frac{\alpha}{5} \right)^2$, and for case vi), $\Pr_{\mathcal{X}_{s+1}}[\bm{x}^{s+1}_{j'} = 1 | \bm{x}^{s}_{jj'} = 10]= \frac{\alpha}{5} + \beta \geq \left( \frac{\alpha}{5} \right)^2$.
Thus, we have seen in all cases that the conditional probabilities are lower bounded by $\left( \frac{\alpha}{5} \right)^2$, which is the lower bound we would have assigned for the two SWAP-active sequenecs if we ignored the interaction between them.

Therefore, we can lower bound the probability of the entire event by assuming conditional independence, giving
\begin{align}
 \left( \frac{\sin^2 \varphi}{3} \right)^{n_t} \left( \frac{\alpha}{5} \right)^{\sum_{i \in A} \abs{(\ell(t_i), \ell(t_i+1), \ldots,i)} },
\end{align}
where $n_t$ is the number of sources of $\mathtt{S}$ strings.
In the worst case, each $\mathtt{S}$ string comes from a different source giving $n_t = \abs{A}$ the locality of $A$.
The entire lower bound is now
\begin{align}
 \Pr_{\mathcal{X}_d}[\bm{x}^d_l = 1 \ \forall \ l \in A] \geq \left( \frac{\sin^2 \varphi}{3} \right)^{\abs{A}} \left( \frac{\alpha}{5} \right)^{ \sum_{i \in A} \abs{(\ell(t_i),\ell(t_i+1), \ldots,i)} }.
\end{align}
The quantity in the exponent is obtained by independently choosing a path for every qubit in $A$.
\end{proof}

We are ready to prove our main result on the absence of barren plateaus.
\begin{thm} \label{apx_thm:noiseless_lower_bound}
Consider a $k$-local Hermitian observable $H = \sum_{\bm{\alpha}} c_{\bm{\alpha}} \bm{\alpha}$ and a parameterized dynamic quantum circuit $\mathcal{C}$ with a distribution $\mathcal{D}_p$ over parameters $\btheta \in \Theta$, which satisfy the following properties:
\begin{enumerate}
\item Every component $\btheta_i\in\btheta$ parameterizes only a single operation in $\mathcal{C}$.
\item The distribution is such that $\mathcal{C}$ is locally scrambling.
\item The distribution over two-qubit gates has a transfer matrix of the form in \cref{eq_twoqubit_transfermat}.
\item $\mathcal{C}$ has constant worst-case feedforward distance $f$. 
\end{enumerate}Then, the variance of the loss function $L = \Tr \rho(\theta) H $ is lower bounded by
\begin{align}
 \Var_{\btheta \sim \mathcal{D}_p}{L} \geq \sum_{\bm{\alpha}} c_{\bm{\alpha}}^2 \left( \frac{\sin^2 \varphi}{3} \right)^{\abs{\bm{\alpha}}}  \left( \frac{\alpha}{5} \right)^{kf}.
\end{align}
\end{thm}

\begin{proof}
Consider a circuit architecture such that the worst-case feedforward distance, $f$ is upper bounded by a constant.
Assume the loss function is $\sum_{\bm{\alpha} \in \mathbb{P}_n} c_{\bm{\alpha}} \bm{\alpha}$ expanded in the Pauli basis.
We may ignore the identity term since it does not contribute to the variance.
Since the maximum locality of the Hamiltonian is $k$, so too is the maximum weight of any Pauli.
Since the worst-case feedforward distance is the maximum feedforward distance among any qubit at the output, in the worst case, the sum $\sum_{i \in A} \abs{(\ell(t_i),\ell(t_i+1), \ldots,i)}$ is at most $f \times k$.
The rest follows from \cref{lem_ham_variance,lem_multiqubit_bound}.
\end{proof}
Further fixing an architecture where $f$ is constant and $\varphi = \pi/2$, we get a lower bound on the variance of the form
\begin{align}
 \sum_{\bm{\alpha}} c_{\bm{\alpha}}^2 \cdot \Omega(1) = \Omega(\norm{H}_{HS}^2),
\end{align}
where $\norm{H}_{HS} := \sqrt{{\Tr H^2}} $ is the Hilbert-Schmidt norm of the operator $H$.

We can also prove the robustness of this result to unital noise:
\begin{thm}
Consider a $k$-local Hermitian observable $H = \sum_{\bm{\alpha}} c_{\bm{\alpha}} \bm{\alpha}$, and a parameterized dynamic quantum circuit $\mathcal{C}$ with a distribution $\mathcal{D}_p$ over parameters $\btheta \in \Theta$, which satisfy the same  conditions as in \cref{apx_thm:noiseless_lower_bound}, along with the additional condition:
\begin{enumerate}
\item[5.] After every two-qubit operation, there is a local noise channel (which can be nonunital in general) acting on every qubit, with a transfer matrix given in \cref{eq_nonunital}, with parameters $\delta \leq \gamma \leq \frac{1}{2}$.
\end{enumerate}
Then, the variance of the loss function $L = \Tr \rho(\theta) H $ is lower bounded by
\begin{align}
 \Var_{\btheta \sim \mathcal{D}_p}{L} \geq \sum_{\bm{\alpha}} c_{\bm{\alpha}}^2 \left( \frac{\sin^2 \varphi}{3} \right)^{\abs{\bm{\alpha}}}  \left(\frac{\alpha}{5}(1-\gamma-\delta) + \delta \right)^{kf}.
\end{align}
\end{thm}
\begin{proof}
We observe that under general nonunital noise, the combined transfer matrix of the two-qubit gate and the noise channel is 
\begin{widetext}
\begin{align}
\begin{pmatrix}
(1 - \delta)^2 & 
\gamma(1 - \delta) + \frac{4}{5} \alpha (1 - \gamma - \delta)\left(1 - \frac{\gamma}{4} - \delta\right) & 
\gamma(1 - \delta) + \frac{4}{5} \alpha (1 - \gamma - \delta)\left(1 - \frac{\gamma}{4} - \delta\right) & 
\gamma^2 \\
\delta (1 - \delta)& 
(1 - \gamma)(1 - \delta) - (1 - \gamma - \delta) \left[\beta + \alpha \left(1 - \frac{\gamma}{5} - \frac{4\delta}{5}\right)  \right] & 
\gamma\delta + (1 - \gamma - \delta) \left[\beta + \frac{\alpha}{5} (\gamma + 4\delta)\right] & 
\gamma(1 - \gamma) \\
\delta(1 - \delta) & 
\gamma\delta + (1 - \gamma - \delta) \left[\beta + \frac{\alpha}{5} (\gamma + 4\delta)\right] & 
(1 - \gamma)(1 - \delta) - (1 - \gamma - \delta) \left[\beta + \alpha \left(1 - \frac{\gamma}{5} - \frac{4\delta}{5}\right)  \right] &
\gamma(1 - \gamma) \\
\delta^2 & 
\delta(1 - \gamma) + \frac{\alpha}{5} (1 - \gamma - \delta)(1 - \gamma - 4\delta) & 
\delta(1 - \gamma) + \frac{\alpha}{5} (1 - \gamma - \delta)(1 - \gamma - 4\delta) & 
(1 - \gamma)^2
\end{pmatrix},
\end{align}
\end{widetext}
where $\gamma$ is the nonunitarity of the noise channel and $\delta$ the nonunitality.
This transfer matrix is obtained simply by composing the two-qubit transfer matrix in \cref{eq_twoqubit_transfermat} with the transfer matrix of the local noise channels \cref{eq_nonunital} applied to each qubit.

We observe that the appropriate part of the proof in \cref{lem_singlequbit_bound,lem_multiqubit_bound} is the lower bound on the probability that an $\mathtt{S}$ operator survives.
For this we consider
\begin{widetext}
\begin{align}
 &\Pr_{\mathcal{X}_t} [\bm{x}^t_{j} = 1| \bm{x}^{t-1}_{j} = 1 \text{ or } \bm{x}^{t-1}_{j'} = 1] \geq 
  \min \left\{\gamma\delta + (1 - \gamma - \delta) \left[\beta + \frac{\alpha}{5} (\gamma + 4\delta)\right] + \delta(1 - \gamma) + \frac{\alpha}{5} (1 - \gamma - \delta)(1 - \gamma - 4\delta), \right. \nonumber \\
& \left. (1 - \gamma)(1 - \delta) - (1 - \gamma - \delta) \left[\beta + \alpha \left(1 - \frac{\gamma}{5} - \frac{4\delta}{5}\right)  \right] + \delta(1 - \gamma) + \frac{\alpha}{5} (1 - \gamma - \delta)(1 - \gamma - 4\delta), (1-\gamma)^2  + \gamma(1-\gamma) \right\}. \nonumber \\
 = &  \min \left\{\delta + (1 - \gamma - \delta)\left(\beta + \frac{\alpha}{5}\right), 1-\gamma -(1-\gamma-\delta)\left( \beta + \frac{4\alpha}{5} \right), 1-\gamma \right\} \\
 \geq  & \frac{\alpha}{5}(1-\gamma -\delta) + \delta.
\end{align}
\end{widetext}
Define this last quantity to be $\alpha'/5$, so that $\alpha' = \alpha (1-\gamma-\delta) + 5 \delta \geq \alpha(1-2\gamma)$.
Note also that $\alpha' \leq 1-2\gamma + 5\delta \leq 1+3\gamma \leq \frac{5}{2}$.
We can derive an analogue of \cref{lem_singlequbit_bound} by replacing $\alpha \to \alpha'$.
As for the analogue of \cref{lem_multiqubit_bound}, we check in each of the six cases whether the bound $\Pr_{\mathcal{X}_{s+1}}[\bm{x}^{s+1}_{jj'} | \bm{x}^{s}_{jj'}] \geq \left( \frac{\alpha'}{5} \right)^2$ holds, which would suffice for the proof.

\textbf{cases i)--ii)}. Here, the conditional probability is 
\begin{align}
\Pr_{\mathcal{X}_{s+1}}[\bm{x}^{s+1}_{jj'} = 11 | \bm{x}^{s}_{jj'} = 11] = (1-\gamma)^2 \geq \left(\frac{\alpha'}{5} \right)^2
\end{align}
since $(1-\gamma)^2 \geq \frac{1}{4}$ and $\frac{\alpha'}{5}  \leq \frac{1}{2}$.

\textbf{case iii)}. We have 
\begin{align}
\Pr_{\mathcal{X}_{s+1}}[\bm{x}^{s+1}_{j'} = 1 | \bm{x}^{s}_{jj'} = 11] &= (1-\gamma)^2 + \gamma(1-\gamma) \nonumber \\
&= 1-\gamma \geq \left( \frac{\alpha'}{5} \right)^2.
\end{align}
Here, unlike before, the conditional probability for this case is larger (since there is also a possibility for the event $\mathtt{SS}_{jj'} \to \mathtt{SI}_{jj'}$ to occur), but the lower bound we assign remains the same.

\textbf{case iv)}. The input for this case is now $\mathtt{SI}_{jj'}$. We have 
\begin{align}
\Pr_{\mathcal{X}_{s+1}}[\bm{x}^{s+1}_{jj'} = 11 | \bm{x}^{s}_{jj'} = 10] = \delta(1 - \gamma) + \nonumber \\
\frac{\alpha}{5} (1 - \gamma - \delta)(1 - \gamma - 4\delta).
\end{align}
To prove this is $\geq \left( \frac{\alpha'}{5} \right)^2$, consider their difference, keeping in mind that $\alpha' = \alpha (1-\gamma-\delta) + 5 \delta$.
\begin{gather}
 \delta(1 - \gamma) + \frac{\alpha}{5}  (1 - \gamma - \delta)(1 - \gamma - 4\delta) - \left( \frac{\alpha'}{5} \right)^2 \nonumber \\
 = (1-\gamma-\delta) \left[ (1-\gamma)\left( \frac{\alpha}{5}\right) \left(1-\frac{\alpha}{5}\right) + \right. \nonumber \\
 \left. \delta \left( \left( \frac{\alpha}{5}\right)^2 - 6 \left( \frac{\alpha}{5}\right) + 1 \right) \right].
\end{gather}
Now, we observe that since $\delta \leq \gamma \leq 1/2$, the first term $1-\gamma -\delta \geq 0$. Next, since $0\leq \alpha \leq 1$, the term $(1-\gamma)\left( \frac{\alpha}{5}\right) \left(1-\frac{\alpha}{5}\right)$ is also nonnegative.
Lastly, so is $\delta \left( \left( \frac{\alpha}{5}\right)^2 - 6 \left( \frac{\alpha}{5}\right) + 1 \right)  = \delta \left(5- \frac{\alpha}{5} \right)\left(1-\frac{\alpha}{5}\right) \geq 0$.
Therefore, $\Pr_{\mathcal{X}_{s+1}}[\bm{x}^{s+1}_{jj'} = 11 | \bm{x}^{s}_{jj'} = 10] \geq \left( \frac{\alpha'}{5} \right)^2$. \newline

\textbf{cases v)--vi)}. Here we are interested in the conditional probabilities of the events $\Pr_{\mathcal{X}_{s+1}}[\bm{x}^{s+1}_{j} = 1 | \bm{x}^{s}_{jj'} = 10]$ and $\Pr_{\mathcal{X}_{s+1}}[\bm{x}^{s+1}_{j'} = 1 | \bm{x}^{s}_{jj'} = 10]$.
Observe that both these conditional probabilities are lower bounded by that in case iv), namely $\Pr_{\mathcal{X}_{s+1}}[\bm{x}^{s+1}_{jj'} = 11 | \bm{x}^{s}_{jj'} = 10]$.
Therefore, the last two cases follow.

In sum, we have proved that we can replace $\frac{\alpha}{5} \to \frac{\alpha}{5}(1-\gamma-\delta) + \delta$ in the proof of \cref{lem_multiqubit_bound}, meaning we can consider each SWAP-active sequence separately.
Following the steps in the proof of \cref{apx_thm:noiseless_lower_bound}, this yields the modified lower bound 
\begin{align}
 \Var{L}_{\btheta \sim \mathcal{D}_p} \geq \sum_{\bm{\alpha}} c_{\bm{\alpha}}^2 \left( \frac{\sin^2 \varphi}{3} \right)^{\abs{\bm{\alpha}}}  \left(\frac{\alpha}{5}(1-\gamma-\delta) + \delta \right)^{kf},
\end{align}
thus completing the proof.
\end{proof}

This theorem can be compared with that of Ref.~\cite{mele2024noiseinducedshallowcircuitsabsence}, which also proves a lower bound on the variance of a local cost function in the presence of nonunital noise. 
They do not need resets to achieve a nontrivial lower bound.
From our expression of the lower bound, we see that the degree of nonunitality of the noise channel, which is governed by $\delta$, favors a larger lower bound: the matrix elements of the transfer matrix that favor high $\mathtt{S}$ mass are monotonically increasing in $\delta$.
We note that the settings in these works are incomparable: in our setting, reset channels may be applied at will at specific places in the circuit, whereas in the setting of Ref.~\cite{mele2024noiseinducedshallowcircuitsabsence}, the noise channels always occur after every gate.


\section{Variational Quantum Imaginary Time Evolution}\label{app:VarQITE}

Variational quantum imaginary time evolution (VarQITE) aims at approximating a time dependent state of the form
\begin{align}
    \rho(t) = \frac{e^{-Ht}\rho(0)e^{-Ht}}{\text{Tr}\left[e^{-2Ht}\rho(0)\right]},
\end{align}
which may equivalently be defined via the differential equation
\begin{align}
    \frac{\partial \rho(t)}{\partial t} = -\Big(\left\{H, \rho(t)\right\} - 2 \text{Tr}\left[H\rho(t)\right]\rho(t)\Big),
\end{align}
with a parameterized, and as such controllable, ansatz state $\tilde\rho(\btheta_t)$, for $\rho(0)=\tilde\rho\left(\btheta_0\right)$.
The differential equation describing the time evolution may now be mapped onto the parameterized state as
\begin{align}
        \sum_i\frac{\partial \tilde\rho(\btheta_t)}{\partial\theta_t^i} \frac{\partial \theta_t^i}{\partial t} = -\Big(\left\{H, \tilde\rho(\btheta_t)\right\} - 2 \text{Tr}\left[H\tilde\rho(\btheta_t)\right]\tilde\rho(\btheta_t)\Big).
        \label{eq:McLachlanDensity}
\end{align}
McLachlan's variational principle \cite{mclachlan1964variationalsolutiontimedependentschrodinger} aims to find parameter updates that minimize the distance of the left and the right hand sides of Eq.~\eqref{eq:McLachlanDensity}:
\begin{align}
    \delta \norm{\sum_i\frac{\partial \tilde\rho(\btheta_t)}{\partial\theta_t^i} \frac{\partial \theta_t^i}{\partial t} + \left\{H, \tilde\rho(\btheta_t)\right\} - 2 \text{Tr}\left[H\tilde\rho(\btheta_t)\right]\tilde\rho(\btheta_t)} = 0.
\end{align}
Further, solving the variational principle for $\frac{\partial\theta_t^i}{\partial t}$ (see, e.g., Ref.~\cite{yuan2019theoryvariationalquantumsimulation}) results in the following system of linear equations which describes the propagation of the parameters according to the time evolution
\begin{align}
\label{eq:McLachlanSLE}
    \sum\limits_j M_{i,j} \frac{\partial\theta_t^j}{\partial t} = Y_i,
\end{align}
for
\begin{align}
    M_{i,j} = \text{Tr}\Big[ \frac{\partial \tilde\rho(\btheta_t)}{\partial\theta_t^i}^{\dagger}  \frac{\partial \tilde\rho(\btheta_t)}{\partial\theta_t^j}\Big],
\end{align}
and
\begin{align}
    Y_i = -\text{Tr}\left[ \frac{\partial \tilde\rho(\btheta_t)}{\partial\theta_t^i}^{\dagger}  \Big(\left\{H, \tilde\rho(\btheta_t)\right\} - 2 \text{Tr}\left[H\tilde\rho(\btheta_t)\right]\tilde\rho(\btheta_t)\Big)\right].
\end{align}
\Cref{eq:McLachlanSLE}, in turn, defines an initial value problem that may be approached with an arbitrary ordinary differential equation (ODE) solver \cite{coddington1984theoryordinarydifferentialequations} such as Forward Euler or Runge Kutta.
Typical error sources underlying this approach are the variational approximation of the accessible Hilbert space, the time discretization of the differential equation, and integration errors underlying the ODE solver.
A notable advantage of this method, however, is the possibility to efficiently evaluate a-posteriori error bounds in terms of the distance between the target state and the prepared state \cite{martinazzo2020localintimeerrorvariationalquantum, zoufal2023errorboundsvariationalquantum}.

%

\end{document}